\documentclass[11pt,letterpaper]{article}
\pdfoutput=1

\usepackage{microtype}%if unwanted, comment out or use option "draft"

\usepackage[T1]{fontenc}
\usepackage{amsthm}
\usepackage{amssymb}
\usepackage[english]{babel}
\usepackage[utf8]{inputenc}
\usepackage{amsmath}
\usepackage{amsfonts}
\usepackage{graphicx}
\usepackage[colorinlistoftodos]{todonotes}
\usepackage{mathtools}
\usepackage{enumitem}
\usepackage{thmtools}
\usepackage{fullpage}
\usepackage{booktabs}
\usepackage{bbm}
\usepackage[normalem]{ulem} % m.in. przekreślenie tekstu
\usepackage{hyperref}

\newcommand{\N}{\mathbb{N}}
\newcommand{\R}{\mathbb{R}}
\newcommand{\Z}{\mathbb{Z}}
\newcommand{\Pc}{\mathcal{P}}
\newcommand{\Gc}{\mathcal{G}}
\newcommand{\Ac}{\mathcal{A}}
\newcommand{\A}{\mathbf{A}}
\newcommand{\C}{\mathbf{C}}
\newcommand{\E}{\mathbb{E}}
\newcommand{\whp}{\emph{w.h.p.}}
\newcommand{\wep}{\emph{w.e.p.}}
\newcommand{\Prob}{\mathbb{P}}
\newcommand\eps{\varepsilon}
\newcommand{\bigo}{\mathcal{O}}
\newcommand{\vol}{\textrm{vol}}
\newcommand{\ABCD}{\textbf{ABCD}}

\theoremstyle{plain}
\newtheorem{definition}{Definition}[section]
\newtheorem{lemma}[definition]{Lemma}
\newtheorem{theorem}[definition]{Theorem}

\newtheorem{corollary}[definition]{Corollary}

\title{Modularity of the ABCD Random Graph Model \\ with Community Structure}

\author{
Bogumi\l{} Kami\'nski\thanks{Decision Analysis and Support Unit, SGH Warsaw School of Economics, Warsaw, Poland; e-mail: \texttt{bogumil.kaminski@sgh.waw.pl}}
\and
Bartosz Pankratz\thanks{Department of Mathematics, Ryerson University, Toronto, ON, Canada; e-mail: \texttt{bartosz.pankratz@ryerson.ca}}
\and
Pawe\l{}~Pra\l{}at\thanks{Department of Mathematics, Ryerson University, Toronto, ON, Canada; e-mail: \texttt{pralat@ryerson.ca}}
\and
Fran\c{c}ois Th\'eberge\thanks{Tutte Institute for Mathematics and Computing, Ottawa, ON, Canada; email: \texttt{theberge@ieee.org}}
}

\begin{document}

\maketitle

\begin{abstract}
The \textbf{A}rtificial \textbf{B}enchmark for \textbf{C}ommunity \textbf{D}etection (\ABCD) graph is a random graph model with community structure and power-law distribution for both degrees and community sizes. The model generates graphs with similar properties as the well-known \textbf{LFR} one, and its main parameter $\xi$ can be tuned to mimic its counterpart in the \textbf{LFR} model, the mixing parameter $\mu$. 

In this paper, we investigate various theoretical asymptotic properties of the \ABCD\ model. In particular, we analyze the modularity function, arguably, the most important graph property of networks in the context of community detection. Indeed, the modularity function is often used to measure the presence of community structure in networks. It is also used as a quality function in many community detection algorithms, including the widely used \emph{Louvain} algorithm.
\end{abstract}

%%%%%%%%%%%%%%%%%%%%%%%%%%%%%%%%%%%%%%%%%%%%%%%%%%%%%%%%%%%
\section{Introduction}
%%%%%%%%%%%%%%%%%%%%%%%%%%%%%%%%%%%%%%%%%%%%%%%%%%%%%%%%%%%

One of the most important features of real-world networks is their community structure, as it reveals the internal organization of nodes~\cite{fortunato2010community}. In social networks communities may represent groups by interest, in citation networks they correspond to related papers, in the Web communities are formed by pages on related topics, etc. Being able to identify communities in a network could help us to exploit this network more effectively. 

\medskip

Unfortunately, there are very few datasets with ground-truth identified and labelled. As a result, there is need for synthetic random graph models with community structure that resemble real-world networks in order to benchmark and tune clustering algorithms that are unsupervised by nature. The \textbf{LFR} (Lancichinetti, Fortunato, Radicchi) model~\cite{lancichinetti2008benchmark,lancichinetti2009benchmarks} generates networks with communities and at the same time it allows for the heterogeneity in the distributions of both node degrees and of community sizes. It became a standard and extensively used method for generating artificial networks. 

In this paper, we analyze the Artificial Benchmark for Community Detection (\textbf{ABCD} graph)~\cite{kaminski2021artificial} that was recently introduced and implemented\footnote{\url{https://github.com/bkamins/ABCDGraphGenerator.jl/}}, including a fast implementation that uses multiple threads (\textbf{ABCDe})\footnote{\url{https://github.com/tolcz/ABCDeGraphGenerator.jl/}}. Undirected variant of \textbf{LFR} and \textbf{ABCD} produce graphs with comparable properties but \textbf{ABCD}/\textbf{ABCDe} is faster than \textbf{LFR} and can be easily tuned to allow the user to make a smooth transition between the two extremes: pure (disjoint) communities and random graph with no community structure. More importantly from the perspective of this paper, it is easier to analyze theoretically. See Subsection~\ref{sec:def_ABCD} for a definition of the model. 

\medskip

The key ingredient for many clustering algorithms is \emph{modularity}, which is at the same time a global criterion to define communities, a quality function of community detection algorithms, and a way to measure the presence of community structure in a network. The definition of modularity for graphs was first introduced by Newman and Girvan in~\cite{newman2004finding}. We present the definition in Subsection~\ref{sec:def_modularity}.

Despite some known issues with this function such as the ``resolution limit'' reported in~\cite{fortunato2007resolution}, many popular algorithms for partitioning nodes of large graphs use it~\cite{clauset2004finding,newman2004fast,lancichinetti2011limits} and perform very well. The list includes one of the mostly used unsupervised algorithms for detecting communities in graphs, the \emph{Louvain} (hierarchical) algorithm~\cite{blondel2008fast}. For more details we direct the reader to any book on complex networks, including the following recent additions~\cite{kaminski2021mining,lambiotte2021modularity}. 

%%%%%%%%%%%%%%%%%%%%%%%%%%%%%%%%%%%%%%%%%%%%%%%%%%%%%%%%%%%
\subsection{Summary of Results} 

In this paper, we investigate the modularity function for the \ABCD\ model $\Ac$. The paper is structured as follows. The \ABCD\ model is introduced in Subsection~\ref{sec:def_ABCD} and the modularity function is defined in Subsection~\ref{sec:def_modularity}. Results for other random graph model in the context of the modularity function are summarized in Section~\ref{sec:related_results}. Some useful general observations are made in Section~\ref{sec:preliminaries}: concentration results that are used in almost all of our proofs are outlined in Subsection~\ref{sec:Chernoff}, useful expansion properties of \textbf{random $d$-regular graphs} are presented in Subsection~\ref{sec:expanders}.

We start analyzing the \ABCD\ model by investigating some basic properties---see Section~\ref{sec:basic_properties}. These properties will be needed to establish results for the modularity function but they are important on their own. Lemma~\ref{lem:degree_distribution} shows that the degree distribution is well concentrated around the corresponding expectations. Corollary~\ref{cor:community_sizes} shows a concentration for the number of communities and well as the distribution of their sizes. The same generating process is applied in \textbf{LFR} so the two results hold for that model as well. The \ABCD\ model assigns nodes to communities randomly. Clearly, there is no hope to predict the volumes of small communities of constant size but sufficiently large communities have their volumes as well as the number of internal edges well concentrated around the corresponding expectations---see Lemma~\ref{lem:volume_of_communities}.

Then we move to the results for the modularity function. By design of the \ABCD\ model, $1-\xi$ fraction of edges should become community edges and so should end up in some part of the ground truth partition $\C$. ($\xi$ is the main parameter of the model responsible for the level of noice.) It is indeed the case but it turns out that a negligible fraction of the background graph join them there. As a result, the modularity function of the ground-truth partition $\C$ is asymptotic to $1-\xi$, as proved in Theorem~\ref{thm:ground-truth}. 

Analyzing the maximum modularity is much more complex. We have two types of results. The first result (Theorem~\ref{thm:large_level_of_noise}) shows that when the level of noise is sufficiently large ($\xi$ close to one), then the maximum modularity $q^*(\Ac)$ is asymptotically larger than $q(\C)$, the modularity of the ground-truth. In this regime, the number of edges within community graphs $G_i$ is relatively small so a partition of the background graph into small connected pieces yields a better modularity function. To show this result, we need to investigate the degree distribution of the background graph (Lemma~\ref{lem:degree_distribution_of_G0}) which might be of independent interest. 

The second set of results is concerned with graphs with low level of noise ($\xi$ close to zero). For these graphs, the situation is quite opposite. It turns out that the ground truth partition is asymptotically the best possible, that is, the maximum modularity $q^*(\Ac)$ is only $o(1)$ away from $q(\C)$, the modularity of the ground truth partition $\C$; both of them are asymptotic to $1-\xi$ (see Theorem~\ref{thm:modularity_small_xi}). For some technical reason, it is assumed that $\delta$, the minimum degree of $\Ac$, is sufficiently large: the lower bound of $100$ easily works but it may be improved with more detailed treatment. Having said that, it seems that one needs a different approach to uncover the real bottleneck. On the other hand, the above property is not true if $\delta=1$ (see Theorem~\ref{thm:modularity_delta1}): if $\delta=1$, then $q^*(\Ac)$ is substantially larger than $q(\C)$, regardless of how close to zero $\xi$ is.

Finally, let us mention that the approach used to prove Theorem~\ref{thm:modularity_small_xi} utilize a coupling of random graph $\Pc(\textbf{w})$ on $n'$ nodes and a given degree sequence with a random $b$-regular graph $\mathcal{P}_{n'',b}$ on $n''$ nodes. We use this coupling to show that good expansion properties of $\mathcal{P}_{n'',b}$ imply good expansion properties for $\Pc(\textbf{w})$. This coupling seems powerful and might be potentially useful for some other applications. 

%%%%%%%%%%%%%%%%%%%%%%%%%%%%%%%%%%%%%%%%%%%%%%%%%%%%%%%%%%%
\subsection{Simulations}

This paper focuses on asymptotic theoretical results of the \ABCD\ model. Having said that, we performed a number of simulations and compared asymptotic predictions with graphs generated by computer. These simulations show that the behaviour of small random instances is similar to what is predicted by the theory. This is a good news for practitioners as it shows that, despite the fact that the generative algorithm is randomized, the model has good stability. The code is accessible on GitHub repository\footnote{\url{https://github.com/tolcz/ABCDeGraphGenerator.jl/}}.

%%%%%%%%%%%%%%%%%%%%%%%%%%%%%%%%%%%%%%%%%%%%%%%%%%%%%%%%%%%
\subsection{Open Problems}

Theoretical results and simulations suggest that if $\delta$, the minimum degree of $\Ac$, satisfies $\delta \ge \delta_0$ for some $\delta_0 \ge 2$, then there exists a constant $\xi_0 = \xi_0(\delta)$ (that possibly depends also on other parameters of the \ABCD\ model $\Ac$) such that the following holds \whp\ (that is, with probability tending to one as $n \to \infty$): 
\begin{itemize}
\item if $0 < \xi < \xi_0$, then $q^*(\Ac) \sim q(\C)$, where $\C$ is the ground truth partition of the set of nodes of $\Ac$,
\item if $\xi > \xi_0$, then $q^*(\Ac)$ is separated by a constant from $q(\C)$.
\end{itemize}
Our results make the first step towards this conjecture by showing upper and lower bounds for such threshold constant $\xi_0$, when $\delta_0=100$. The bounds for $\xi_0$ are not close to each other. The next step would be to narrow the gap down or perhaps to determine the threshold value exactly, provided that $\delta_0$ is sufficiently large. Another natural direction would be to decrease the lower bound for $\delta$, that is, to decrease the value of $\delta_0$. We showed that $\delta=1$ does not have the desired property but maybe $\delta_0=2$? Or maybe one can always construct a better partition than $\C$ when $\delta = 2$, regardless how small parameter $\xi$ is? These questions are left as open questions for future investigation.

%%%%%%%%%%%%%%%%%%%%%%%%%%%%%%%%%%%%%%%%%%%%%%%%%%%%%%%%%%%
\section{Definitions (of ABCD Model and Modularity)}
%%%%%%%%%%%%%%%%%%%%%%%%%%%%%%%%%%%%%%%%%%%%%%%%%%%%%%%%%%%

%%%%%%%%%%%%%%%%%%%%%%%%%%%%%%%%%%%%%%%%%%%%%%%%%%%%%%%%%%%
\subsection{Asymptotic Notation}

Our results are asymptotic in nature, that is, we will assume that the number of nodes $n\to\infty$. Formally, we consider a sequence of graphs $G_n=(V_n,E_n)$ and we are interested in events that hold \emph{with high probability} (\whp), that is, events that hold with probability tending to 1 as $n\to \infty$. It would be also convenient to consider events that hold \emph{with extreme probability} (\wep), that is, events that hold with probability at least $1-\exp(-\Omega( (\log n)^2 ))$. An easy but convenient property is that if a polynomial number of events hold \wep, then \wep\ all of them hold simultaneously.

Given two functions $f=f(n)$ and $g=g(n)$, we will write $f(n)=\bigo(g(n))$ if there exists an absolute constant $c \in \R_+$ such that $|f(n)| \leq c|g(n)|$ for all $n$, $f(n)=\Omega(g(n))$ if $g(n)=\bigo(f(n))$, $f(n)=\Theta(g(n))$ if $f(n)=\bigo(g(n))$ and $f(n)=\Omega(g(n))$, and we write $f(n)=o(g(n))$ or $f(n) \ll g(n)$ if $\lim_{n\to\infty} f(n)/g(n)=0$. In addition, we write $f(n) \gg g(n)$ if $g(n)=o(f(n))$ and we write $f(n) \sim g(n)$ if $f(n)=(1+o(1))g(n)$, that is, $\lim_{n\to\infty} f(n)/g(n)=1$.

%%%%%%%%%%%%%%%%%%%%%%%%%%%%%%%%%%%%%%%%%%%%%%%%%%%%%%%%%%%
\subsection{Other Notation}

We will use $\log n$ to denote a natural logarithm of $n$. For a given $n \in \N := \{1, 2, \ldots \}$, we will use $[n]$ to denote the set consisting of the first $n$ natural numbers, that is, $[n] := \{1, 2, \ldots, n\}$. Finally, as typical in the field of random graphs, for expressions that clearly have to be an integer, we round up or down but do not specify which: the choice of which does not affect the argument.

%%%%%%%%%%%%%%%%%%%%%%%%%%%%%%%%%%%%%%%%%%%%%%%%%%%%%%%%%%%
\subsection{ABCD Model}\label{sec:def_ABCD}

\begin{table}[htp]
\caption{Parameters of the \ABCD\ model}
\begin{center}
\begin{tabular}{|l|l|l|}
\hline
parameter & range & description \\
\hline
$n$ & $\N$ & number of nodes \\
$\gamma$ & $(2,3)$ & power-law exponent of degree distribution \\
$\delta$ & $\N$ & minimum degree at least $\delta$ \\
$\zeta$ & $(0,\frac{1}{\gamma-1}]$ & maximum degree at most $n^{\zeta}$ \\
$\beta$ & $(1,2)$ & power-law exponent of distribution of community sizes \\
$s$ & $\N \setminus [\delta]$ & community sizes at least $s$ \\
$\tau$ & $(\zeta,1)$ & community sizes at most $n^{\tau}$ \\
$\xi$ & $(0,1)$ & level of noise \\
\hline
\end{tabular}
\end{center}
\label{tab:parameters}
\end{table}

The \ABCD\ model is governed by 8 parameters summarized in Table~\ref{tab:parameters}.
For a fixed set of parameters, we generate the \ABCD\ graph $\Ac=\Ac(n, \gamma, \delta, \zeta, \beta, s, \tau, \xi)$ following the steps outlined below. Each time we refer to graph $\Ac$ in this paper, we implicitly (or explicitly, but it happens rather rarely) fix all of these parameters.

%%%%%%%%%%%%%%%%%%%%%%%%%%%%%%%%%%%%%%%%%%%%%%%%%%%%%%%%%%%
\subsubsection{Degree Distribution}

Let $\gamma \in (2,3)$, $\delta \in \N$, and $\zeta \in (0,1)$. Degrees of nodes of \ABCD\ graph $\Ac$ are generated randomly following the (truncated) \emph{power-law distribution} $\Pc(\gamma, \delta, \zeta)$ with exponent $\gamma$, minimum value $\delta$, and maximum value $D = n^{\zeta}$. In order to make sure the sum of degrees is even, if needed, we decrease by one the degree of one node of the largest degree. 

It is easy to show (see Lemma~\ref{lem:max_degree}) that for any $\omega=\omega(n)$ tending to infinity as $n \to \infty$ \whp\ the maximum degree of $\Ac$ is at most $n^{1/(\gamma-1)} \omega$ (of course, by definition, it is deterministically at most $n^{\zeta}$). As a result, for any two values of $\zeta_1, \zeta_2 \in ( \frac{1}{\gamma-1},1)$ one may couple the two corresponding ABCD graphs $\Ac$ so that \whp\ they produce exactly the same graph. Hence, for convenience but without loss of generality, we will later on assume that $\zeta \in (0,\frac {1}{\gamma-1}]$.

%%%%%%%%%%%%%%%%%%%%%%%%%%%%%%%%%%%%%%%%%%%%%%%%%%%%%%%%%%%
\subsubsection{Distribution of Community Sizes}

Let $\beta \in (1,2)$, $s \in \N \setminus [\delta]$, and $\tau \in (\zeta,1)$. Community sizes of \ABCD\ graph $\Ac$ are generated randomly following the (truncated) \emph{power-law distribution} $\Pc(\beta, s, \tau)$ with exponent $\beta$, minimum value $s$, and maximum value $S = n^{\tau}$. Communities are generated with this distribution as long as the sum of their sizes is less than $n$, the desired number of nodes. Suppose that the last community has size $z$ and after adding it to the remaining ones, the sum of their sizes will exceed $n$ by $k \in \N \cup \{0\}$. If $k=0$, then there is nothing else to do. If $z - k \ge s$, then the size of the last community is reduced to $z-k$ so that the total number of nodes is exactly $n$. Otherwise, we select $z-k < s$ old communities at random, increase their sizes by one, and remove the last community so that the desired property holds. 

The assumption that $\tau > \zeta$ is introduced to make sure large degree nodes have large enough communities to be assigned to. Similarly, the assumption that $s \ge \delta+1$ is required to guarantee that small communities are not too small and so that they can accommodate small degree nodes. 

%%%%%%%%%%%%%%%%%%%%%%%%%%%%%%%%%%%%%%%%%%%%%%%%%%%%%%%%%%%
\subsubsection{Assigning Nodes into Communities}

At this point, the degree distribution $(w_1 \ge w_2 \ge \ldots \ge w_n)$ and the distribution of community sizes $(c_1 \ge c_2 \ge \ldots \ge c_\ell)$ are already fixed. The final \ABCD\ graph $\Ac$ will be formed as the union of $\ell+1$ independent graphs: $\ell$ community graphs $G_i=(C_i, E_i)$, $i \in [\ell]$, and a single background graph $G_0=(V, E_0)$, where $V = \bigcup_{i \in [\ell]} C_i$. Roughly $\xi w_i$ edges incident to node $i$ will, by definition, belong to its own community but a few additional edges from the background graph might end up in that community. In order to create enough room for these edges, node of degree $w_i$ will be allowed to be assigned to a community of size $c_j$ if the following inequality is satisfied:
$$
\lceil (1-\xi\phi) w_i \rceil \le c_j - 1, \hspace{1cm} \text{ where } \phi = 1 - \sum_{k \in [\ell]} (c_k/n)^2.
$$
Note that this condition is equivalent to the following one:
\begin{equation}\label{eq:admissible}
w_i \le \frac {c_j-1}{1-\xi\phi}.
\end{equation}
An assignment of nodes into communities will be called \emph{admissible} if the above inequality is satisfied for all nodes. We will show in Subsection~\ref{sec:assigning} that there are many admissible assignments. In particular, there are linearly many nodes of degree $\delta$ but, fortunately, \whp\ communities of size more than $n^{\zeta}$ (more than the maximum degree) have space for almost all nodes. We select one admissible assignment uniformly at random. Sampling uniformly one of such assignments turns out to be relatively easy from both theoretical and practical points of view. We will discuss it in detail in Subsection~\ref{sec:assigning}.

%%%%%%%%%%%%%%%%%%%%%%%%%%%%%%%%%%%%%%%%%%%%%%%%%%%%%%%%%%%
\subsubsection{Distribution of Weights}

Parameter $\xi \in (0,1)$ reflects the amount of noise in the network. It controls the fraction of edges that are between communities. Indeed, asymptotically (but not exactly) $1-\xi$ fraction of edges are going to end up within one of the communities. Each node will have its degree $w_i$ split into two parts: \emph{community degree} $y_i$ and \emph{background degree} $z_i$ ($w_i=y_i+z_i$). Our goal is to get $y_i \approx (1-\xi) w_i$ and $z_i \approx \xi w_i$. However, both $y_i$ and $z_i$ have to be non-negative integers and for each community $C \subseteq V$, $\sum_{i \in C} y_i$ has to be even. Note that since $\sum_{i \in V} w_i$ is even, it will imply that 
$$
\sum_{i \in V} z_i = \sum_{i \in V} (w_i - y_i) = \sum_{i \in V} w_i - \sum_C \sum_{i \in C} y_i
$$ 
is even too. 

For each community $C \subseteq V$ we identify the \emph{leader}, a node of the largest degree $w_i$ associated with community $C$. (If many nodes in $C$ have the largest degree, then we arbitrarily select one of them to be the leader.) For non-leaders we split the weights as follows:
$$
y_i = \Big\lfloor (1-\xi) w_i \Big\rceil \hspace{1cm} \text{ and } \hspace{1cm} z_i = w_i - y_i,
$$
where for a given integer $a \in \Z$ and real number $b \in [0, 1)$ the random variable $\lfloor a+b \rceil$ is defined as 
\begin{equation}\label{eq:rounding}
\lfloor a+b \rceil = 
\begin{cases}
a & \text{ with probability } 1-b \\
a+1 & \text{ with probability } b.
\end{cases}
\end{equation}
(Note that $\E [\lfloor a+b \rceil] = a(1-b) + (a+1)b = a + b$.) For the leader of community $C$ we round $(1-\xi) w_i$ up or down so that the sum of weights in each cluster is even. If $(1-\xi) w_i \in \N$ and the sum of weights $y_i$ in $C$ is odd, then we randomly make a decision whether subtract or add one to make the sum to be even. 

%%%%%%%%%%%%%%%%%%%%%%%%%%%%%%%%%%%%%%%%%%%%%%%%%%%%%%%%%%%
\subsubsection{Creating Graphs}\label{sec:creating_graphs}

As already mentioned, the final \ABCD\ graph $\Ac = (V,E)$ will be formed as the union of $\ell+1$ independent graphs: $\ell$ community graphs $G_i=(C_i, E_i)$, $i \in [\ell]$, and a single background graph $G_0=(V, E_0)$, where $V = \bigcup_{i \in [\ell]} C_i$, that is, $E = \bigcup_{i \in [\ell] \cup \{0\}} E_i$. Each of these $\ell+1$ graphs will be created independently. The partition $\C = \{C_1, C_2, \ldots, C_{\ell}\}$ will be called a \emph{ground-truth} partition.

Suppose then that our goal is to create a graph on $n$ nodes with a given degree distribution $\textbf{w} := (w_1, w_2, \ldots, w_n)$, where $\textbf{w}$ is any vector of non-negative integers such that $w := \sum_{i \in [n]} w_i$ is even. We define a random multi-graph $\Pc(\textbf{w})$ with a given degree sequence known as the \textbf{configuration model} (sometimes called the \textbf{pairing model}), which was first introduced by Bollob\'as~\cite{bollobas1980probabilistic}. (See~\cite{bender1978asymptotic,wormald1984generating,wormald1999models} for related models and results.)
We start with $w$ \emph{points} that are partitioned into $n$ \emph{buckets} labelled with labels $v_1, v_2, \ldots, v_n$; bucket $v_i$ consists of $w_i$ points. It is easy to see that there are $\frac {w!}{(w/2)! 2^w}$ pairings of points. We select one of such pairings uniformly at random, and construct a multi-graph $\Pc(\textbf{w})$, with loops and parallel edges allowed, as follows: nodes are the buckets $v_1, v_2, \ldots, v_n$, and a pair of points $xy$ corresponds to an edge $v_iv_j$ in $\Pc(\textbf{w})$ if $x$ and $y$ are contained in the buckets $v_i$ and $v_j$, respectively.

\subsubsection{Simulation Corner}

Note that the \ABCD\ model $\Ac$ allows loops and multiple edges. Indeed, they can occur both in any of the generated graphs $G_i$ ($i \in [\ell] \cup \{0\}$) or after taking a union of their edge sets. In general, however, there will not be very many of them. To keep the theoretical model simple, in this paper we allow $\Ac$ be a multi-graph but, alternatively, one may condition on $\Ac$ to be a simple graph. In practice, the algorithm performs some kind of edges ``switching'' that is known to generate a random graph that is very close to the uniform distribution~\cite{janson2020random}. We will use simple graphs for our experiments to show that, indeed, the difference is not detectable. Moreover, to get closed formulas in some theoretical results proved in this paper we use the continuous variant of the power law distribution. For example, if $X \in \Pc(\gamma, \delta, \zeta)$, then for any $k \in \{ \delta, \delta+1, \ldots, D\}$ we assume that
$$
q_k = \Pr( X = k ) = \frac { \int_{k}^{k+1} x^{-\gamma} dx }{ \int_{\delta}^{D+1} x^{-\gamma} dx }.
$$
Alternatively, one may use the discrete counterpart, namely, assume that 
\begin{equation}\label{eq:power-law-discrete}
r_k = \Pr( X = k ) = \frac { k^{-\gamma} }{ \sum_{x = \delta}^{D} x^{-\gamma} }.
\end{equation}
All results proved in this paper hold for both variants. We state them for continuous one but one may replace $q_k$ with $r_k$ to get the discrete counterparts. Since the default implementation of the \ABCD\ model uses the discrete distribution, we use it for our simulations.
For more details we direct the reader to the original paper on the \ABCD\ model~\cite{kaminski2021artificial}.

%%%%%%%%%%%%%%%%%%%%%%%%%%%%%%%%%%%%%%%%%%%%%%%%%%%%%%%%%%%
\subsection{Modularity Function}\label{sec:def_modularity}

The modularity function favours partitions of the set of nodes of a graph $G$ in which a large proportion of the edges fall entirely within the parts  but benchmarks it against the expected number of edges one would see in those parts in the corresponding \textbf{Chung-Lu} random graph model~\cite{chung2006complex} which generates graphs with the expected degree sequence following exactly the degree sequence in $G$. 

Formally, for a graph $G=(V,E)$ and a given partition $\A = \{A_1, A_2, \ldots, A_{\ell}\}$ of $V$, the \emph{modularity function} is defined as follows:
\begin{eqnarray}
q(\A) &=& \sum_{A_i \in \A} \frac{e(A_i)}{|E|}  - \sum_{A_i \in \A} \left( \frac{\vol(A_i)}{\vol(V)} \right)^2, \label{eq:q_G_A}
\end{eqnarray}
where for any $A \subseteq V$, $e(A) = |\{ uv \in E : u, v \in A\}|$ is the number of edges in the subgraph of $G$ \emph{induced by} set $A$, and $\vol(A) = \sum_{v \in A} \deg(v)$ is the \emph{volume} of set $A$. In particular, $\vol(V) = 2|E|$. The first term in~(\ref{eq:q_G_A}), $\sum_{A_i \in \A} e(A_i)/|E|$, is called the \emph{edge contribution} and it computes the fraction of edges that fall within one of the parts. The second one, $\sum_{A_i \in \A} (\vol(A_i)/\vol(V))^2$, is called the \emph{degree tax} and it computes the expected fraction of edges that do the same in the corresponding random graph (the null model). The modularity measures the deviation between the two.

\medskip

It is easy to see that for any partition $\A$, $q(\A) \le 1$. On the other hand, it can be shown that $q(\A) \ge -1/2$. Also, if $\A = \{V\}$, then $q(\A) = 0$, and if $\A = \{ \{v_1\}, \{v_2\}, \ldots, \{v_n\}\}$, then $q(\A) = - \sum (\deg(v)/\vol(V))^2 < 0$. The maximum \emph{modularity} $q^*(G)$ is defined as the maximum of $q(\A)$ over all possible partitions $\A$ of $V$; that is, $q^*(G) = \max_{\A} q(\A).$ In order to maximize $q(\A)$ one wants to find a partition with large edge contribution subject to small degree tax. If $q^*(G)$ approaches 1 (which is the trivial upper bound), we observe a strong community structure; conversely, if $q^*(G)$ is close to zero (which is the trivial lower bound), there is no community structure. The definition in~(\ref{eq:q_G_A}) can be generalized to weighted edges by replacing edge counts with sums of edge weights. It can also be generalized to hypergraphs~\cite{kaminski2019clustering,kaminski2020community}. 

%%%%%%%%%%%%%%%%%%%%%%%%%%%%%%%%%%%%%%%%%%%%%%%%%%%%%%%%%%%
\section{Related Results for Random Graphs}\label{sec:related_results}
%%%%%%%%%%%%%%%%%%%%%%%%%%%%%%%%%%%%%%%%%%%%%%%%%%%%%%%%%%%

Analyzing the maximum modularity $q^*(G)$ for sparse random graphs is a challenging task. The most attention was paid to \textbf{random $d$-regular graphs} $\mathcal{G}_{n,d}$ but even for this family of graphs we only know upper and lower bounds for $q^*(\mathcal{G}_{n,d})$ that are quite apart from each other. For example, for random $3$-regular graph $\mathcal{G}_{n,3}$ we only know that \whp\ 
$$
0.667026 \le q^*(\mathcal{G}_{n,3}) \le 0.789998.
$$
These bounds were recently proved in~\cite{lichev2020modularity} but the main goal of that paper was to confirm the conjecture from~\cite{mcdiarmid2018modularity} that \whp\ $q^*(\mathcal{G}_{n,3}) \ge 2/3 + \eps$ for some $\eps>0$. We refer the reader to~\cite{mcdiarmid2018modularity,prokhorenkova2017modularity} for numerical bounds on $q^*(\mathcal{G}_{n,d})$ for other values of $d \ge 3$ and for some explicit but weaker bounds. It is also known that \whp\ $q^*(\mathcal{G}_{n,2}) \sim 1$~\cite{mcdiarmid2018modularity}.

The \textbf{binomial random graphs} $\mathcal{G}(n,p)$ were studied in~\cite{mcdiarmid2020modularity} where it was shown that \whp\ $q^*(\mathcal{G}(n,p)) \sim 1$, provided that $pn \le 1$ whereas \whp\ $q^*(\mathcal{G}(n,p)) = \Theta ( 1/\sqrt{pn})$, provided that $pn \ge 1$ and $p < 1-\eps$ for some $\eps > 0$. 
The modularity of the well-known \textbf{Preferential Attachment (PA) model}~\cite{barabasi1999emergence} and the \textbf{Spatial Preferential Attachment (SPA) model}~\cite{aiello2008spatial} was studied in~\cite{prokhorenkova2017modularity}. 
Finally, the modularity of a model of random geometric graphs on the hyperbolic plane~\cite{krioukov2010hyperbolic}, known as the \textbf{KPKBV model} after its inventors, was recently studied in~\cite{chellig2022modularity}.

%%%%%%%%%%%%%%%%%%%%%%%%%%%%%%%%%%%%%%%%%%%%%%%%%%%%%%%%%%%
\section{Preliminaries}\label{sec:preliminaries}
%%%%%%%%%%%%%%%%%%%%%%%%%%%%%%%%%%%%%%%%%%%%%%%%%%%%%%%%%%%

%%%%%%%%%%%%%%%%%%%%%%%%%%%%%%%%%%%%%%%%%%%%%%%%%%%%%%%%%%%
\subsection{Chernoff Bounds and Their Generalization}\label{sec:Chernoff}

Let us first state a specific instance of Chernoff's bound that we will find often useful. Let $X \in \textrm{Bin}(n,p)$ be a random variable with the binomial distribution with parameters $n$ and $p$. Then, a consequence of \emph{Chernoff's bound} (see e.g.~\cite[Corollary~2.3]{JLR}) is that 
\begin{equation}\label{chern}
\Prob( |X-\E [X]| \ge \eps \ \E [X] ) \le 2\exp \left( - \frac {\eps^2 \ \E [X]}{3} \right)  
\end{equation}
for  $0 < \eps < 3/2$. However, at some point we will need need a stochastic upped bound for $X$ when $\E[X]$ is small. In such situations the following bound can be applied instead of~(\ref{chern}) (see e.g.~\cite[Theorem~2.1]{JLR}):
\begin{equation}\label{chern2}
\Prob( X \ge \E [X] + u ) \le \exp \left( - \frac {u^2}{2(\E [X] + u/3)} \right).
\end{equation}

Let us mention that the above bounds hold for the general case in which $X=\sum_{i=1}^n X_i$ and $X_i \in \textrm{Bernoulli}(p_i)$ with (possibly) different $p_i$ (e.g.~see~\cite[Theorem~2.8]{JLR}). Moreover, they also hold for a hypergeometric distribution with parameters $n$, $z$, and $t$, where $\max(z,t) \le n$. Let $T$ be a subset of $[n]$ of size $t$ selected uniformly at random. Then a random variable with a hypergeometric distribution is defined as follows: $X = |T \cap [z]|$. Then the above inequalities hold with $\E [X] = zt/n$ (again, e.g.~see~\cite[Theorem~2.10]{JLR}).

\bigskip

We will also need the following result which can be viewed as a generalization of Chernoff bounds. In particular, the two inequalities above ((\ref{chern} and~(\ref{chern2})) are special cases when $c=1$.

\begin{lemma}\label{lem:chernoff_gen}
Let $(c_1, c_2, \ldots, c_r)$ be a sequence of natural numbers with  $c = \max_{i} c_i$. 
Let $S_j = \sum_{i=1}^j c_i Z_i$, where $Z_i, i \in [r]$ are independent Bernoulli($p$) random variables. Let $\mu_j = \E [S_j]=  p\sum_{i=1}^j c_i$, and let $\mu=\mu_r=\E[S_r]$.
Then for $u \ge 0$ we have that
\begin{eqnarray*}
\Prob \left(\max_{1\leq j\leq r} (S_j-\mu_j) \ge u \right) &\le& \exp \left( - \frac {u^2}{2c(\mu+u/3)} \right) \mbox{ and }\\
\Prob \left(\max_{1\leq j\leq r} (\mu_j-S_j) \ge u \right) &\le& \exp \left( - \frac {u^2}{2c \mu} \right). \\
\end{eqnarray*}
In particular, for $\eps \le 3/2$ we have that
$$
\Prob \left( \max_{1\leq j\leq r} |S_j-\mu_j| \ge \eps \mu \right) \le 2 \exp \left( - \frac {\eps^2 \mu}{3c} \right).
$$
\end{lemma}

To prove this lemma, one can easily adjust the proof of the classic Chernoff bound. Alternatively, the same bounds come from~\cite{mcdiarmid1998concentration}. In that paper, the counterpart of Lemma~\ref{lem:chernoff_gen} is stated for $S_r-\mu$ (see Theorem~2.3); however, the author comments that $S_r-\mu$ can be replaced with $\max_{1\leq j\leq r} (S_j-\mu_j)$ (which is a slightly stronger version than we need here) as follows. A standard martingale bound shows that eg.\ for any $h>0$:
\[\Prob \left( \max_{1\leq j\leq r} (S_j-\mu_j) \ge u \right) \leq e^{-hu}\E \left[ e^{h(S_r-\mu)} \right]. \]
Then plugging this into the appropriate place in the proof of Theorem 2.3 yields the desired bounds.

%%%%%%%%%%%%%%%%%%%%%%%%%%%%%%%%%%%%%%%%%%%%%%%%%%%%%%%%%%%
\subsection{Expansion Properties of Random $d$-regular graphs}\label{sec:expanders}

Let $\mathcal{G}_{n,d}$ be the probability space of \textbf{random $d$-regular simple graphs} with uniform probability distribution ($d \ge 2$ is fixed and $n$ is even if $d$ is odd). It is an easy fact that the probability of a random pairing $\mathcal{P}_{n,d} := \Pc(\textbf{w})$ with $\textbf{w} := (d, d, \ldots, d)$ corresponding to a given simple $d$-regular graph $G$ is independent of the graph. As a result, the restriction of the probability space $\mathcal{P}_{n,d}$ to simple graphs is precisely $\mathcal{G}_{n,d}$. Moreover, it is well known that a random pairing generates a simple graph with probability asymptotic to $e^{(1-d^2)/4}$ depending on $d$, so that any event holding \whp\ over the probability space of random pairings also holds \whp\ over the corresponding space $\mathcal{G}_{n,d}$. For this reason, asymptotic results over random pairings can be immediately transferred to $\mathcal{G}_{n,d}$. For more information on this model, see~\cite{wormald1999models}.

We will use the well-known expansion properties of random $d$-regular graphs that follow from their eigenvalues. These expansion properties are known to hold \whp\ for $\mathcal{G}_{n,d}$ but, fortunately, they actually hold \whp\ for $\mathcal{P}_{n,d}$ (and so, by the argument mentioned above, they immediately hold \whp\ for $\mathcal{G}_{n,d}$ which is what is typically used). In one of our proofs (see the proof of Theorem~\ref{thm:modularity_small_xi}), we will couple the \ABCD\ model $\Ac$ with a random pairing $\mathcal{P}_{n,d}$. This coupling will allow us to deduce some useful expansion properties of $\Ac$ from the corresponding properties of $\mathcal{P}_{n,d}$. 

\medskip

The adjacency matrix $A=A(G)$ of a given a $d$-regular (multi)graph $G$ with $n$ nodes, is an $n \times n$ real and symmetric matrix. Thus, the matrix $A$ has $n$ real eigenvalues which we denote by $\lambda_1 \ge \lambda_2 \ge \cdots \ge \lambda_n$. It is known that certain properties of a $d$-regular graph are reflected in its spectrum but, since we focus on expansion properties, we are particularly interested in the following quantity: $\lambda = \lambda(G) = \max( |\lambda_2|, |\lambda_n|)$. In words, $\lambda$ is the largest absolute value of an eigenvalue other than $\lambda_1 = d$. For more details, see the general survey~\cite{hoory2006expander} about expanders, or~\cite[Chapter 9]{alon2016probabilistic}.

The value of $\lambda$ for random $d$-regular graphs has been studied extensively. A major result due to Friedman~\cite{friedman2008proof} is the following:
\begin{lemma}[\cite{friedman2008proof}]\label{lem:Fri}
For every fixed  $\varepsilon > 0$ and for $G\in \mathcal{P}_{n,d}$ (and so also for $G\in \mathcal{G}_{n,d}$), \whp\
$$
\lambda(G) \le 2 \sqrt{d-1}+ \varepsilon.
$$
\end{lemma}

The number of edges $|E(S,T)|$ between sets $S$ and $T$ is expected to be close to the expected number of edges between $S$ and $T$ in a random graph of edge density $d/n$, namely, $d|S||T|/n$. A~small $\lambda$ (or large spectral gap) implies that this deviation is small. The following useful bound is essentially proved in~\cite{alon1988explicit} (see also~\cite{alon2016probabilistic}):
\begin{lemma}[Expander Mixing Lemma]\label{lem:AC}
Let $G=(V,E)$ be a $d$-regular (multi)graph with $n$ nodes and set $\lambda = \lambda(G)$. Then for all $S, T \subseteq V$
$$
\left| |E(S, T)| - \frac{d|S||T|}{n} \right| \le \lambda \sqrt{|S||T|} \,.
$$
\end{lemma}
\noindent 
(Note that $S \cap T$ does not have to be empty; in general, $|E(S,T)|$ is defined to be the number of edges between $S \setminus T$ to $T$ plus twice the number of edges that contain only nodes of $S \cap T$.)

The Expander Mixing Lemma is very useful but for our purpose it is better to apply a slightly stronger lower estimate for $|E(S,V \setminus S)|$, namely,
\begin{equation} \label{eqn:bisection}
|E(S,V \setminus S)| \geq \frac{(d-\lambda)|S||V \setminus S|}{n}
\end{equation}
for all $S \subseteq V$. This is proved in~\cite{alon1985lambda1}, see also~\cite{alon2016probabilistic}. Combining this inequality with a simple averaging argument gives immediately the following useful bound for the maximum modularity that was observed in~\cite{prokhorenkova2017modularity}. 

\begin{lemma}[\cite{prokhorenkova2017modularity}]\label{lem:mod_expander}
Let $G=(V,E)$ be a $d$-regular (multi)graph with $n$ nodes and set $\lambda = \lambda(G)$. Then, 
$$
q^*(G) \le \frac {\lambda}{d}\,.
$$
\end{lemma}

%%%%%%%%%%%%%%%%%%%%%%%%%%%%%%%%%%%%%%%%%%%%%%%%%%%%%%%%%%%
\section{Some Properties of ABCD}\label{sec:basic_properties}
%%%%%%%%%%%%%%%%%%%%%%%%%%%%%%%%%%%%%%%%%%%%%%%%%%%%%%%%%%%

%%%%%%%%%%%%%%%%%%%%%%%%%%%%%%%%%%%%%%%%%%%%%%%%%%%%%%%%%%%
\subsection{Degree Distribution}

Let $\gamma \in (2,3)$, $\delta \in \N$, and $\zeta \in (0,1)$. We will show soon that we may assume that $\zeta \le \frac {1}{\gamma-1}$ but for now we allow $\zeta$ to be any value from $(0,1)$. Recall that the degrees of nodes of the \ABCD\ model are generated randomly following the (truncated) \emph{power-law distribution} $\Pc(\gamma, \delta, \zeta)$ with exponent $\gamma$, minimum value $\delta$, and maximum value $D = n^{\zeta}$. More precisely, if $X \in \Pc(\gamma, \delta, \zeta)$, then for any $k \in \{ \delta, \delta+1, \ldots, D\}$,
\begin{eqnarray}
q_k &=& \Pr( X = k ) = \frac { \int_{k}^{k+1} x^{-\gamma} dx }{ \int_{\delta}^{D+1} x^{-\gamma} dx } = \frac { k^{1-\gamma} - (k+1)^{1-\gamma} } { \delta^{1-\gamma} - (D+1)^{1-\gamma} } \nonumber \\
&=& (1+\bigo(n^{-\zeta (\gamma-1)})) \left( k^{-(\gamma-1)} - (k+1)^{-(\gamma-1)} \right) \delta^{\gamma-1} \label{eq:qk_small} \\
&=& (1+\bigo(n^{-\zeta (\gamma-1)})) \ k^{-(\gamma-1)} \left( 1 - (1+1/k)^{-(\gamma-1)} \right) \delta^{\gamma-1} \nonumber \\
&=& (1+\bigo(n^{-\zeta (\gamma-1)})) \ k^{-(\gamma-1)} \left( 1 - (1-(\gamma-1)/k+\bigo(1/k^2)) \right) \delta^{\gamma-1} \nonumber \\
&=& (1+\bigo(n^{-\zeta (\gamma-1)})+\bigo(k^{-1})) \ k^{-(\gamma-1)} \frac { \gamma-1 }{ k } \delta^{\gamma-1} \nonumber \\
&=& (1+\bigo(n^{-\zeta (\gamma-1)})+\bigo(k^{-1})) \ k^{-\gamma} (\gamma-1) \delta^{\gamma-1}.  \label{eq:qk_large}
\end{eqnarray}

\medskip

As promised earlier, we start with a proof of an upper bound for the maximum degree, which justifies our future assumption that $\zeta \in (0,1/(\gamma-1)]$. 

\begin{lemma}\label{lem:max_degree}
Let $\gamma \in (2,3)$, $\delta \in \N$, and $\zeta \in (0,1)$. Let $\omega=\omega(n)$ be any function tending to infinity as $n \to \infty$. Then, \whp\ the maximum degree of $\Ac$ is at most $\min (n^{\zeta},n^{1/(\gamma-1)} \omega)$.
\end{lemma}
\begin{proof}
Trivially, by definition, the maximum degree is at most $D = n^{\zeta}$. Hence, the desired property holds (deterministically) if $\zeta \le 1/(\gamma-1)$. Let us then assume that $\zeta > 1/(\gamma-1)$. Our goal is to show that \whp\ the maximum degree of $\Ac$ is at most $K:=n^{1/(\gamma-1)} \omega$. Note that the expected number of nodes with degree at least $K$ is equal to
$$
n \sum_{k=K}^{D} q_k = n \ \frac { \int_{K}^{D+1} x^{-\gamma} dx }{ \int_{\delta}^{D+1} x^{-\gamma} dx } = n \ \frac { K^{1-\gamma} - (D+1)^{1-\gamma} } { \delta^{1-\gamma} - (D+1)^{1-\gamma} } = \bigo( n \, K^{-(\gamma-1)} ) = \bigo( \omega^{-(\gamma-1)} ) = o(1).
$$
Hence, by the first moment method, \whp\ there is no node with degree at least $K$ and the proof of the lemma is finished.
\end{proof}

\medskip

Now, we are ready to investigate the degree distribution. 

\begin{lemma}\label{lem:degree_distribution}
Let $\gamma \in (2,3)$, $\delta \in \N$, and $\zeta \in (0,\frac{1}{\gamma-1}]$. 
For $k \in \N$, let $Y_k$ be the random variable counting the number of nodes in $\Ac$ that are of degree $k$.
For $k \in \N$ and $\eta = \eta(n)$, let $Y_k^{\eta}$ be the random variable counting the number of nodes in $\Ac$ that are of degree at least $k$ but at most $(1+\eta)k$.

The following properties hold \wep:
\begin{itemize}
\item[(a)] If $\zeta \in (0, \frac{1}{\gamma})$, then for any $k \in \N$ such that $\delta \le k \le n^{\zeta}$ we have
\begin{eqnarray}
Y_k &=& (1+\bigo( (\log n)^{-1} )) \ n q_k \nonumber \\
&=& (1+\bigo( (\log n)^{-1} )) \ n \left( k^{-(\gamma-1)} - (k+1)^{-(\gamma-1)} \right) \delta^{\gamma-1} \nonumber \\
&=& (1+\bigo( k^{-1}) + \bigo( (\log n)^{-1} )) \ n k^{-\gamma} (\gamma-1) \delta^{\gamma-1}. \label{eq:light_nodes}
\end{eqnarray}
\item[(b)] If $\zeta \in [\frac{1}{\gamma}, \frac{1}{\gamma-1})$, then for any $k \in \N$ such that $\delta \le k \le n^{1/\gamma} (\log n)^{-4/\gamma} \ll n^{\zeta}$ random variable $Y_k$ satisfies~(\ref{eq:light_nodes}), and for any $k \in \N$ such that $n^{1/\gamma} (\log n)^{-4/\gamma} \le k \le n^{\zeta}/(1+\eta)$ we have
\begin{eqnarray}
Y_k^{\eta} &=& (1+\bigo( (\log n)^{-1} )) \ n \eta k q_k \nonumber \\
&=& (1+\bigo( (\log n)^{-1} )) \ (\gamma-1) \delta^{\gamma-1} (\log n)^4, \label{eq:heavy_nodes}
\end{eqnarray}
where 
$$
\eta = \eta(k) = n^{-1} (\log n)^4 k^{\gamma - 1} = \bigo( (\log n)^{-1} ) = o(1).
$$
\item[(c)] If $\zeta = \frac{1}{\gamma-1}$, then for any $k \in \N$ such that $\delta \le k \le n^{1/\gamma} (\log n)^{-4/\gamma} \ll n^{\zeta}$ random variable $Y_k$ satisfies~(\ref{eq:light_nodes}), and for any $k \in \N$ such that $n^{1/\gamma} (\log n)^{-4/\gamma} \le k \le n^{\zeta} (\log n)^{-5/(\gamma-1)}$ random variable $Y_k^{\eta}$ satisfies~(\ref{eq:heavy_nodes}). The number of nodes of degree at least $n^{\zeta} (\log n)^{-5/(\gamma-1)}$ is equal to $\Theta( (\log n)^5 )$.
\end{itemize}
\end{lemma}

Before we move to the proof of this lemma, let us mention the following straightforward corollary. Recall that for a given set of nodes $A \subseteq V$, the \emph{volume} of $A$ is defined as follows: 
$$
\vol(A) = \sum_{v \in A} \deg(v).
$$
In particular, $\vol(V) = 2|E|$.

\begin{corollary}\label{cor:volume_of_A}
The volume of all nodes in $\Ac$ is \wep\ equal to
$$
\vol(V) = \sum_{k = \delta}^{D} k Y_k = (1+\bigo( (\log n)^{-1} )) \ d n, \hspace{1cm} \text{ where } d := \sum_{k = \delta}^{D} \ k q_k.
$$
\end{corollary}

Unfortunately, there is no closed formula for a constant $d$ but one can easily approximate it numerically and obtain some theoretical bounds. For example, note that
\begin{eqnarray*}
d &=& \sum_{k = \delta}^{D} \ k \frac { \int_{k}^{k+1} x^{-\gamma} dx }{ \int_{\delta}^{D+1} x^{-\gamma} dx } \le \sum_{k = \delta}^{D} \ \frac { \int_{k}^{k+1} x^{1-\gamma} dx }{ \int_{\delta}^{D+1} x^{-\gamma} dx } = \frac { \int_{\delta}^{D+1} x^{1-\gamma} dx }{ \int_{\delta}^{D+1} x^{-\gamma} dx } = \frac { (\delta^{2-\gamma} - (D+1)^{2-\gamma} ) / (\gamma-2) } { (\delta^{1-\gamma} - (D+1)^{1-\gamma} ) / (\gamma-1) } \\
&=& (1+\bigo( n^{-\zeta(\gamma-2)} )) \ \delta \ \frac {\gamma-1}{\gamma-2}.
\end{eqnarray*}
On the other hand, 
\begin{eqnarray*}
d &=& \sum_{k = \delta}^{D} \ k \frac { \int_{k}^{k+1} x^{-\gamma} dx }{ \int_{\delta}^{D+1} x^{-\gamma} dx } \ge \sum_{k = \delta}^{D} \ \frac {k}{k+1} \cdot \frac { \int_{k}^{k+1} x^{1-\gamma} dx }{ \int_{\delta}^{D+1} x^{-\gamma} dx } \ge \frac {\delta}{\delta+1} \cdot \frac { \int_{\delta}^{D+1} x^{1-\gamma} dx }{ \int_{\delta}^{D+1} x^{-\gamma} dx } \\
&=& (1+\bigo( n^{-\zeta(\gamma-2)} )) \ \frac {\delta^2}{\delta+1} \cdot \frac {\gamma-1}{\gamma-2}.
\end{eqnarray*}

\begin{proof}[Proof of Lemma~\ref{lem:degree_distribution}]
For now, suppose that the degree sequence is simply a sequence of $n$ independent random variables, each of them following power-law distribution $\Pc(\gamma, \delta, \zeta)$. In the end, the degree of a node of the largest degree might possibly be decreased by one to make sure the sum of degrees is even. Clearly, this small adjustment will not affect out asymptotic bounds.

We will call a node \emph{light} if its degree is at most $K:=n^{1/\gamma} (\log n)^{-4/\gamma}$; otherwise, it will be called \emph{heavy}. Note that if $\zeta < 1/\gamma$, then trivially all nodes are light. On the other hand, if $\zeta \ge 1/\gamma$, then (as will be shown soon) \whp\ the maximum degree is equal to $(1+o(1)) n^{\zeta}$ and so we will have a mixture of light and heavy nodes. We need to make a distinction because for a given value of $k$, $\delta \le k \le K$, there are plenty of light nodes with their degrees are equal to $k$, that is, we will show that  \wep\ $Y_k > 0$. Unfortunately, for a given value of $k$, $K < k \le n^{\zeta}$, it might happen (in fact, it is quite often the case) that there will be no heavy node that has its degree equal to $k$, that is, $Y_k = 0$. In order to solve this problem, we will need to group heavy nodes together and consider them in batches, that is, consider $Y_k^{\eta}$ instead of $Y_k$. 

We start with light nodes. Let us fix $k \in \N$ such that $\delta \le k \le K$. Using~(\ref{eq:qk_large}), we get that 
$$
\E [ Y_k ] = n q_k = \Theta( n k^{-\gamma} ) = \Omega( n K^{-\gamma} ) = \Omega( (\log n)^4 ).
$$
After applying Chernoff's bound~(\ref{chern}) with $\eps = (\log n)^{-1}$ we get that \wep\ 
$$
Y_k = (1+\bigo( (\log n)^{-1} )) \E [ Y_k ] = (1+\bigo( (\log n)^{-1} )) n q_k.
$$
The other two equalities in~(\ref{eq:light_nodes}) follow immediately from~(\ref{eq:qk_small}) and~(\ref{eq:qk_large}). This finishes part~(a) and the first half of part~(b).

We move now to heavy nodes. Assume that $\zeta \ge 1/\gamma$ and let us fix $k \in \N$ such that $K \le k \le n^{\zeta}/(1+\eta)$. If $\zeta=1/(\gamma-1)$ (part~(c)), then we additionally assume that $k \le n^{\zeta} (\log n)^{-5/(\gamma-1)}$. (We will independently deal with nodes of degrees that are very close to the maximum degree.)
Recall that
$$
\eta = \eta(k) = n^{-1} (\log n)^4 k^{\gamma - 1}.
$$
Recall also that random variable $Y_k^{\eta}$ counts a \emph{batch} of nodes with degrees at least $k$ but at most $(1+\eta)k$. Note that $\eta(k)$ is an increasing function of $k$. Moreover, $\eta(K) = 1/K$ so only a few values are considered for $k=K$ (note that $\eta(K)K = 1$). On the other extreme, $\eta(n^{\zeta}) = n^{-1+\zeta(\gamma-1)} (\log n)^4$ which is $\bigo( (\log n)^{-1} ) = o(1)$, provided that $\zeta < 1/(\gamma-1)$.  If $\zeta = 1/(\gamma-1)$ (part~(c)), then additional upper bound for $k$ also guarantees that at the extreme case $\eta(n^{\zeta} (\log n)^{-5/(\gamma-1)}) = (\log n)^{-1} = o(1)$. In any case, if $\eta = \bigo( (\log n)^{-1} )$, then all nodes in the batch have degrees equal to $(1+\bigo(\eta)) k = (1+\bigo( (\log n)^{-1} )) k$. Using~(\ref{eq:qk_large}) as before, we get that 
$$
\E [ Y_k^{\eta} ] = n \sum_{i=k}^{(1+\eta)k } q_i = \Theta \left( n (\eta k) k^{-\gamma} \right) = \Theta( (\log n)^4 ).
$$
After applying Chernoff's bound~(\ref{chern}) with $\eps = (\log n)^{-1}$ we get that \wep\ 
$$
Y_k^{\eta} = (1+O( (\log n)^{-1} )) \E [ Y_k^{\eta} ] = (1+\bigo( (\log n)^{-1} )) \ n \eta k q_k.
$$
The second equality in~(\ref{eq:heavy_nodes}) follows immediately from~(\ref{eq:qk_large}). This finishes the second half of part~(b).

To finish part~(c), it remains to concentrate on the case $\zeta = 1/(\gamma-1)$ and deal with nodes of degrees that are very close to the maximum degree. Since the expected number of nodes of degree at least $n^{\zeta} (\log n)^{-5/(\gamma-1)}$ is of order
$$
n \left( n^{\zeta} (\log n)^{-5/(\gamma-1)} \right)^{1-\gamma} = (\log n)^{5},
$$
we may apply Chernoff's bound~(\ref{chern}) with $\eps = (\log n)^{-1}$ for the last time to get the desired concentration. The proof of the lemma is finished.
\end{proof}

\subsubsection*{Simulation Corner}

In order to see whether asymptotic results can be used to predict the behaviour for relatively small values of $n$, we generated the degree distributions for two \ABCD\ graphs $\Ac$ on $n=1{,}000$ and, respectively, $n=1{,}000{,}000$ nodes. In both cases, we used parameters $\gamma = 2.5$, $\delta = 5$, and $\zeta = 1/2 < 2/3 = 1/(\gamma-1)$ (that is, $D=\sqrt{n}$). On Figure~\ref{fig:degree} we plot the complement of the cumulative degree distribution (that is, the fraction of nodes of degree at least $K$) and compare it with asymptotic, theoretical predictions, namely, function
$$
\sum_{k=K}^{D} q_k = \frac { \int_{K}^{D+1} x^{-\gamma} dx }{ \int_{\delta}^{D+1} x^{-\gamma} dx } = \frac { K^{1-\gamma} - (D+1)^{1-\gamma} } { \delta^{1-\gamma} - (D+1)^{1-\gamma} }.
$$
We observe almost perfect agreement, especially for a larger graph. As mentioned earlier, the theoretical distribution uses the continuous model whereas when generating \ABCD\ we used the discrete distribution. As a consequence, we see a slight deviation between the curves indicating that the theoretical continuous model generates slightly larger node degrees.
\begin{figure}[ht]
     \centering
     \includegraphics[width=0.48\textwidth]{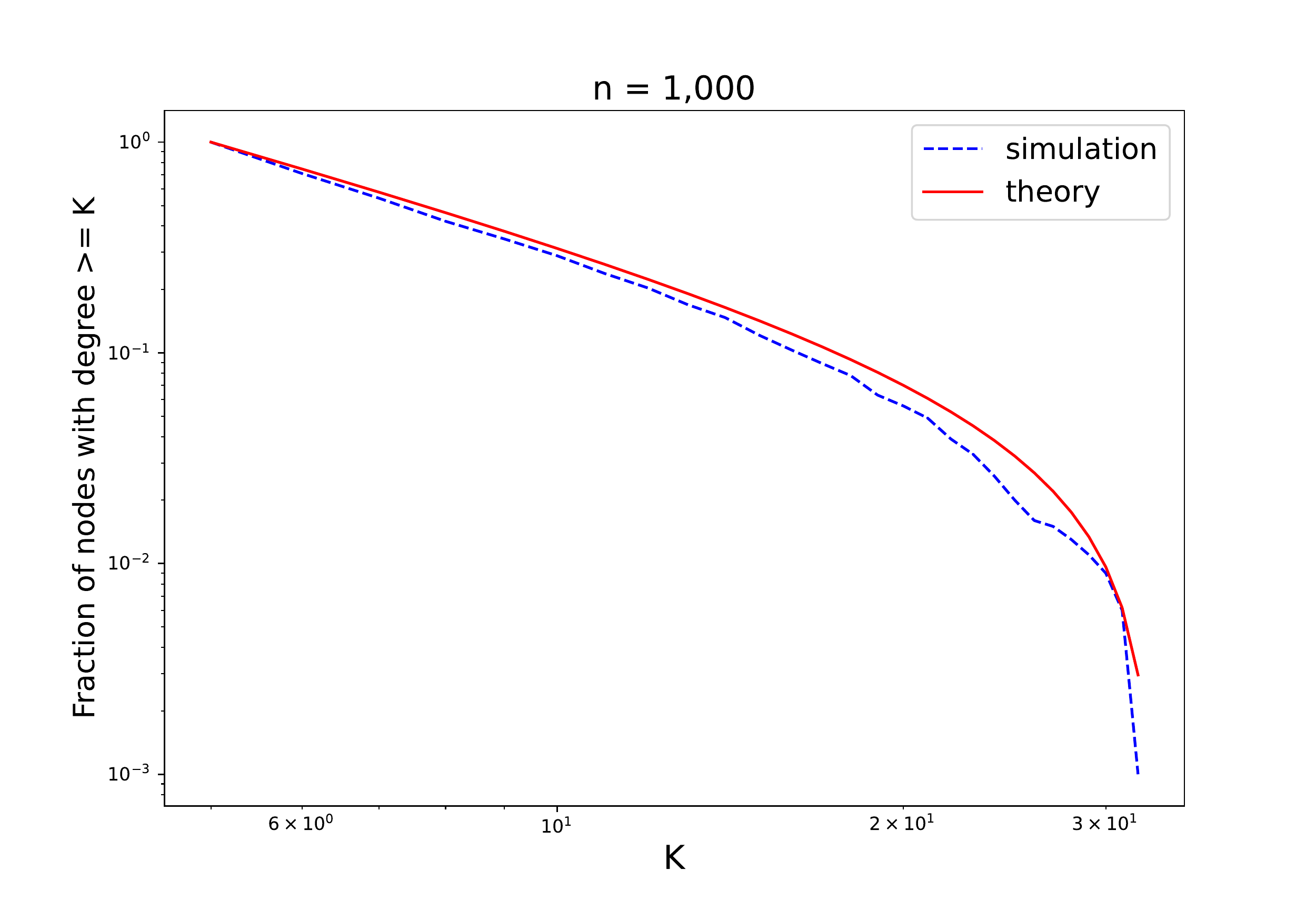}
         \hspace{.1cm}
     \includegraphics[width=0.48\textwidth]{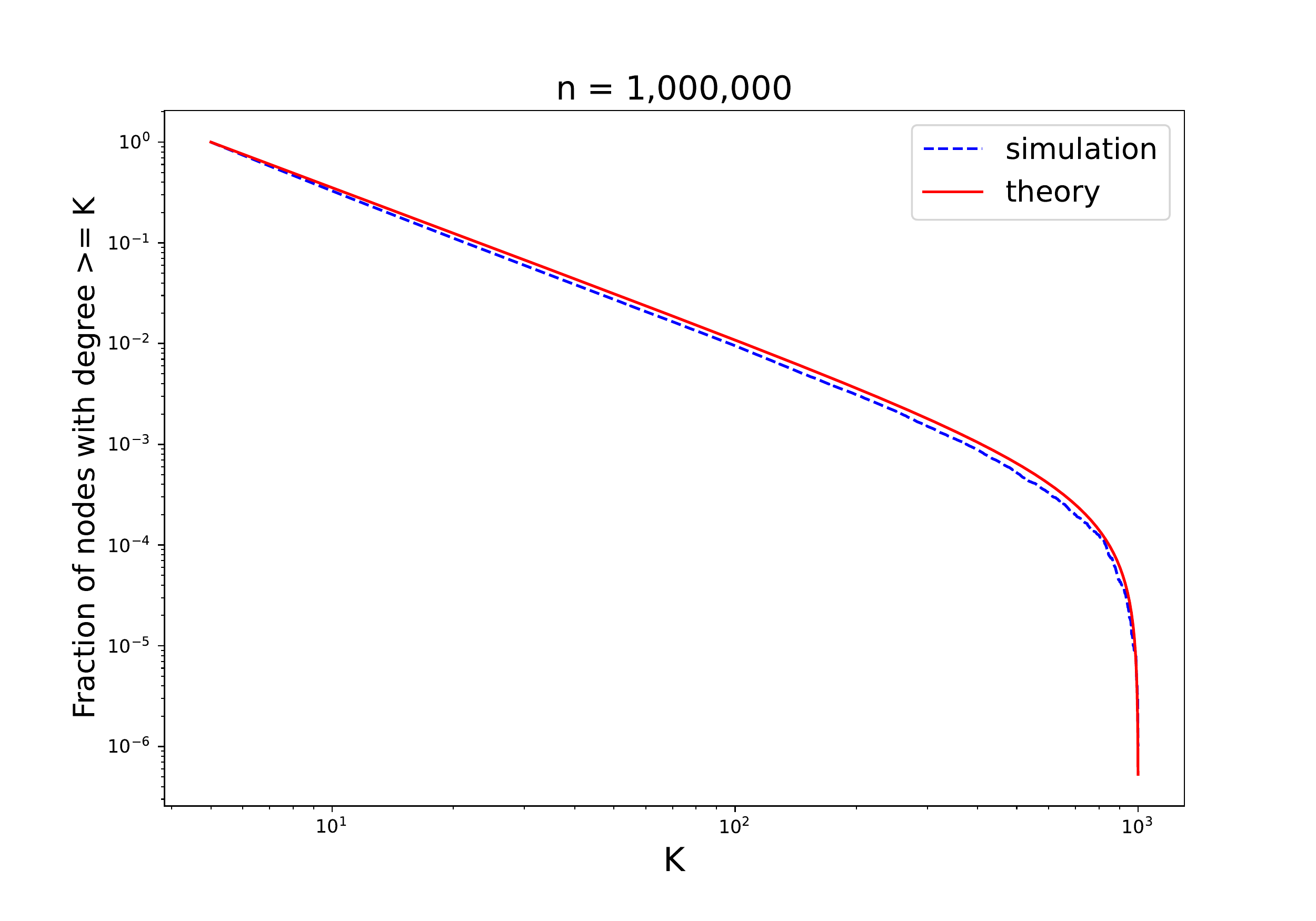}
     \caption{Complement of cumulative degree distribution for small ($n=1{,}000$; left plot) and large ($n=1{,}000{,}000$; right plot) graphs using the following parameters: $\gamma = 2.5$, $\delta = 5$, and $\zeta = 1/2$.}
     \label{fig:degree}
 \end{figure}

Our next experiment investigates how well Corollary~\ref{cor:volume_of_A} predicts the volume of $\Ac$ in practice. For each value of $n = 1000 \cdot 2^i$, $i \in \{0, 1, \ldots,15\}$, we independently generated 100 graphs with the same parameters as in the first experiment. On Figure~\ref{fig:degree_ratio} we present the average value and the standard deviation of $\vol(V)$. We use two different scalings: $dn$, the theoretical asymptotic prediction using continuous power-law distribution, and $\hat{d}n$ where $\hat{d} = \sum_{k = \delta}^{D} \ k r_k$ (see~(\ref{eq:power-law-discrete}) for a definition of $r_k$), the discrete counterpart of~$d$. Since discrete distribution is used by default by the \ABCD\ generator, in line what we observed in Figure \ref{fig:degree}, the continuous prediction is a little bit off but the discrete one works well, especially for large graphs. 

\begin{figure}[ht]
     \centering
     \includegraphics[width=0.48\textwidth]{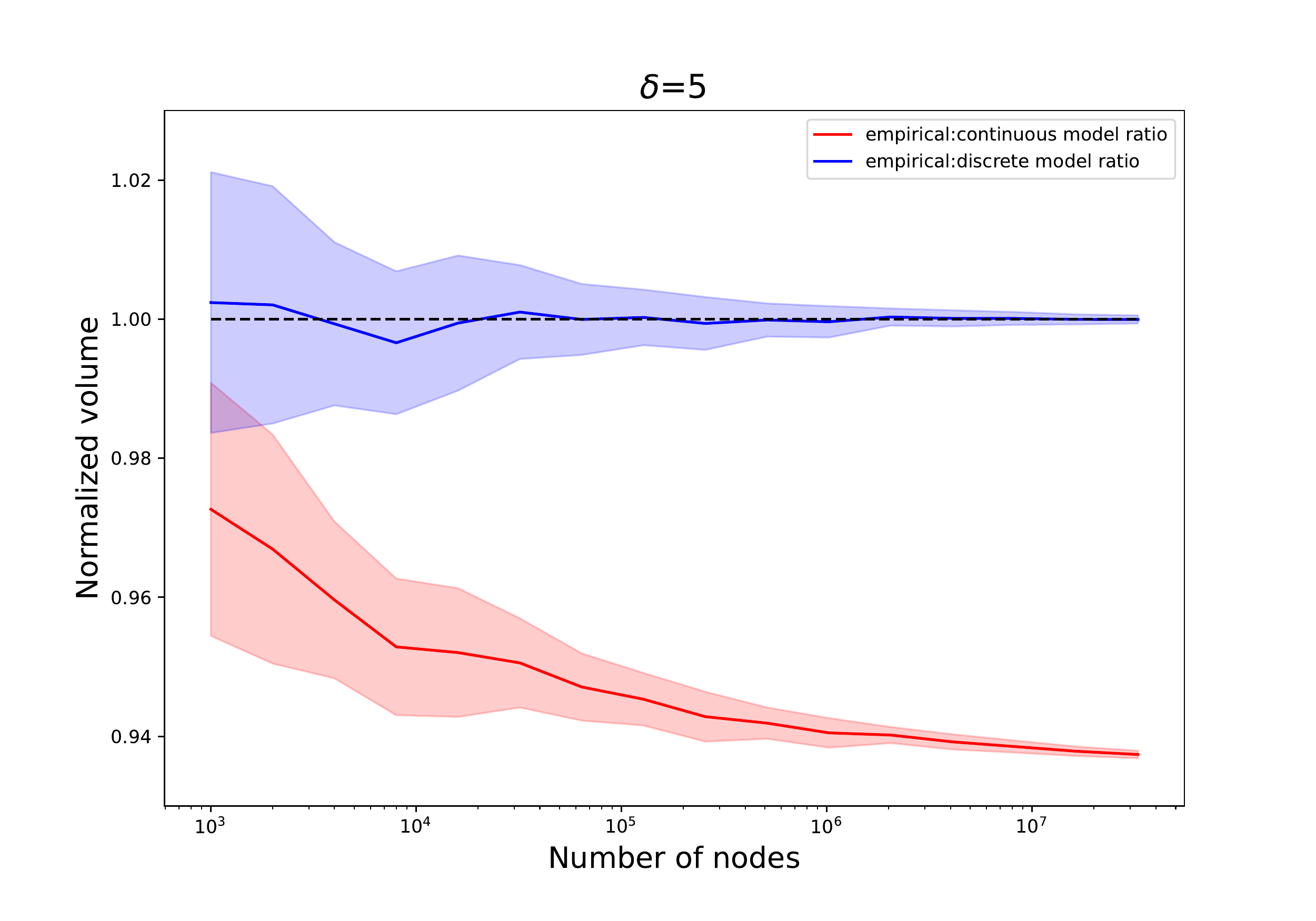}
        \hspace{.1cm}
     \includegraphics[width=0.48\textwidth]{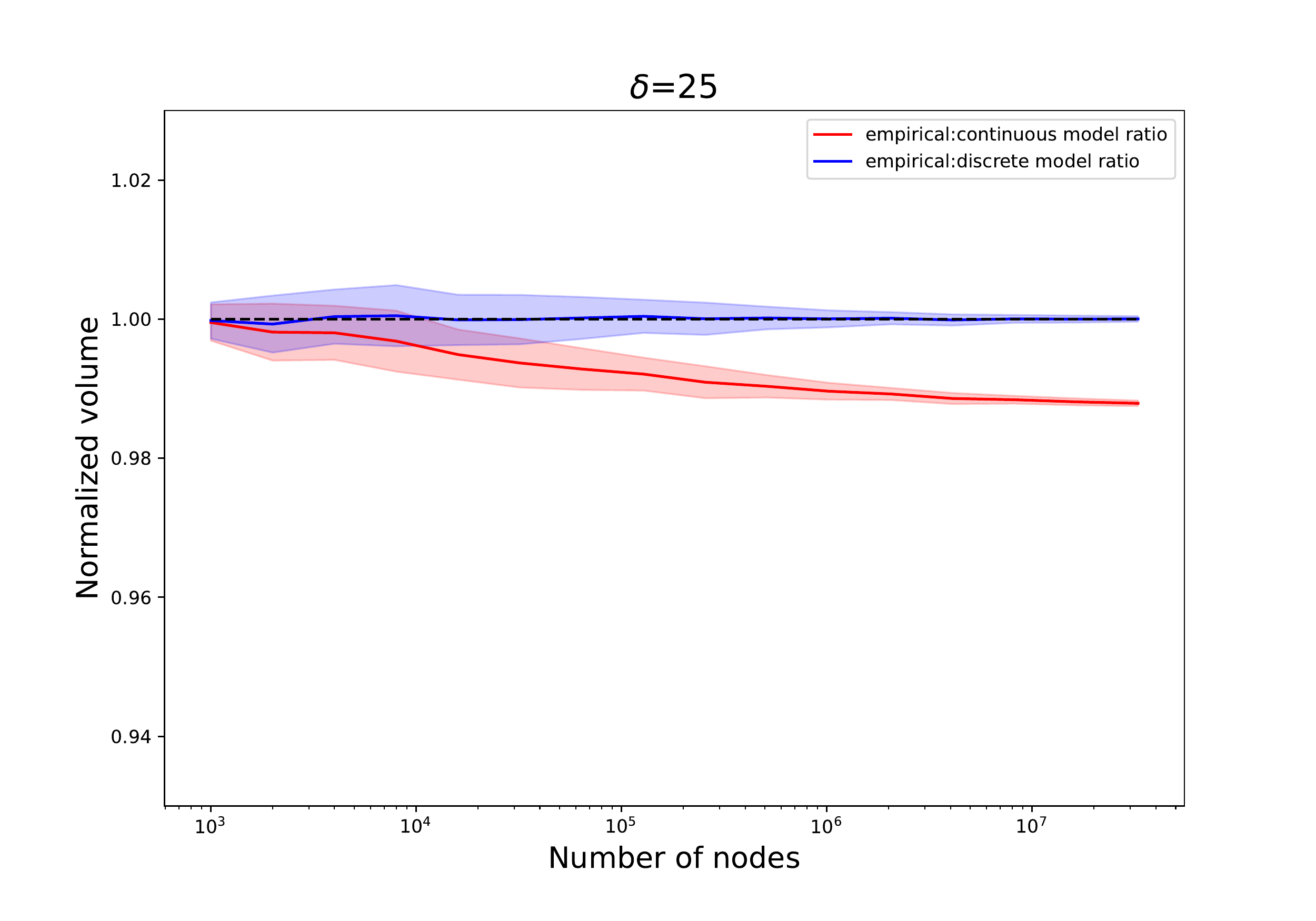}
     \caption{The average volume of 100 independently generated graphs; shaded area represents the standard deviation. Both quantities were normalized by the theoretical prediction for continuous model (red curve) and discrete one (blue curve). The dashed line at 1 corresponds to a perfect prediction. Parameters used: $\gamma = 2.5$, $\delta = 5$, and $\zeta = 1/2$ (left plot), and $\gamma = 2.5$, $\delta = 25$, and $\zeta = 1/2$ (right plot).}
     \label{fig:degree_ratio}
 \end{figure}
 
%%%%%%%%%%%%%%%%%%%%%%%%%%%%%%%%%%%%%%%%%%%%%%%%%%%%%%%%%%%
\subsection{Distribution of Community Sizes}

Let $\beta \in (1,2)$, $s \in \N$, and $\tau \in (\zeta,1)$. Recall that community sizes of the \ABCD\ model are generated randomly following the (truncated) \emph{power-law distribution} $\Pc(\beta, s, \tau)$ with exponent $\beta$, minimum value $s$, and maximum value $S = n^{\tau}$. More precisely, if $X \in \Pc(\beta, s, \tau)$, then after following exactly the same computation as in~(\ref{eq:qk_small}) and~(\ref{eq:qk_large}) we get that for any $k \in \{ s, s+1, \ldots, S\}$,
\begin{eqnarray}
p_k &=& \Pr( X = k ) = \frac { \int_{k}^{k+1} x^{-\beta} dx }{ \int_{s}^{S+1} x^{-\beta} dx } = \frac { k^{1-\beta} - (k+1)^{1-\beta} } { s^{1-\beta} - (S+1)^{1-\beta} } \nonumber \\
&=& (1+\bigo(n^{-\tau (\beta-1)})) \left( k^{-(\beta-1)} - (k+1)^{-(\beta-1)} \right) s^{\beta-1} \label{eq:pk_small} \\
%&=& (1+O(n^{-\tau (\beta-1)})) \ k^{-(\beta-1)} \left( 1 - (1+1/k)^{-(\beta-1)} \right) s^{\beta-1} \nonumber \\
%&=& (1+O(n^{-\tau (\beta-1)})) \ k^{-(\beta-1)} \left( 1 - (1-(\beta-1)/k+O(1/k^2)) \right) s^{\beta-1} \nonumber \\
%&=& (1+O(n^{-\tau (\beta-1)})+O(k^{-1})) \ k^{-(\beta-1)} \frac { \beta-1 }{ k } s^{\beta-1} \nonumber \\
&=& (1+\bigo(n^{-\tau (\beta-1)})+\bigo(k^{-1})) \ k^{-\beta} (\beta-1) s^{\beta-1}.  \label{eq:pk_large}
\end{eqnarray}

Community sizes of \ABCD\ are generated randomly and independently, each of them following power-law distribution $\Pc(\beta, s, \tau)$. 
Our first task is to investigate how many communities there are in a graph on $n$ nodes, and how their sizes are distributed. It is slightly easier to fix the number of communities, generate their sizes independently, and check how large graph they span. Once we establish this, we will simply inverse the process, fix the number of nodes of the graph to be $n$ and compute the number of communities~$\ell$ that need to be generated to reach the desired number of nodes.

The assumption that $\tau > \zeta$ is introduced to make sure large degree nodes have large enough communities to be assigned to. We do not need this assumption in the first lemma as it is concerned with a sequence of random variables, not the \ABCD\ model. 

\begin{lemma}\label{lem:communities}
Let $\beta \in (1,2)$, $s \in \N$, and $\tau \in (0,1)$. 
Let $$\ell = \ell(n) = c n^{1-\tau (2-\beta)}$$ for some function $c = c(n) \to \hat{c}$, where $\hat{c} \in \R_+$.
Let $(X_1, X_2, \ldots, X_\ell)$ be a sequence of independent random variables, $X_i \in \Pc(\beta, s, \tau)$ for any $i \in [\ell]$. 
For $k \in \N$, let $Y_k$ be the random variable counting the number of variables $X_i$ that are equal to $k$.
For $k \in \N$ and $\eta=\eta(n)$, let $Y_k^{\eta}$ be the random variable counting the number of variables $X_i$ that are at least $k$ but at most $(1+\eta)k$. 

The following properties hold \wep:
\begin{itemize}
\item[(a)] If $\tau \in (0,1/2)$, then for any $k \in \N$ such that $s \le k \le n^{\tau}$ we have
\begin{eqnarray}
Y_k &=& (1+\bigo( (\log n)^{-1} )) \ \ell p_k \nonumber \\
&=& (1+\bigo( (\log n)^{-1} )) \ c n^{1-\tau (2-\beta)} \left( k^{-(\beta-1)} - (k+1)^{-(\beta-1)} \right) s^{\beta-1} \nonumber \\
&=& (1+\bigo( k^{-1}) + \bigo( (\log n)^{-1} )) \ c n^{1-\tau (2-\beta)} k^{-\beta} (\beta-1) s^{\beta-1}. \label{eq:light_nodes_comm}
\end{eqnarray}
\item[(b)] If $\tau \in [1/2, 1)$, then for any $k \in \N$ such that $s \le k \le n^{\tau - (2\tau - 1)/\beta} (\log n)^{-4/\beta} \ll n^{\tau}$ random variable $Y_k$ satisfies~(\ref{eq:light_nodes_comm}),
and for any $k \in \N$ such that $n^{\tau - (2\tau - 1)/\beta} (\log n)^{-4/\beta} \le k \le n^{\tau}/(1+\delta)$ we have
\begin{eqnarray*}
Y_k^{\eta} &=& (1+\bigo( (\log n)^{-1} )) \ \ell \eta k p_k \\
&=& (1+\bigo( (\log n)^{-1} )) \ c (\beta-1) s^{\beta-1} (\log n)^4,
\end{eqnarray*}
where 
$$
\eta = \eta(k) = n^{\tau(2-\beta)-1} (\log n)^4 k^{\beta - 1} \le n^{-(1-\tau)} (\log n)^4 = \bigo( (\log n)^{-1} ) = o(1).
$$
\end{itemize}
In any case, 
\begin{equation}\label{eq:sum_communities}
\sum_{i=1}^{\ell} X_i = (1+\bigo( (\log n)^{-1} )) \ c n s^{\beta-1} \ \frac { \beta-1 } {2-\beta} = (1+o(1)) n,
\end{equation}
provided
$$
\hat{c} = \frac {2-\beta}{(\beta-1) s^{\beta-1}}.
$$
\end{lemma}

Before we prove the lemma, let us state the following straightforward implication for the \ABCD\ model. 

\begin{corollary}\label{cor:community_sizes}
Let $\beta \in (1,2)$, $s \in \N$, and $\tau \in (\zeta,1)$. The following properties hold \wep\ for $\Ac$. 
\begin{itemize}
\item [(a)] The number of communities is equal to 
$$
\ell = \ell(n) = (1+\bigo( (\log n)^{-1} )) \, \hat{c} \, n^{1-\tau (2-\beta)},
$$
where
$$
\hat{c} = \frac {2-\beta}{(\beta-1) s^{\beta-1}}.
$$
\item [(b)] For $k \in \N$, let $Y_k$ be the number of communities of size $k$.
For $k \in \N$ and $\eta=\eta(n)$, let $Y_k^{\eta}$ be the number of communities of size at least $k$ but at most $(1+\eta)k$.
Random variables $Y_k$ and $Y_k^{\eta}$ satisfy properties (a) and (b) in Lemma~\ref{lem:communities}.
\end{itemize}
\end{corollary}

\begin{proof}
Recall that community sizes are generated independently following power-law distribution $\Pc(\beta, s, \tau)$ until the sum of their sizes reaches $n$. By Lemma~\ref{lem:communities} (see~(\ref{eq:sum_communities})), \wep\ after generating
$$
\ell_{\ell} = \ell_{\ell}(n) = (1 - a \, (\log n)^{-1} ) \, \hat{c} \, n^{1-\tau (2-\beta)}
$$
communities, the sum is still below $n$ provided that $a$ is a large enough constant. On the other hand, after generating 
$$
\ell_{u} = \ell_{u}(n) = (1 + a \, (\log n)^{-1} ) \, \hat{c} \, n^{1-\tau (2-\beta)}
$$
communities, \wep\ the sum exceeds $n$, again, provided that $a$ is large enough. It follows that \wep\ $\ell_{\ell} \le \ell \le \ell_u$. This finishes part (a) of the corollary.

Recall that once the desired number of nodes is reached, either the last community is ``trimmed'' or at most $s-1 = O(1)$ communities increase their sizes. To see that part~(b) holds, note that the bounds for $Y_k$ and $Y_k^{\eta}$ implied by Lemma~\ref{lem:communities} are asymptotically the same, regardless whether the lemma is applied for $\ell_{\ell}$ or for $\ell_u$. Moreover, these bounds are of order at least $(\log n)^4$ and so affecting $O(1)$ communities at the end of the process will not affect the final bounds. The proof of the corollary is finished.
\end{proof}

Now we are ready to get back to Lemma~\ref{lem:communities}. The proof of the lemma is a straightforward adaptation of the one of Lemma~\ref{lem:degree_distribution}.

\begin{proof}[Proof of Lemma~\ref{lem:communities}] 
We will call a random variable $X_i$ \emph{light} if $X_i \le K$, where
$$
K := \min \left( n^{\tau - (2\tau - 1)/\beta} (\log n)^{-4/\beta}, n^{\tau} \right);
$$
otherwise, $X_i$ will be called \emph{heavy}. Note that if $\tau < 1/2$, then $K = n^{\tau}$ so all variables are light. On the other hand, if $\tau \ge 1/2$, then $K = n^{\tau - (2\tau - 1)/\beta} (\log n)^{-4/\beta} \ll n^{\tau}$ and so we will have a mixture of light and heavy variables. 
% We need to make a distinction because for a given value of $k$, $s \le k \le K$, there are plenty of light variables that are equal to $k$, that is, $Y_k > 0$. Unfortunately, for a given value of $k$, $K < k \le n^{\tau}$, it might happen (in fact, it is quite often the case) that there will be no heavy variable that is equal to $k$, that is, $Y_k = 0$. In order to solve this problem, we will need to group heavy variables together and consider them in batches, that is, consider $Y_k^{\delta}$ instead of $Y_k$. 

We start with light variables. Let us fix $k \in \N$ such that $s \le k \le K$. Using~(\ref{eq:pk_large}), we get that 
$$
\E [ Y_k ] = \ell p_k = \Theta( n^{1-\tau (2-\beta)} k^{-\beta} ) 
= \Omega( n^{1-\tau (2-\beta)} K^{-\beta} )
= \Omega( (\log n)^4 ).
$$
After applying Chernoff's bound~(\ref{chern}) with $\eps = (\log n)^{-1}$ we get that \wep\ 
$$
Y_k = (1+\bigo( (\log n)^{-1} )) \E [ Y_k ] = (1+\bigo( (\log n)^{-1} )) \ell p_k.
$$
The other two equalities in~(\ref{eq:light_nodes_comm}) follow immediately from~(\ref{eq:pk_small}) and~(\ref{eq:pk_large}).
This finishes part~(a) and the first half of part~(b).

We move now to heavy variables. Assume that $\tau \ge 1/2$ and let us fix $k \in \N$ such that $K \le k \le n^{\tau}/(1+\eta)$. Let us also fix
$$
\eta = \eta(k) = n^{\tau(2-\beta)-1} (\log n)^4 k^{\beta - 1}.
$$
Recall that random variable $Y_k^{\eta}$ counts a \emph{batch} of variables $X_i$ that are at least $k$ but at most $(1+\eta)k$. Note that $\eta(k)$ is an increasing function of $k$. Moreover, $\eta(K) = 1/K$ so only a few values are considered for $k=K$ (note that $\eta(K)K = 1$). On the other extreme, $\eta(n^{\tau}) = n^{-(1-\tau)} (\log n)^4$ which is still $\bigo( (\log n)^{-1} ) = o(1)$.  It implies that all variables in the batch are equal to $(1+\bigo(\eta)) k = (1+\bigo( (\log n)^{-1} )) k$. As before, using~(\ref{eq:pk_large}) we get that 
$$
\E [ Y_k^{\eta} ] = \ell \sum_{i=k}^{(1+\eta)k} p_k = \Theta \left( n^{1-\tau (2-\beta)} (\eta k) k^{-\beta} \right) = \Omega( (\log n)^4 ).
$$
After applying Chernoff's bound~(\ref{chern}) with $\eps = (\log n)^{-1}$ we get that \wep\ 
$$
Y_k^{\eta} = (1+\bigo( (\log n)^{-1} )) \E [ Y_k^{\eta} ] = (1+\bigo( (\log n)^{-1} )) \ \ell \eta k p_k.
$$
This finishes the second half of part~(b).

It remains to show~(\ref{eq:sum_communities}). Since we aim for a statement that holds \wep, we may assume that parts~(a) and~(b) hold. First, let us note that
$$
\sum_{i=1}^{\ell} X_i = \sum_{k=s}^{n^{\tau}} k Y_k.
$$
Light variables can be delt with immediately: for any $k \in \N$ such that $s \le k \le K$ we have 
$$
k Y_k = (1+\bigo( (\log n)^{-1} )) \ k \ell p_k.
$$
For heavy variables we again use the fact that all variables in one batch have asymptotically the same value. For any $k \in \N$ such that $K \le k \le n^{\tau}/(1+\eta)$ we have
\begin{eqnarray*}
\sum_{i=k}^{(1+\eta)k} i Y_i &=& (1+\bigo(\eta)) k \sum_{i=k}^{(1+\eta)k} Y_i = (1+\bigo( (\log n)^{-1} )) \ k Y_k^{\eta} \\
&=& (1+\bigo( (\log n)^{-1} )) \ k \ell \eta k p_k = (1+\bigo( (\log n)^{-1} )) \sum_{i=k}^{(1+\eta)k} i \ell p_i.
\end{eqnarray*}
It follows that
$$
\sum_{i=1}^{\ell} X_i = \sum_{k=s}^{n^{\tau}} k Y_k = (1+\bigo( (\log n)^{-1} )) \sum_{k=s}^{n^{\tau}}  k \ell p_k.
$$
Now, note that the contribution from the first $\log n$ terms is negligible: 
\begin{eqnarray*}
\sum_{k=s}^{\log n}  k \ell p_k &=& \bigo \left( n^{1-\tau (2-\beta)} \sum_{k=s}^{\log n} k^{1-\beta} \right) = \bigo \left( n^{1-\tau (2-\beta)} \int_{s}^{\log n} x^{1-\beta} dx \right) \\
&=& \bigo \left( n^{1-\tau (2-\beta)} (\log n)^{2-\beta} \right) = \bigo \left( n / \log n \right).
\end{eqnarray*}
On the other hand, 
\begin{eqnarray*}
\sum_{k=\log n}^{n^{\tau}} k \ell p_k &=& (1+\bigo( (\log n)^{-1} )) \ c n^{1-\tau (2-\beta)} (\beta-1) s^{\beta-1} \sum_{k=\log n}^{n^{\tau}} k^{1-\beta} \\
&=& (1+\bigo( (\log n)^{-1} )) \ c n^{1-\tau (2-\beta)} (\beta-1) s^{\beta-1} \int_{\log n}^{n^{\tau}} x^{1-\beta} dx + \bigo(1) \\
&=& (1+\bigo( (\log n)^{-1} )) \ c n^{1-\tau (2-\beta)} (\beta-1) s^{\beta-1} \frac { n^{\tau (2-\beta)} - (\log n)^{2-\beta} } {2-\beta} \\
&=& (1+\bigo( (\log n)^{-1} )) \ c n s^{\beta-1} \ \frac { \beta-1 } {2-\beta}.
\end{eqnarray*}
This finishes the proof of~(\ref{eq:sum_communities}), and so the theorem holds.
\end{proof}

\subsubsection*{Simulation Corner}

We generated community sizes for two \ABCD\ graphs $\Ac$ on $n=1{,}000$ and, respectively, $n=1{,}000{,}000$ nodes. In both cases, we used parameters $\beta = 1.5$, $s = 50$, and $\tau = 3/4$ (that is, $S=n^{3/4}$). On Figure~\ref{fig:comm} we plot the complement of the cumulative distribution (that is, the fraction of communities that consist of at least $K$ nodes) and compare it with asymptotic, theoretical predictions, namely, function
$$
\sum_{k=K}^{S} p_k = \frac { \int_{K}^{S+1} x^{-\beta} dx }{ \int_{s}^{S+1} x^{-\beta} dx } = \frac { K^{1-\beta} - (S+1)^{1-\beta} } { s^{1-\beta} - (S+1)^{1-\beta} }.
$$
Simulation results are averaged over 30 independent runs in view of the small number of communities compared to $n$. Note that since $s=50$ is much larger than $\delta=5$, in Figure~\ref{fig:degree} we do not observe a large difference between theoretical and simulated curves (the largest discrepancy between the continuous ($q_k$) and the discrete ($r_k$) distributions is observed for small values of $k$). 

\begin{figure}[ht]
     \centering
     \includegraphics[width=0.48\textwidth]{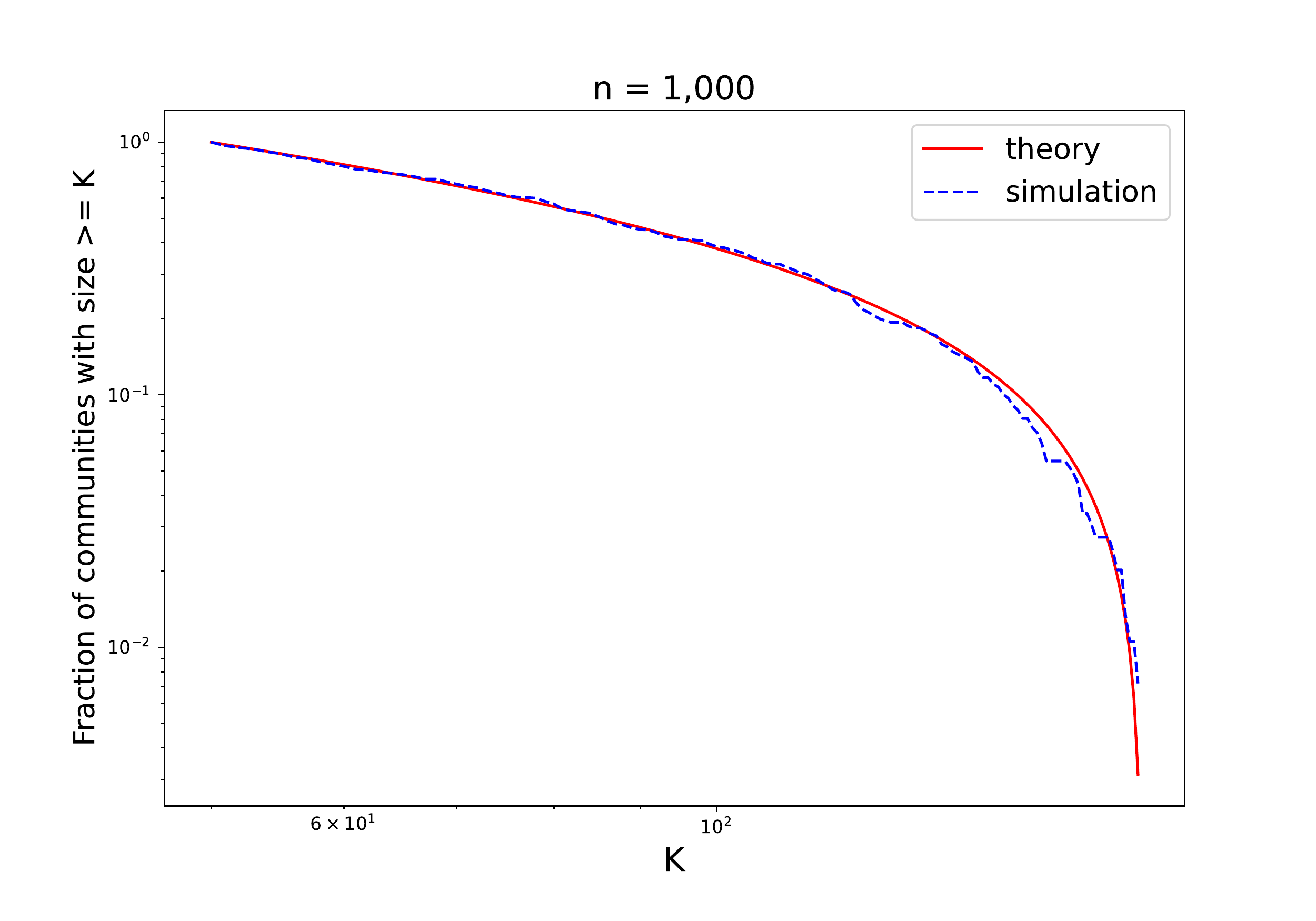}
         \hspace{.1cm}
     \includegraphics[width=0.48\textwidth]{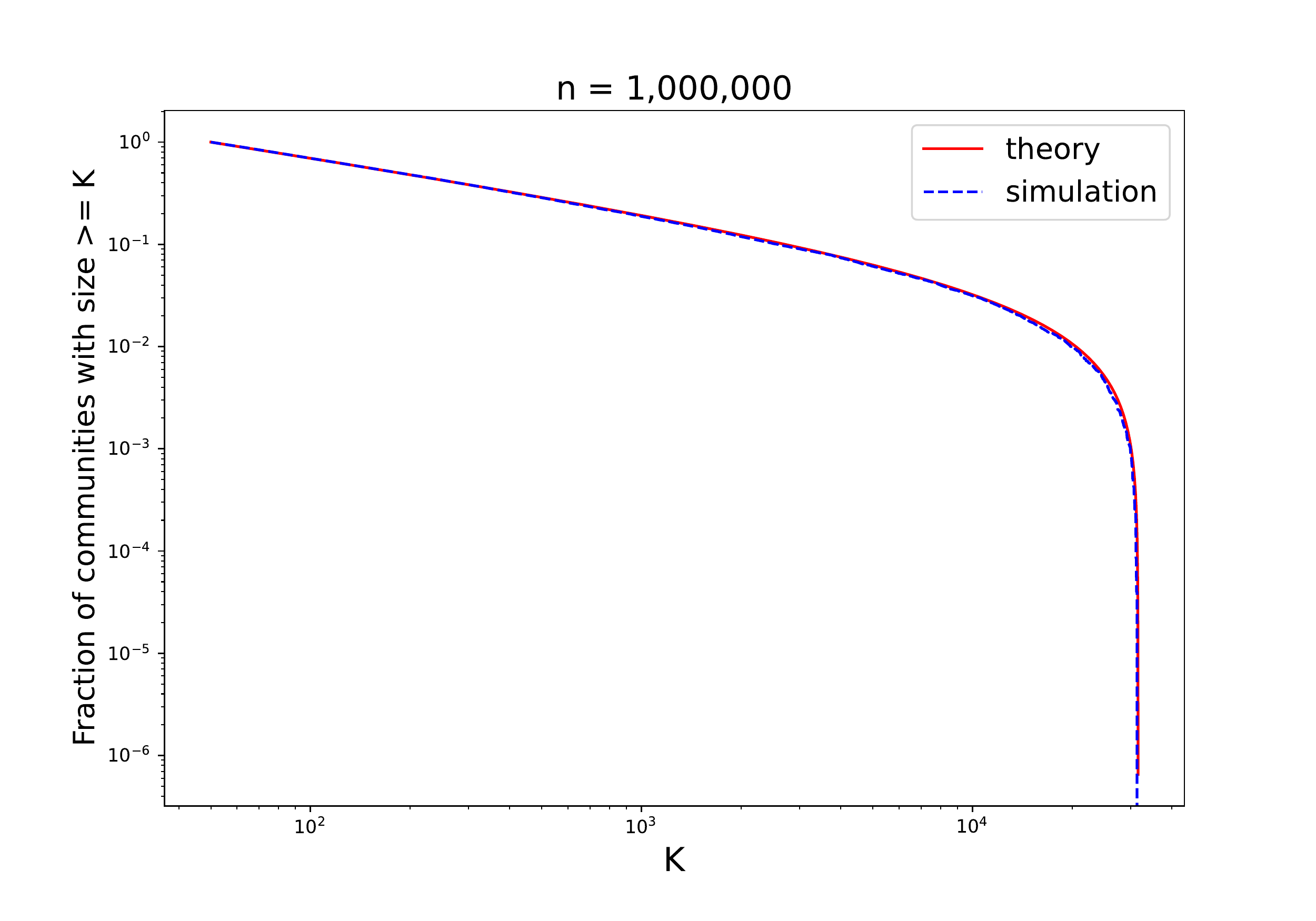}
     \caption{Complement of cumulative community size distribution for small ($n=1{,}000$; left plot) and large ($n=1{,}000{,}000$; right plot) graphs. Parameter used: $\beta = 1.5$, $s = 50$, and $\tau = 3/4$.}
     \label{fig:comm}
 \end{figure}

To see how well Corollary~\ref{cor:community_sizes} predicts the number of communities in practice, for each value of $n = 1000 \cdot 2^i$, $i \in \{0, 1, \ldots,15\}$, we independently generated 100 graphs with the same parameters as in the above experiment and three other sets of parameters. On Figure~\ref{fig:comm_ratio} we present the average number of communities and its standard deviation. In our theoretical results, we do not pay attention to signs of error terms as all of them tend to zero as $n \to \infty$. However, based on the simulation, it seems that the error term is negative for $\beta < 1.5$ and positive for $\beta > 1.5$ but, as expected, it diminishes to zero asymptotically. Additionally, note that the standard deviation of the normalized number of communities is significant even for large graphs, as opposed to standard deviation of the normalized volume presented in Figure~\ref{fig:degree_ratio}.

\begin{figure}[ht]
     \centering
     \includegraphics[width=0.48\textwidth]{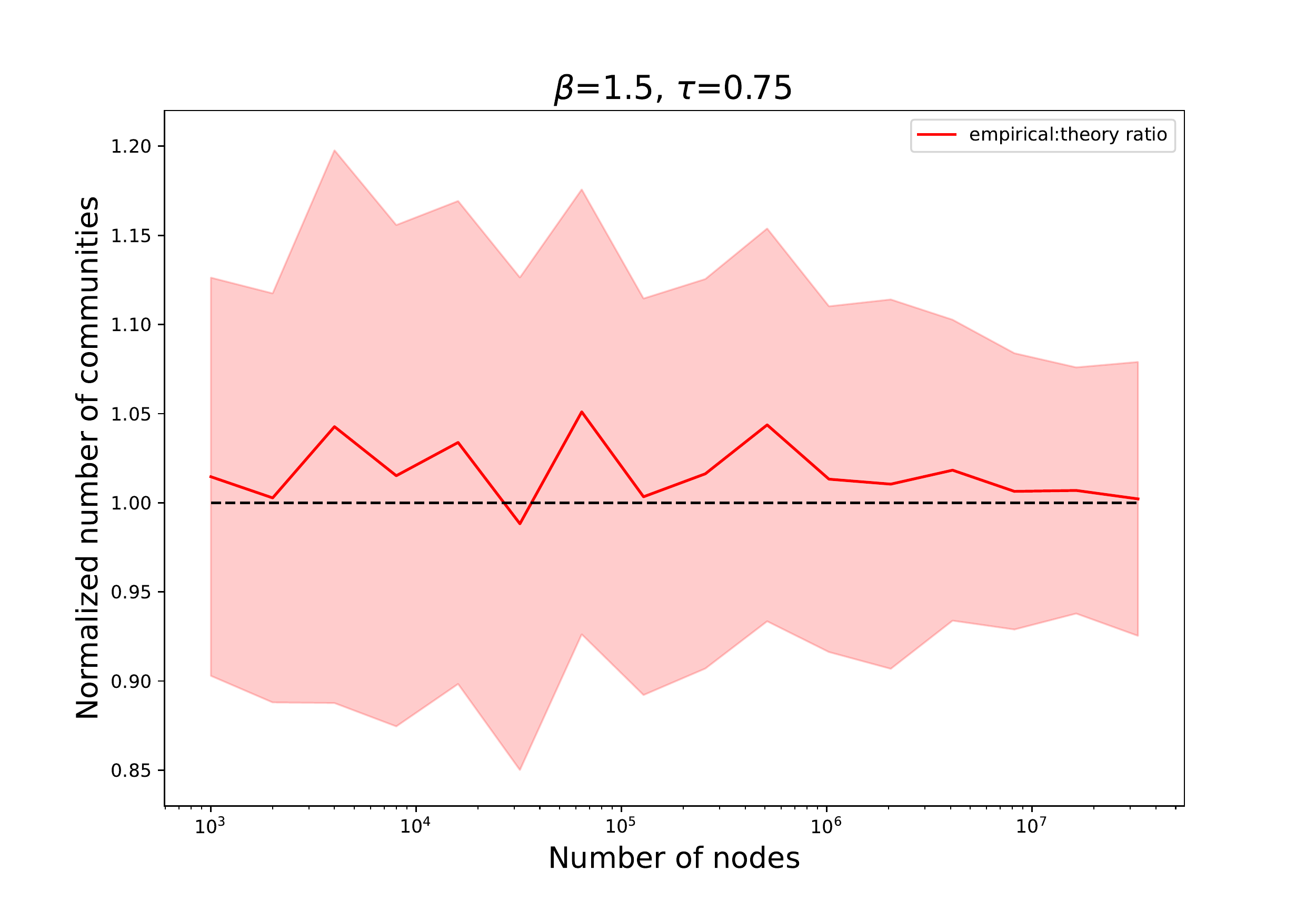}
     \hspace{.1cm}
     \includegraphics[width=0.48\textwidth]{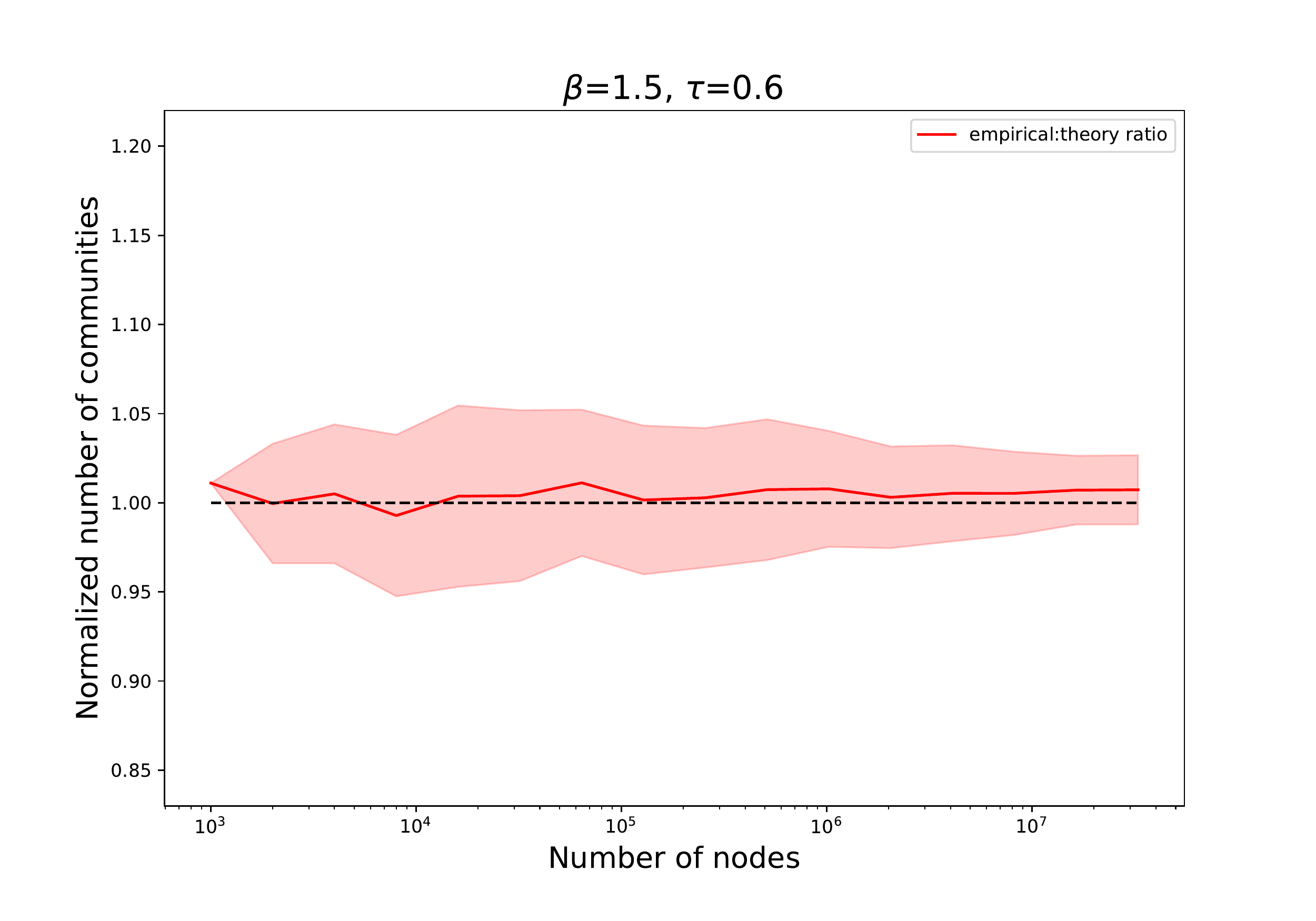}
     \vspace{.1cm}
     \includegraphics[width=0.48\textwidth]{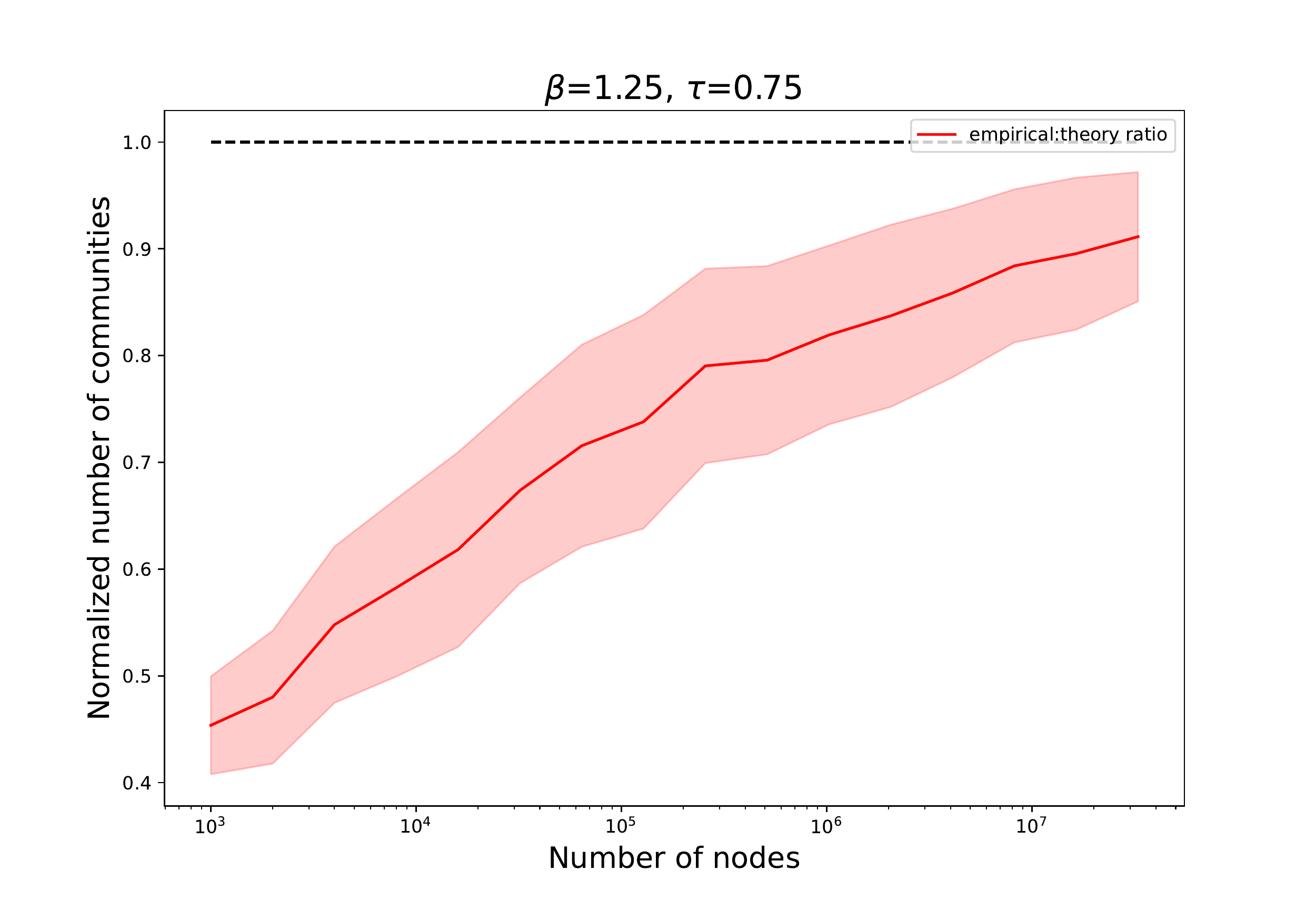}
     \hspace{.1cm}
     \includegraphics[width=0.48\textwidth]{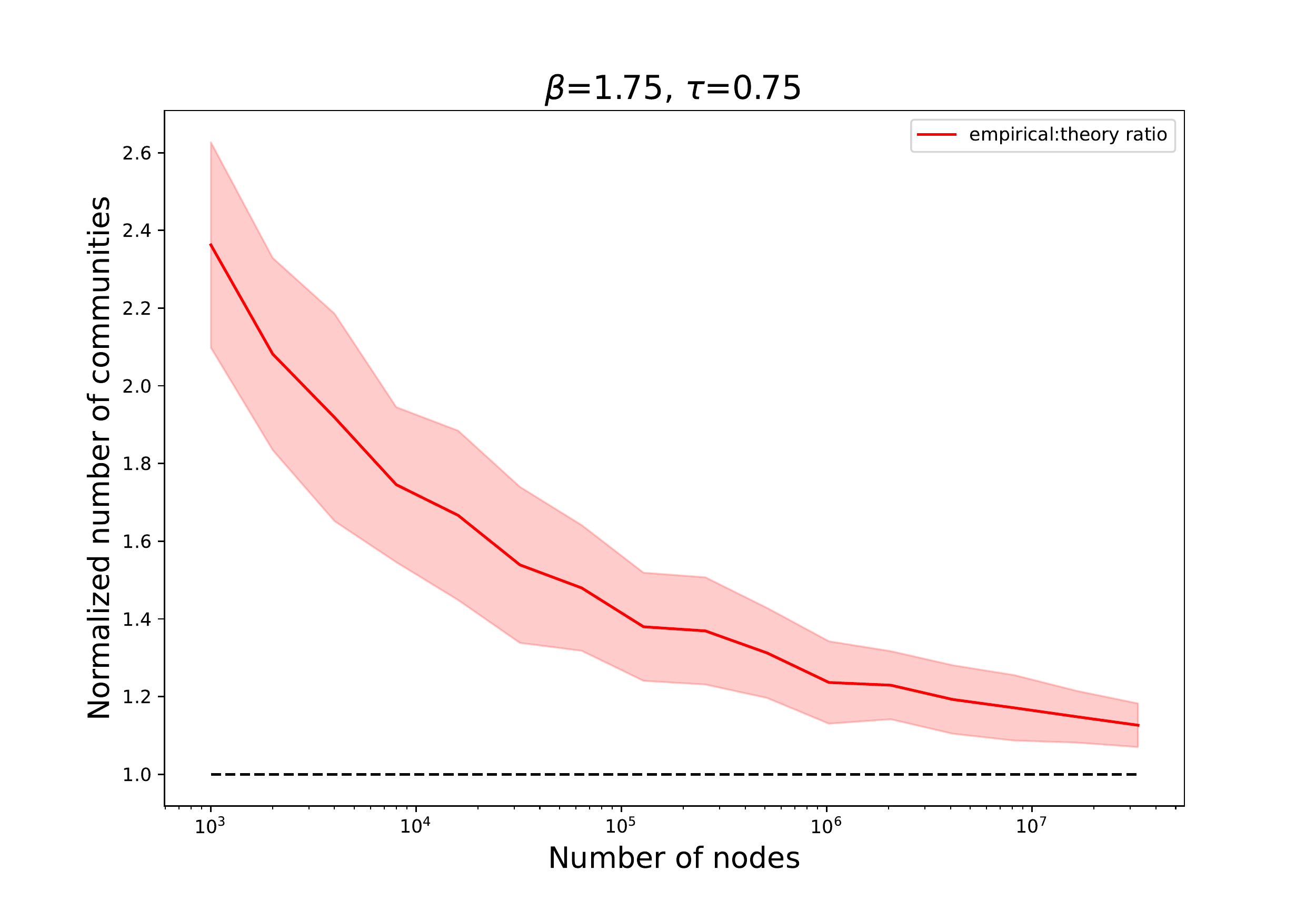}
     \caption{The average number of communities for 100 independently generated graphs; shaded area represents the standard deviation. Both quantities were normalized by the theoretical prediction. The dashed line at 1 corresponds to a perfect prediction. All plots are with $s=50$. The top two plots are shown with the same range on the y axis for comparison, i.e. decreasing $\tau$ lowers the variability. In the bottom plots, we see that convergence is slower for values of $\beta$ away from 1.5.}
     \label{fig:comm_ratio}
 \end{figure}
 
%%%%%%%%%%%%%%%%%%%%%%%%%%%%%%%%%%%%%%%%%%%%%%%%%%%%%%%%%%%
\subsection{Assigning Nodes into Communities and Distribution of Weights}\label{sec:assigning}

Recall that at this point of the process, the degree distribution $(w_1 \ge w_2 \ge \ldots \ge w_n)$ and the distribution of community sizes $(c_1 \ge c_2 \ge \ldots \ge c_\ell)$ are already fixed. In order to assign nodes to communities we will use the following easy and natural algorithm. We consider nodes, one by one, starting from $w_1$ (high degree node) and finishing with $w_n$ (low degree node). Recall that node $i$ of degree $w_i$ has to be assigned to a community of size $c_j$ so that inequality~(\ref{eq:admissible}) holds. We assign node $w_i$ randomly to one of the communities that have size larger than $\lceil (1-\xi\phi) w_i \rceil$ and still have some “available spots”. We do it with probability proportional to the number of available spots left. 

Note that it follows immediately from Corollary~\ref{cor:community_sizes} that \wep\ $\phi = 1 - \sum_{k \in [\ell]} (c_k/n)^2 = 1-o(1)$, since 
\begin{eqnarray*}
\sum_{k \in [\ell]} (c_k/n)^2 &=& \sum_{k = s}^{n^{\tau}} Y_k (k/n)^2 = \sum_{k = s}^{n^{\tau}} \bigo \left( n^{1-\tau (2-\beta)} k^{-\beta} (k/n)^2 \right) \\
&=& \bigo \left( n^{-1-\tau (2-\beta)} \sum_{k = s}^{n^{\tau}} k^{2-\beta} \right) = \bigo \left( n^{-1-\tau (2-\beta)} n^{\tau (3-\beta)} \right) = \bigo \left( n^{-(1-\tau)} \right) = o(1).
\end{eqnarray*}

In order to see that the above simple algorithm generates one of the admissible assignments uniformly at random, let $t_i$ be the number of available spots (at the beginning of the process) for node $i$, that is, the total number of nodes that belong to communities satisfying inequality~(\ref{eq:admissible}):
$$
t_i = \sum_{j \in I_i} c_j, \hspace{1cm} \text{ where } I_i = \Big\{ j \in [\ell] : w_i \le \frac {c_j - 1}{1-\xi\phi} \Big\} = \Big\{ j \in [\ell] : c_j \ge  (1-\xi\phi) w_i + 1 \Big\}.
$$
Since nodes are considered in non-increasing order of their degrees, exactly $i-1$ of these $t_i$ available spots are taken by other nodes when it is time for node $i$ to be assigned to some community. 

To see that \wep\ $t_i \ge i$, consider any node of degree $k$ such that $\delta \le k \le n^{\zeta} (\log n)^{-5/(\gamma-1)}$. It follows from Lemma~\ref{lem:degree_distribution} that when this node is considered by the algorithm, \wep\ its index $i$ satisfies the following property:
$$
i = \Theta \left( \sum_{i = \delta}^{k} n q_i \right) = \Theta \left( n \sum_{i = \delta}^{k} i^{-\gamma} \right) = \Theta \left( n k^{-(\gamma-1)} \right).
$$
On the other hand, by Corollary~\ref{cor:community_sizes}, \wep\
$$
t_i = n -  \Theta \left( \sum_{i = s}^{(1-\xi\phi) k} \ell i p_i \right) = n -  \Theta \left( n^{1-\tau (2-\beta)} \sum_{i = s}^{(1-\xi\phi) k} i^{1-\beta} \right) = n -  \Theta \left( n \left( \frac {k}{n^{\tau}} \right)^{2-\beta} \right).
$$
Both $i$ and $t_i$ are decreasing functions of $k$ but clearly $t_i \ge i$. In fact, note that \wep\ $t_1 = n(1-o(1))$ so almost all nodes are available right from the very beginning when a node of maximum degree is considered. Since \wep\ there are linearly many nodes of degree $\delta$, we process all nodes of degrees at least $\delta+1$ before there is potentially a problem.

It follows that each admissible assignment is used with probability equal to
$$
\prod_{i=1}^n \frac {1}{t_i - (i-1)},
$$
which is a fixed number that depends on the sequences $(w_1 \ge w_2 \ge \ldots \ge w_n)$ and $(c_1 \ge c_2 \ge \ldots \ge c_\ell)$ but not on the choice of admissible assignment. Since the algorithm cannot produce any non-admissible assignment, this shows a uniformity. 

\bigskip

We say that the community is \emph{large} if its size is at least $(\log n)^{8}$ and \emph{very large} if it is of size at least $n^{\zeta} (\log n)^4$; otherwise, it will be called \emph{small}. The volumes of small communities are not well concentrated around their means. For example, a community $C \subseteq V$ of constant size ($|C| = O(1)$) has to have average degree satisfying 
$$
\delta \le \frac{\vol(C)}{|C|} \le \left\lfloor \frac {|C|-1}{1-\xi\phi} \right\rfloor,
$$
but it could be equal to any of the two extreme values with probability bounded away from zero. Hence, in this situation there is no hope for determining an asymptotic value for its volume that holds \wep\ On the other hand, the volumes of very large communities are well concentrated around their means, as the next lemma shows. 

Recall that each node has its degree $w_i$ randomly split into community degree $y_i$ and background degree $z_i$. For non-leaders, we have precisely $\E[y_i] = (1-\xi) w_i$ but the leaders might require adjustment implying that $\E[y_i] = (1-\xi) w_i + O(1)$ (which is, of course, negligible from our asymptotic results point of view). For any community $C \subseteq V$, let
$$
\vol_c(C) := \sum_{i \in C} y_i
$$
be the community volume. We expect $(1-\xi)$ fraction of the total volume of each community to be assigned to community degrees, that is, $\E[ \vol_c(C) ] = (1-\xi) \vol(C) + O(1)$. Similarly to the total volumes, for very large communities we will prove that $\vol_c(C)$ is well concentrated around their expectations.

Now we are ready to state the next lemma.

\begin{lemma}\label{lem:volume_of_communities}
Let $C \subseteq V$ be any large community in $\Ac$, that is, community of size $|C| \ge (\log n)^{8}$. Then, 
\begin{eqnarray*}
\E [\vol(C)] &=& (1+\bigo( (\log n)^{-(\gamma-2)} )) \, d |C|, \text{ and } \\
\E [\vol_c(C)] &=& (1+\bigo( (\log n)^{-(\gamma-2)} )) \, (1-\xi) d |C|, \hspace{1cm} \text{ where } d := \sum_{k = \delta}^{D} \ k q_k.
\end{eqnarray*}
Moreover, the following properties hold  \wep\  
If $C$ is very large (that is, $|C| \ge n^{\zeta} (\log n)^4$), then
\begin{eqnarray*}
\vol(C) &=& (1+\bigo( (\log n)^{-(\gamma-2)} )) \, d |C|, \text{ and }\\
\vol_c(C) &=& (1+\bigo( (\log n)^{-(\gamma-2)} )) \, (1-\xi) d |C|.
\end{eqnarray*}
If $C$ is large (that is, $(\log n)^{8} \le |C| < n^{\zeta} (\log n)^4$), then $\vol(C) = \bigo( n^{\zeta} (\log n)^4 )$.\\
If $C$ is small (that is, $|C| < (\log n)^{8}$), then trivially $\vol(C) \le D |C| < n^{\zeta} (\log n)^8$.\\
Finally, $\vol(F) = \bigo( n (\log n)^{-2} )$, where $F$ is the union of small and large communities. 
\end{lemma}

Let us note that the constant $d$ used in the above lemma is the same as in Corollary~\ref{cor:volume_of_A}. Some useful and explicit bounds for $d$ are provided right after the statement of this corollary. 

\begin{proof}
Let us fix any community $C \subseteq V$ of size $z = |C| \ge z_0 := (\log n)^{8}$. If, for example, $z > n^{\zeta}$, then community $C$ has enough room even for a node of the largest degree. However, if $z$ is small then only nodes of degree at most $\left\lfloor \frac {z-1}{1-\xi\phi} \right\rfloor$ are allowed to be assigned to $C$. Fortunately, \wep\ the number of nodes of degree larger than $K_0 := \left\lfloor \frac {z-1}{1-\xi\phi} \right\rfloor \ge k_0 := \left\lfloor \frac {z_0-1}{1-\xi\phi} \right\rfloor = \Omega ( (\log n)^8 )$ is, by Lemma~\ref{lem:degree_distribution}, equal to
\begin{eqnarray*}
n' &=& \sum_{k > K_0} Y_k \le \sum_{k > k_0} Y_k = (1+\bigo( (\log n)^{-1} )) \ \sum_{k > k_0} n q_k = \bigo\left( \sum_{k > k_0} n k^{-\gamma} \right) \\
&=& \bigo\left( n k_0^{1-\gamma} \right) = \bigo\left( n \ (\log n)^{-8(\gamma-1)} \right) = \bigo\left( n \ (\log n)^{-2} \right) = o(n),
\end{eqnarray*}
and so is negligible. 

In order to track which nodes end up in community $C$, we need to start paying attention to the algorithm generating admissible assignments from time $n'+1$ when considered nodes have a chance to end up in community $C$ (because their degrees are at most $K_0$). Based on the observation above, right before this happens there are only $n' = \bigo\left( n \ (\log n)^{-2} \right)$ spots taken from communities of size at least $z+1$. On the other hand, all the remaining nodes may be assigned to any community of size more than $n^{\zeta}$, the trivial upper bound for the maximum degree. By Corollary~\ref{cor:community_sizes}, \wep\ the total number of spots available in communities of size at most $n^{\zeta}$ is equal to
\begin{eqnarray}
n'' &=& (1+\bigo( (\log n)^{-1} )) \ \sum_{k \le n^{\zeta}} \ell k p_k = \bigo\left( \sum_{k \le n^{\zeta}} n^{1-\tau (2-\beta)} k^{1-\beta} \right) \nonumber \\
&=& \bigo\left( n^{1-\tau (2-\beta)} n^{\zeta (2-\beta)} \right) = \bigo\left( n^{1- (\tau-\zeta) (2-\beta)} \right) = \bigo\left( n \ (\log n)^{-2} \right) = o(n). \label{eq:nbis}
\end{eqnarray}
It follows that once $n'+t$ nodes are assigned to communities, the number of nodes assigned to community $C$ can be stochastically 
upper bounded by $U_t$, the hypergeometric random variable with parameters $n-n'-n''=n(1- \bigo( (\log n)^{-2} ))$, $z$, $t$ and 
lower bounded by $L_t$, the hypergeometric random variable with parameters $n-n'=n(1- \bigo( (\log n)^{-2} ))$, $z$, $t$.

Suppose first that $t \le n \, (\log n)^6 / z = \bigo( n \, (\log n)^{-2} )$. Note that $\E [U_t] = zt/(n-n'-n'') = \bigo ( (\log n)^6 )$. Chernoff bound~(\ref{chern2}) applied with $u = (\log n)^6$ gives us that \wep\ $U_t = \bigo ( (\log n)^6 )$. Hence, during this phase of the algorithm, a node considered at time $t+1$ is assigned to community $C$ with probability 
\begin{eqnarray*}
\exp( - \Omega ( (\log n)^2 )) + \frac {z - \bigo (U_t) }{ n(1- \bigo( (\log n)^{-2} )) - t} &=& \frac {z - \bigo ( (\log n)^6 )}{ n(1- \bigo( (\log n)^{-2} )) - \bigo ( n \, (\log n)^{-2} )} \\
&=& \frac {z}{n} \left( 1 + \bigo( (\log n)^{-2} ) \right).
\end{eqnarray*}
Suppose now that $n \, (\log n)^6 / z \le t \le n - n \, (\log n)^{-1}$ so that $\E [U_t] \ge \E [L_t] = zt/(n-n') \ge (\log n)^6$. Chernoff bound~(\ref{chern}) applied with $\eps = (\log n)^{-2}$ gives us that \wep\ 
$$
U_t = \E [U_t] \left( 1 + \bigo ( (\log n)^{-2} ) \right) = \frac {zt}{n} \left( 1 + \bigo ( (\log n)^{-2} ) \right).
$$ 
Similarly, \wep\ $L_t = (zt/n) ( 1 + \bigo ( (\log n)^{-2} ))$. During this phase of the algorithm, a node considered at time $t+1$ is assigned to community $C$ with probability 
\begin{eqnarray*}
\exp( - \Omega ( (\log n)^2 )) + \frac {z - (zt/n) ( 1 + \bigo ( (\log n)^{-2} ))}{ n(1- \bigo( (\log n)^{-2} )) - t} &=& \frac {z ( 1 - t/n + \bigo ( (t/n) (\log n)^{-2} ))}{ (n-t) (1- \bigo( (\log n)^{-1} )) } \\
&=& \frac {z ( 1 - t/n ) (1 + \bigo ( (\log n)^{-1} ))}{ (n-t) (1- \bigo( (\log n)^{-1} )) } \\
&=& \frac {z}{n} \left( 1 + \bigo( (\log n)^{-1} ) \right),
\end{eqnarray*}
since $t \le n - n \, (\log n)^{-1}$. The end of the algorithm is unpredictable; we claim no bound for the probability of assigning a node to community $C$ when $t > n - n \, (\log n)^{-1}$. However, fortunately, the contribution from these nodes will turn out to be negligible. 

\bigskip

Let us now ``rewind'' the process and ``play'' it from the very beginning, this time paying attention what kind of nodes are assigned to community $C$. Our goal is to estimate $\vol(C)$, the volume of community $C$ of size $z$ and $\vol_c(C)$, the community volume of $C$. Since the adjustment needed for the leader of $C$ changes the value of $\vol_c(C)$ by at most 1 (and so is negligible), we may assume that each node of degree $w_i$ assigned to community $C$ is non-leader, that is, its community degree is equal to $y_i = \lfloor (1-\xi) w_i \rceil$.

Let $A$ be the set of nodes of degree larger than $K_0 = \left\lfloor \frac {z-1}{1-\xi\phi} \right\rfloor \ge k_0 = \left\lfloor \frac {z_0-1}{1-\xi\phi} \right\rfloor = \Omega ( (\log n)^8 )$ that cannot be assigned to community $C$. By Lemma~\ref{lem:degree_distribution}, \wep\ the volume of $A$ is equal to
\begin{eqnarray*}
\vol(A) &=& \sum_{k > K_0} k Y_k \le \sum_{k > k_0} k Y_k = (1+\bigo( (\log n)^{-1} )) \ \sum_{k > k_0} n k q_k = \bigo\left( \sum_{k > k_0} n k^{1-\gamma} \right) \\
&=& \bigo\left( n k_0^{2-\gamma} \right) = \bigo\left( n \ (\log n)^{-8(\gamma-2)} \right) = o(n),
\end{eqnarray*}
and so is negligible comparing to the total volume that is \wep\ linear (see Corollary~\ref{cor:volume_of_A}). Let $B$ be the set of the last $n \, (\log n)^{-1}$ nodes that we do not control, that is, we cannot predict the probability that a given node joins $C$ but we will be able to say how many of them do it. Since \wep\ these nodes have degree $\delta$ and so $\vol(B) = \bigo (n \, (\log n)^{-1})$. More importantly, the expected number of spots already taken from community $C$ when the nodes from $B$ are about to be considered is equal to
$$
\frac {z}{n} \left( 1 + \bigo( (\log n)^{-1} ) \right) \cdot (|V| - |A| - |B|) = z \left( 1 + \bigo( (\log n)^{-1} ) \right).
$$
Hence, the expected number of spots left is $\bigo( z (\log n)^{-1} )$ and so \wep\ at most $\bigo( z (\log n)^{-1} )$ are indeed left by Chernoff bound~(\ref{chern2}) applied with $u = z (\log n)^{-1} \ge (\log n)^7$. It follows that
\begin{eqnarray*}
\E [\vol(C)] &=& \frac {z}{n} \left( 1 + \bigo( (\log n)^{-1} ) \right) \cdot (\vol(V) - \vol(A) - \vol(B)) + \bigo( z (\log n)^{-1} ) \\
&=& \vol(V) \cdot \frac {z}{n} \left( 1 + \bigo( (\log n)^{-\min(1,8(\gamma-2))} ) \right) \\
&=& d z \left( 1 + \bigo( (\log n)^{-(\gamma-2)} ) \right).
\end{eqnarray*}
From this we immediately get that
$$
\E [\vol_c(C)] = (1-\xi) \E [\vol(C)] = (1-\xi) d z \left( 1 + \bigo( (\log n)^{-(\gamma-2)} ) \right).
$$

The concentration for very large communities follows from Lemma~\ref{lem:chernoff_gen}. Note that if $C$ is very large, then both $\E [\vol(C)]$ and $\E [\vol_c(C)]$ are of order $z = \Omega( n^{\zeta} (\log n)^4 )$. We get the desired bound for the failure probability by applying the lemma with $\eps = (\log n)^{-1}$ and $c = D = n^{\zeta}$. If $C$ is large, then $\E [\vol(C)] = \bigo( n^{\zeta} (\log n)^4 )$. The lemma applied with $u = n^{\zeta} (\log n)^4$ and $c = D = n^{\zeta}$ gives us that \wep\ $\vol(C) = \bigo( n^{\zeta} (\log n)^4 )$. Finally, by performing the same computation as in~(\ref{eq:nbis}) we get that \wep\ $|F| = \bigo\left( n \ (\log n)^{-2} \right) = o(n)$. It follows that $\E [ \vol(F) ] = \bigo\left( n \ (\log n)^{-2} \right)$ and so \wep\ $\vol(F) = \bigo\left( n \ (\log n)^{-2} \right)$, which finishes the proof of the lemma.
\end{proof}

\subsubsection*{Simulation Corner}

In order to see whether asymptotic predictions of volumes of large communities can be used to predict the behaviour for all communities and for relatively small values of $n$, we generated two \ABCD\ graphs $\Ac$ on $n=1{,}000$ and, respectively, $n=1{,}000{,}000$ nodes. In both cases, we used parameters $\gamma = 2.5$, $\delta = 5$, $\zeta = 1/2 < 2/3 = 1/(\gamma-1)$ (that is, $D=\sqrt{n}$), $\beta=1.5$, $s=50$, and $\tau=3/4$ (that is, $S = n^{3/4}$). On Figure~\ref{fig:comm_vol} for each generated graph $\Ac$ we plot $\ell$ points $(x_i, y_i)$, one for each community $C_i$, $i \in [\ell]$: $x_i = |C_i|$ and $y_i = \vol(C_i) / (\hat{d} |C_i|)$, where $\hat{d} = \sum_{k = \delta}^{D} \ k r_k$ is the discrete counterpart of~$d$ (see~(\ref{eq:power-law-discrete}) for a definition of $r_k$). As expected, larger communities in a larger graph on $n=1{,}000{,}000$ nodes show good concentration but small graphs on $n=1{,}000$ nodes are too small and deviate from the expectation even for the largest communities (that are still too small).

\begin{figure}[ht]
     \centering
     \includegraphics[width=0.48\textwidth]{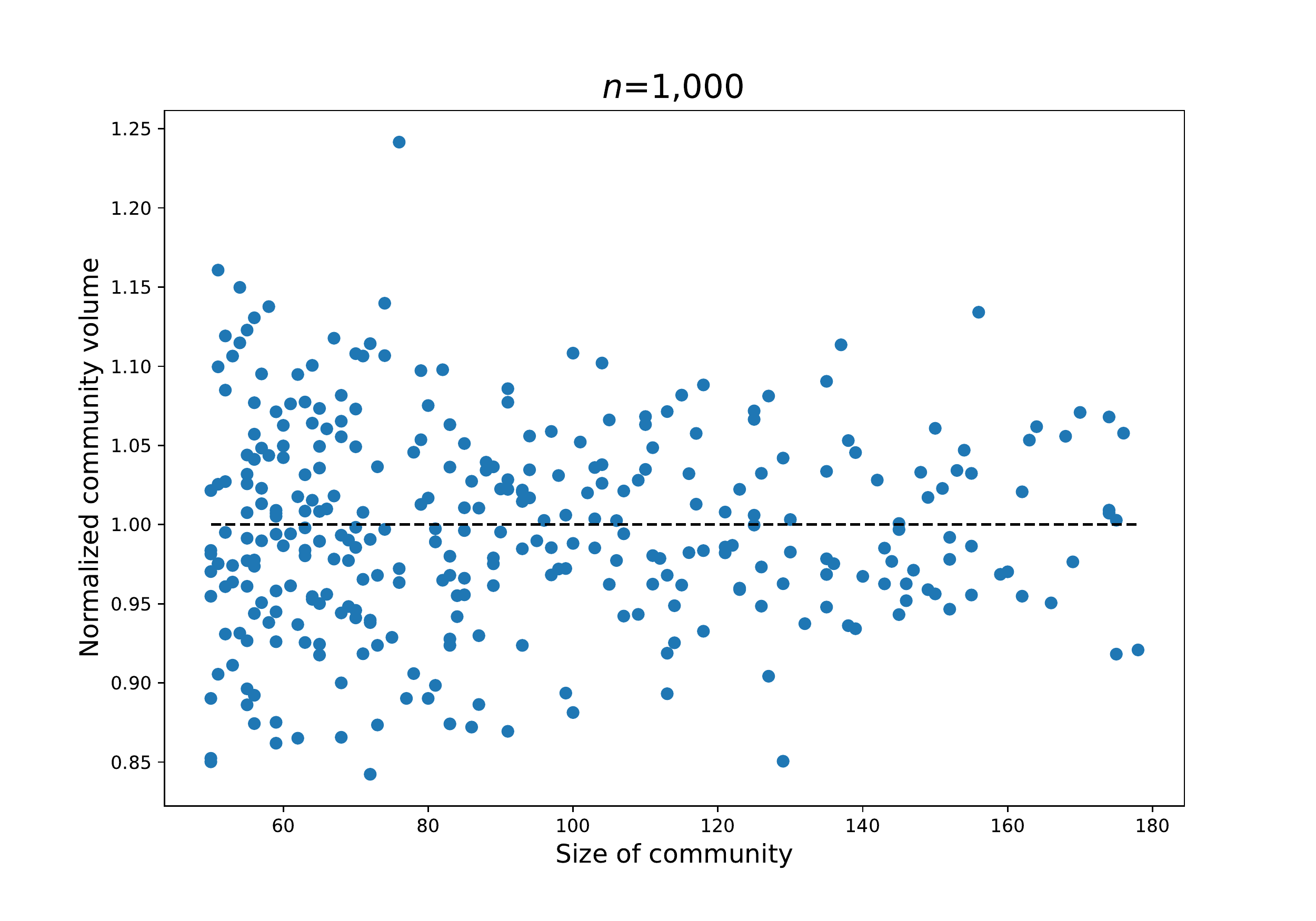}
     \hspace{.1cm}
     \includegraphics[width=0.48\textwidth]{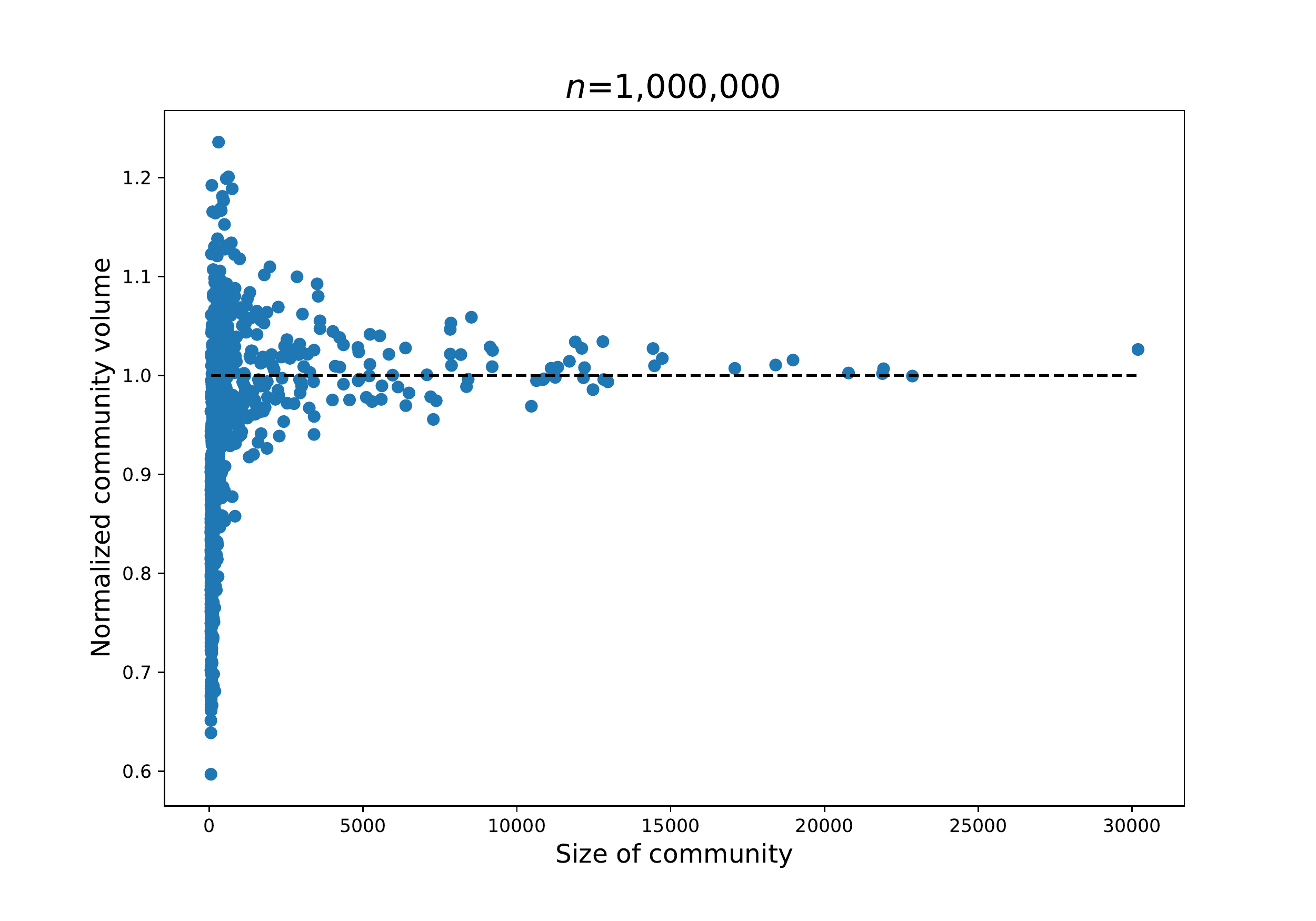}
     \caption{Volumes of communities $C_i$ scaled by $\hat{d}|C_i|$. The dashed line at 1 corresponds to a perfect prediction. 30 independent small graphs were generated ($n=1{,}000$; left plot) but only one large ($n=1{,}000{,}000$; right plot).}
     \label{fig:comm_vol}
 \end{figure}

%%%%%%%%%%%%%%%%%%%%%%%%%%%%%%%%%%%%%%%%%%%%%%%%%%%%%%%%%%%
\section{Modularity}
%%%%%%%%%%%%%%%%%%%%%%%%%%%%%%%%%%%%%%%%%%%%%%%%%%%%%%%%%%%

%%%%%%%%%%%%%%%%%%%%%%%%%%%%%%%%%%%%%%%%%%%%%%%%%%%%%%%%%%%
\subsection{Modularity of the Ground-truth Partition: $q(\C)$}

Let us start by investigating the modularity of the ground-truth partition of $\Ac$.

\begin{theorem}\label{thm:ground-truth}
Let $\C = \{C_1, C_2, \ldots, C_{\ell}\}$ be the ground-truth partition of the set of nodes of $\Ac$. Then, \wep
$$
q^*(\Ac) \ge q(\C) = (1+\bigo( (\log n)^{-(\gamma-2)} )) \, (1-\xi).
$$
\end{theorem}

\begin{proof}
Let us first estimate the degree tax. By Corollary~\ref{cor:volume_of_A}, \wep\ $\vol(V) = (1+\bigo( (\log n)^{-1} )) \ d n$, where $d = \sum_{k = \delta}^{D} \ k q_k$. By Lemma~\ref{lem:volume_of_communities}, \wep\ for each community $C_i$ we have $\vol(C_i) = \bigo( n^{\zeta} (\log n)^8 )$. It follows that \wep\
$$
\sum_{C_i \in \C} \left( \frac{\vol(C_i)}{\vol(V)} \right)^2 =  \bigo( n^{-(1-\zeta)} (\log n)^8 ) \sum_{C_i \in \C} \frac{\vol(C_i)}{\vol(V)} =  \bigo( n^{-(1-\zeta)} (\log n)^8 ) =  \bigo( (\log n)^{-2} ) = o(1),
$$
and so it is negligible. 

Let us now move to the edge contribution that is more challenging to estimate. We will use the terminology introduced in Lemma~\ref{lem:volume_of_communities}. In particular, we will call community $C_i$ very large if $|C_i| \ge n^{\zeta} (\log n)^4$, and $F$ is the union of communities that are not very large. By Lemma~\ref{lem:volume_of_communities}, \wep\ $\vol(F) = \bigo( n \, (\log n)^{-2} )$ and so the contribution (to the edge contribution) from communities that are \emph{not} very large is \wep\ equal to
$$
\sum_{C_i \in \C, |C_i| < n^{\zeta} (\log n)^4} \frac{e(C_i)}{|E|} \le  \sum_{C_i \in \C, |C_i| < n^{\zeta} (\log n)^4} \frac{\vol(C_i)/2}{\vol(V)/2} = \frac {\vol(F)}{\vol(V)} = \bigo( (\log n)^{-2} ) = o(1),
$$
and so it is negligible. It remains to concentrate on very large communities. 

Let $C_i$ be any very large community. By definition, trivially, all edges of the community graph $G_i=(C_i,E_i)$ appear in $C_i$. By Lemma~\ref{lem:volume_of_communities}, \wep\ the number of edges in $G_i$ is equal to
\begin{eqnarray*}
|E_i| = \vol_c(C_i) / 2 &=& (1+\bigo( (\log n)^{-(\gamma-2)} )) \, (1-\xi) d |C| / 2 \\
&=& (1+\bigo( (\log n)^{-(\gamma-2)} )) \, (1-\xi) \vol(C_i) / 2.
\end{eqnarray*}
To estimate the number of edges in the background graph $G_0$ that appear within $C_i$ we can use the following useful property of the pairing model. One does not need to select one pairing uniformly at random from the set of all pairings but, instead, pairs of points may be chosen sequentially. Moreover, the first point may be selected using any rule whatsoever as long as the second point is chosen uniformly at random from the set of the remaining unchosen points. By Lemma~\ref{lem:volume_of_communities}, \wep\ the number of points in the background graph that are associated with nodes in $C_i$ is equal to
$$
W = \vol(C_i) - \vol_c(C_i) = (1+\bigo( (\log n)^{-(\gamma-2)} )) \, \xi d |C_i|.
$$
In our application, we will always select the first point of the $j$th pair from the set of unchosen points associated with nodes in $C_i$ (arbitrarily). The probability that the second point is also in $C_i$ is equal to
$$
p_j = \frac {W - (j+e_j)} {\vol(V) - (2j - 1)} \le \frac{W}{\vol(V)} =: p,
$$
where $e_j$ is the number of pairs of points that already appeared within $C_i$.
(Indeed, at this point $2(j-1)+1=2j-1$ points are already chosen, $j+e_j$ of them are associated with nodes in $C_i$.)
Hence, the number of edges from the background graph that end up within community $C_i$ can be stochastically upper bounded by the binomial random variable $X \in \textrm{Bin}(w,p)$ with $\E [X] = W p = W^2 / \vol(V)$. 
If $|C_i| \ge \sqrt{n} \, (\log n)$, then $\E [X] = \Omega( (\log n)^2 )$ and we get from Chernoff bound (applied with $\eps = 1$) that \wep\ 
$$
X = \bigo(W^2 / n) = \bigo( W n^{-(1-\tau)}) = \bigo( W \, (\log n)^{-2} ) = \bigo( \vol(C_i) \, (\log n)^{-2} ),
$$
since, by definition, $W = \Theta( |C_i| ) = \bigo ( n^{\tau} )$.
On the other hand, if $|C_i| < \sqrt{n} \, (\log n)$, then $\E [X] = \bigo( (\log n)^2 )$ and we get from Chernoff bound (applied with $u = (\log n)^2$) that \wep\ $X = \bigo( (\log n)^2 )$. 

Since there are clearly at most $n / (n^{\zeta} (\log n)^4) = n^{1-\zeta} (\log n)^{-4}$ very large communities, \wep\ the total number of edges in the background graph $G_0$ that appear within some very large community $C_i$ is at most
\begin{align*}
\sum_{C_i \in \C, |C_i| \ge n^{\zeta} (\log n)^4} & \bigo( \vol(C_i) \, (\log n)^{-2} ) + n^{1-\zeta} (\log n)^{-4} \cdot \bigo( (\log n)^2 ) \\
& = \bigo( \vol(V) \, (\log n)^{-2} ) +  \bigo( n^{1-\zeta} (\log n)^{-2} ) =  \bigo( n \, (\log n)^{-2} ).
\end{align*}
The contribution (to the edge contribution) from very large communities is then \wep\ equal to
\begin{align*}
\sum_{C_i \in \C, |C_i| \ge n^{\zeta} (\log n)^4} & \frac{e(C_i)}{|E|} = \frac { \bigo( n (\log n)^{-2} ) } {|E|}  \ + \sum_{C_i \in \C, |C_i| \ge n^{\zeta} (\log n)^4} \frac{|E_i|}{|E|} \\
&= \bigo( (\log n)^{-2} ) \ + (1+\bigo( (\log n)^{-(\gamma-2)} )) \, \sum_{C_i \in \C, |C_i| \ge n^{\zeta} (\log n)^4} \frac{(1-\xi) \vol(C_i) / 2}{\vol(V) / 2} \\
&= \bigo( (\log n)^{-2} ) \ + (1+\bigo( (\log n)^{-(\gamma-2)} )) \, (1-\xi) \, \frac{ \vol(V) - \vol(F) }{\vol(V)} \\ 
&= \bigo( (\log n)^{-2} ) \ + (1+\bigo( (\log n)^{-(\gamma-2)} )) \, (1-\xi) \, (1+\bigo( (\log n)^{-2} )) \\ 
&= (1+\bigo( (\log n)^{-(\gamma-2)} )) \, (1-\xi),
\end{align*}
which is the only non-negligible contribution to the modularity function. This finishes the proof of the theorem.
\end{proof}

\subsubsection*{Simulation Corner}

In order to see how well Theorem~\ref{thm:ground-truth} predicts the modularity function $q(\C)$ in practice, for each value of $n = 1000 \cdot 2^i$, $i \in \{0, 1, \ldots,15\}$, we independently generated 30 graphs with the same parameters as in the previous experiment: $\gamma = 2.5$, $\delta = 5$, $\zeta = 1/2 < 2/3 = 1/(\gamma-1)$ (that is, $D=\sqrt{n}$), $\beta = 1.5$, $s = 50$, and $\tau = 3/4$ (that is, $S=n^{3/4}$). On Figure~\ref{fig:q_two_values_of_xi} we present the average value and the standard deviation of $q(\C)$, the modularity of the ground-truth partition, for two values of $\xi$: $\xi=0.2$ (low level of noise) and $\xi=0.7$ (high level of noise). We also present the edge contribution part of $q(\C)$, again, its average value and the standard deviation. It seems that the edge contribution for small graphs is slightly larger than the corresponding asymptotic prediction but it converges quite fast. As expected, the degree tax for small graphs is non-negligible but it converges to zero quickly. As a result, both the value of $q(\C)$ and the edge contribution tend to $1-\xi$ as $n$ grows.
 
\begin{figure}[ht]
     \centering
     \includegraphics[width=0.48\textwidth]{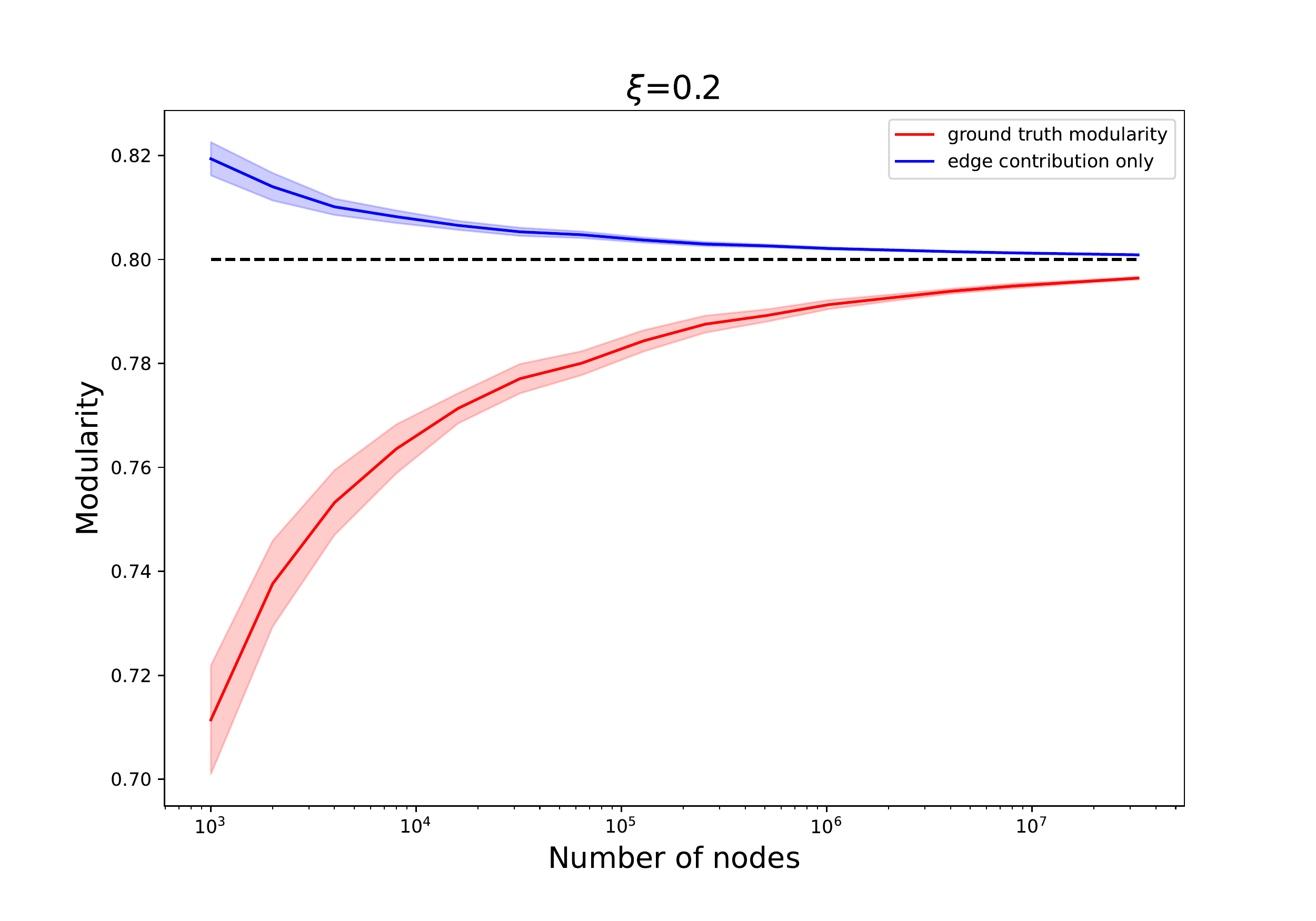}
     \hspace{.1cm}
     \includegraphics[width=0.48\textwidth]{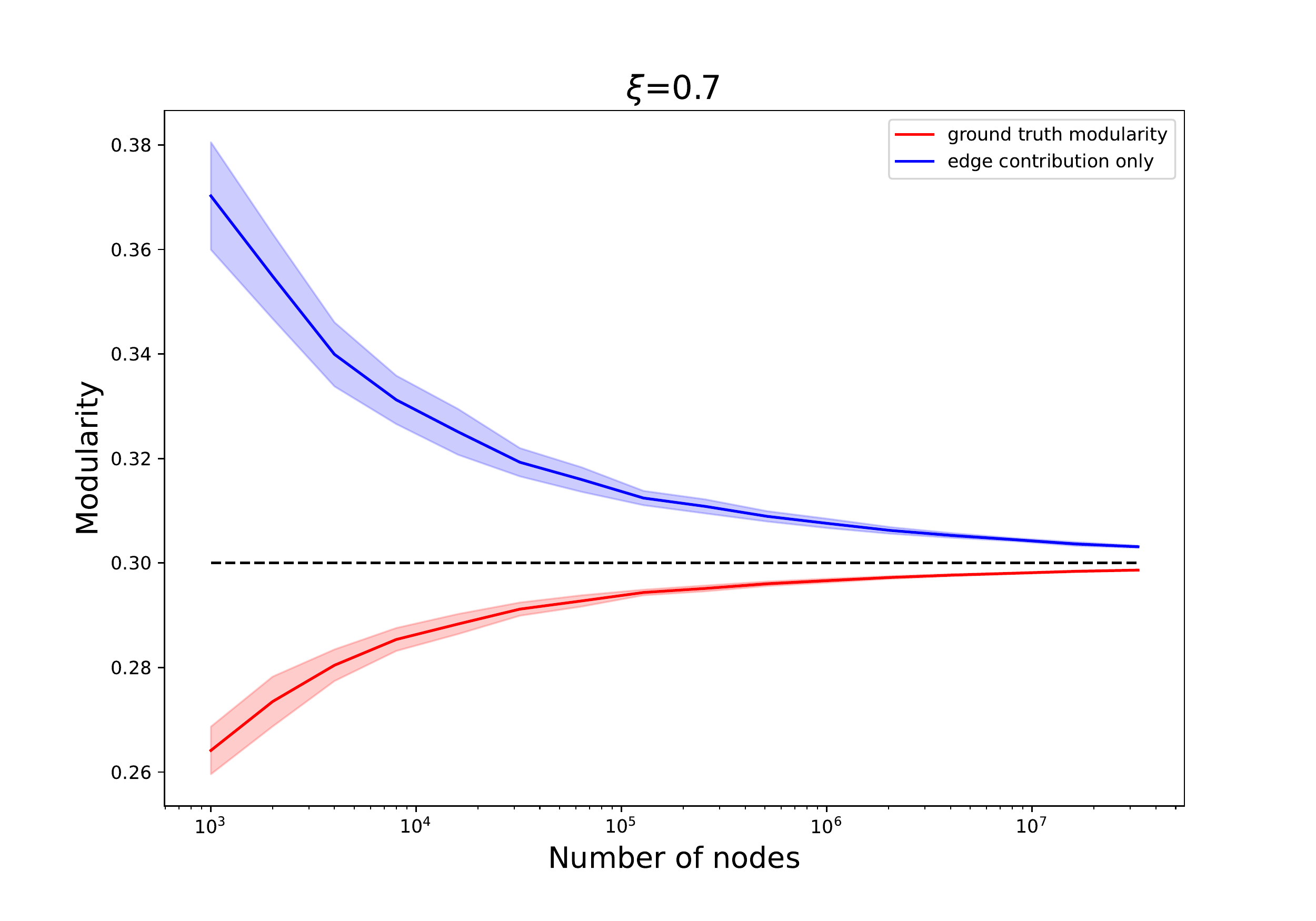}
     \caption{The modularity $q(\C)$ of the ground-truth partition (red) and the corresponding edge contribution (blue) for 30 independently generated graphs; shaded areas represent the standard deviation. The dashed line at $1-\xi$ corresponds to a perfect prediction. Parameters used: $\gamma = 2.5$, $\delta = 5$, $\zeta = 1/2$, $\beta = 1.5$, $s = 50$, and $\tau = 3/4$. Two different levels of noise are investigated.}
     \label{fig:q_two_values_of_xi}
 \end{figure}
 
Additionally, in order to investigate whether there is any difference for various levels of noise, for each value of $\xi=(0.1) i$, $i \in [9]$, we independently generated 30 graphs on $n=1{,}000$ nodes and $n=1{,}000{,}000$ nodes. The results are presented in Figure~\ref{fig:q}. As observed earlier, the edge contribution for small graphs is slightly larger than $1-\xi$, the asymptotic prediction. As expected, the difference is more visible for larger values of $\xi$ as background graph edges in noisy graphs contribute more. On the other hand, the modularity function is closer to its asymptotic prediction for more noisy graphs. Large graphs show almost perfect agreement with the asymptotic prediction. 

\begin{figure}[ht]
     \centering
     \includegraphics[width=0.48\textwidth]{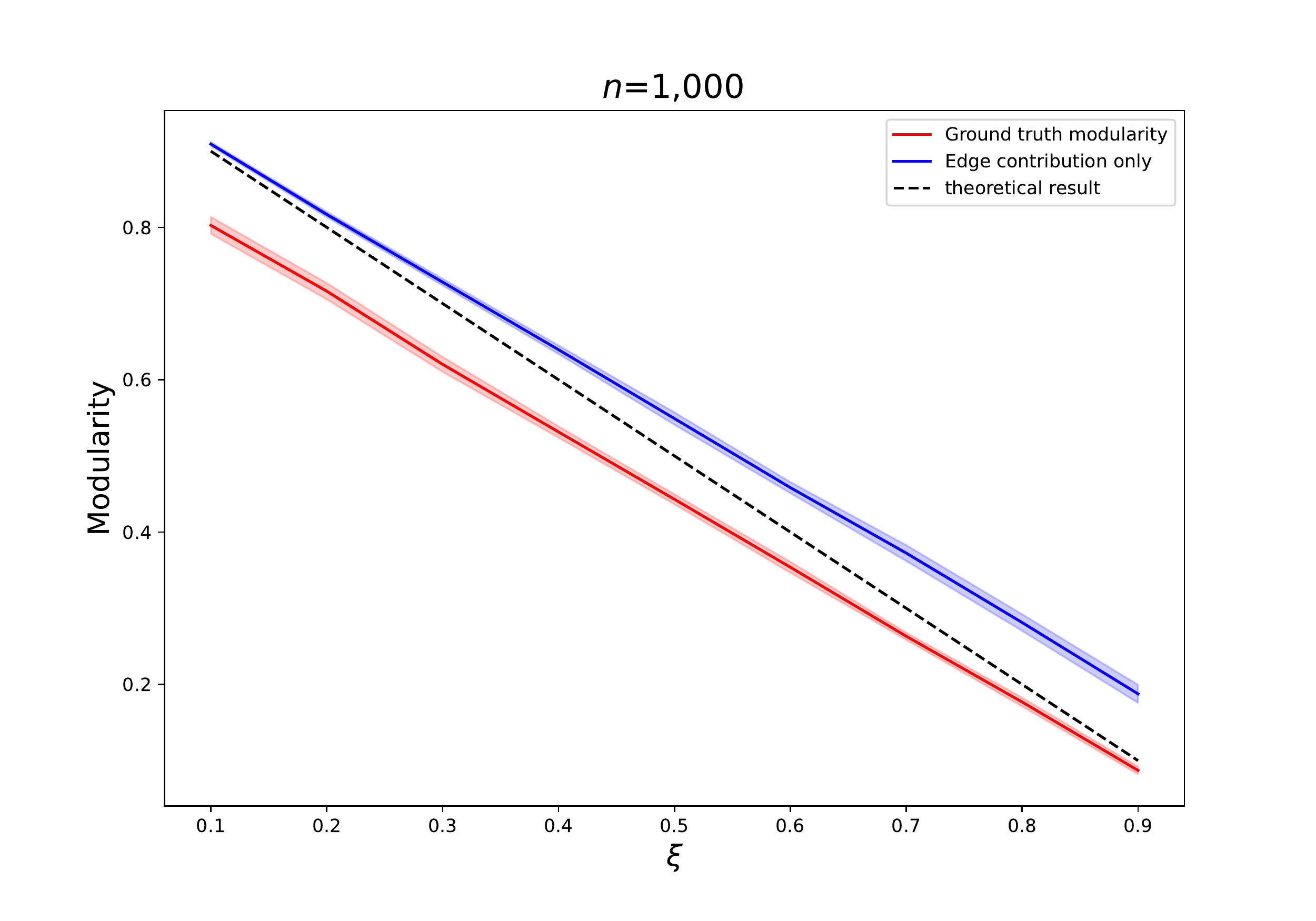}
     \hspace{.1cm}
     \includegraphics[width=0.48\textwidth]{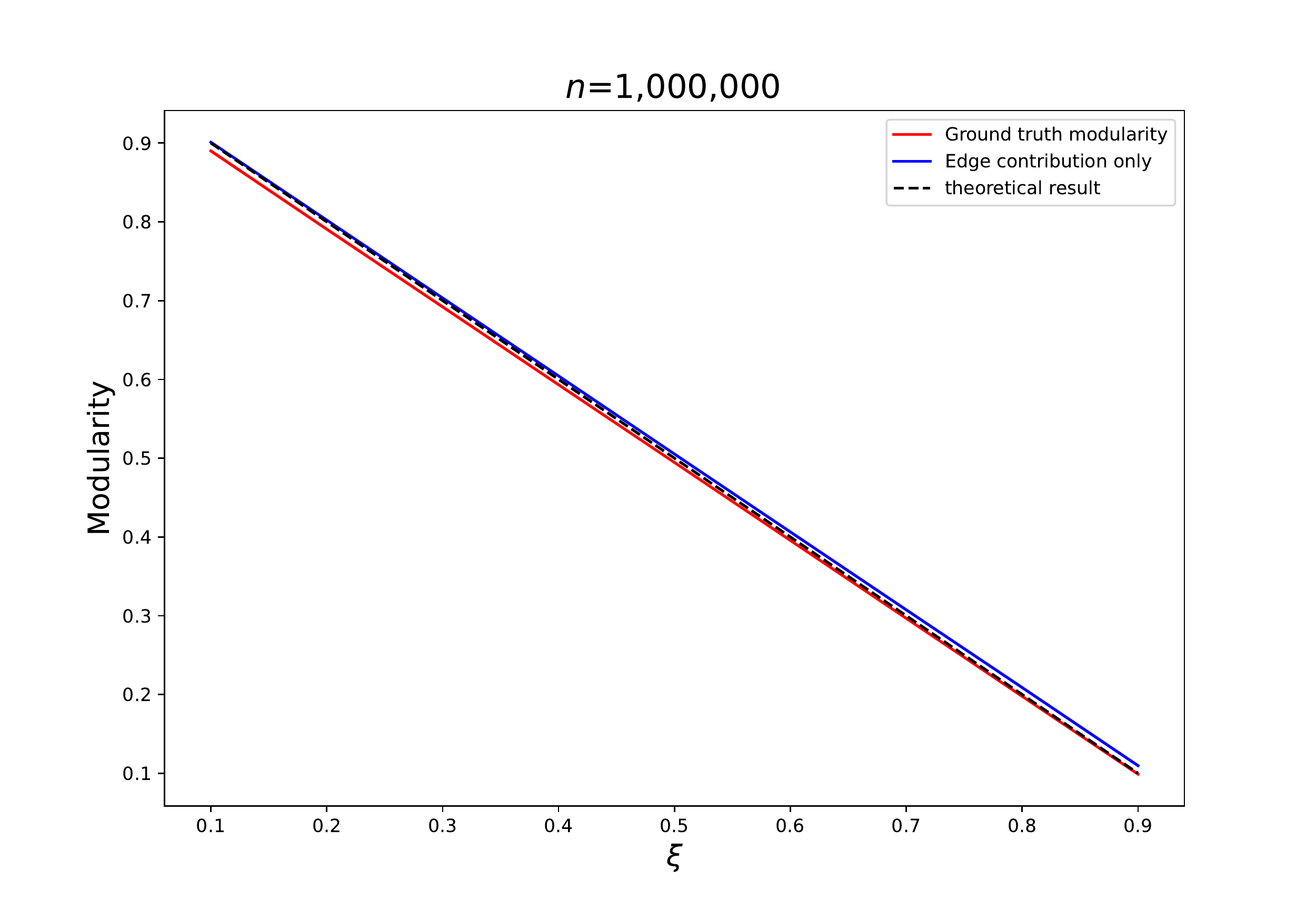}
     \caption{The modularity $q(\C)$ of the ground-truth partition (red) and the corresponding edge contribution (blue) for 30 independently generated graphs; shaded areas represent the standard deviation. The dashed line at $1-\xi$ corresponds to a perfect prediction. Parameters used: $\gamma = 2.5$, $\delta = 5$, $\zeta = 1/2$, $\beta = 1.5$, $s = 50$, and $\tau = 3/4$. Two different graph sizes are investigated.}
     \label{fig:q}
\end{figure}

\subsection{Maximum Modularity: $q^*(G)$}

As mentioned in Section~\ref{sec:related_results}, analyzing the maximum modularity $q^*(G)$ for sparse random graphs is a challenging task and typically only bounds for $q^*(G)$ are known that are far apart from each other. Since the \ABCD\ model $\Ac$ is more complex than other sparse random graphs, especially random $d$-regular graphs, there is no hope for tight bounds for the maximum modularity function but we will make some interesting observations below.

\subsubsection*{Large Level of Noise}

Let us start with investigating graphs with a large level of noise, that is, with $\xi$ close to one. For such graphs, one should focus on the background graph $G_0$ which involves all but a small fraction of edges. It turns out that $G_0$ is connected \whp, provided that its minimum degree is at least 3, or otherwise \whp\ it has a giant component. By restricting ourselves to a spanning tree of the giant component of $G_0$, we may partition the set of nodes into small parts such that each part induces a connected graph. This is not much, but for noisy graphs it yields the modularity that is larger than the modularity of the ground-truth partition. 

\begin{theorem}\label{thm:large_level_of_noise}
Let $\gamma \in (2,3)$, $\delta \in \N$, $\zeta \in (0,\frac{1}{\gamma-1}]$, and $\xi \in (0,1)$. 
\begin{itemize} 
\item[(a)]  If $\xi \delta \ge 3$, then set $\alpha = 1$.
\item[(b)]  If $\xi \delta < 3$, then there exists a universal constant $\alpha > 0$ which depends on the parameters of the model but it is always separated from 0 (that is, $\alpha$ is \emph{not} a function of $n$).
\end{itemize}
There exists a partition $\C$ of the set of nodes $V$ of $\Ac$ such that the following properties hold \whp
\begin{eqnarray*}
q^*(\Ac) \ge q(\C) &\ge& ( 1 + \bigo( n^{-(1-\zeta)/2}) ) \frac {2\alpha n}{\vol(V)} \\
&=& (1+\bigo( (\log n)^{-1} )) \ \frac {2 \alpha}{d}, \hspace{1cm} \text{ where } d = \sum_{k = \delta}^{D} \ k q_k.
\end{eqnarray*}
(Note that $q_i$ is defined in~(\ref{eq:qk_small}).)
\end{theorem}

Recall that the modularity function of the ground-truth partition is \wep\ asymptotic to $1-\xi$. The above theorem implies that if $\delta \ge 4$ and the graph has a large level of noise, namely, $\xi \ge 3/\delta$ and $\xi > 1-2/d$, then \whp\ the modularity function obtained from dissecting the spanning tree of $G_0$ is larger! The same conclusion can be derived when $\delta \le 3$ by considering $\xi$ sufficiently close to one.

\medskip

In order to prove the above theorem, we first investigate the degree distribution of the background graph $G_0$. 

\begin{lemma}\label{lem:degree_distribution_of_G0}
Let $\gamma \in (2,3)$, $\delta \in \N$, $\zeta \in (0,\frac{1}{\gamma-1}]$, and $\xi \in (0,1)$. 
For $k \in \N \cup \{0\}$, let $\hat{Y}_k$ be the random variable counting the number of nodes in the background graph $G_0$ of $\Ac$ that are of degree $k$.
For $k \in \N$ and $\eta = \eta(n)$, let $\hat{Y}_k^{\eta}$ be the random variable counting the number of nodes in the background graph $G_0$ of $\Ac$ that are of degree at least $k$ but at most $(1+\eta)k$.
Finally, for $k \in \N \cup \{0\}$, let
$$
u_k := \sum_{\substack{i \in \N \\ k-1 < \xi i < k+1 \\ \delta \le i \le \Delta}} \Big( 1-|\xi i - k| \Big) q_i ,
$$
where $q_i$ is defined in~(\ref{eq:qk_small}).

The following properties hold \wep:
\begin{itemize}
\item[(a)] If $\zeta \in (0, \frac{1}{\gamma})$, then for any $k \in \N$ such that $\lfloor \xi \delta \rfloor \le k \le \xi \lfloor n^{\zeta} \rfloor$ we have
\begin{equation}\label{eq:light_nodes_G0}
\hat{Y}_k = (1+\bigo( (\log n)^{-1} )) \ n u_k = \Theta( n q_k) = \Theta( n k^{-\gamma} ).
\end{equation}
\item[(b)] If $\zeta \in [\frac{1}{\gamma}, \frac{1}{\gamma-1})$, then for any $k \in \N$ such that $\lfloor \xi \delta \rfloor \le k \le n^{1/\gamma} (\log n)^{-4/\gamma} \ll n^{\zeta}$ random variable $\hat{Y}_k$ satisfies~(\ref{eq:light_nodes_G0}), and for any $k \in \N$ such that $n^{1/\gamma} (\log n)^{-4/\gamma} \le k \le \xi \lfloor n^{\zeta} \rfloor/(1+\eta)$ we have
\begin{equation}\label{eq:heavy_nodes_G0}
\hat{Y}_k^{\eta} = \Theta( n \eta k q_k ),
\end{equation}
where 
$$
\eta = \eta(k) = n^{-1} (\log n)^4 k^{\gamma - 1} = \bigo( (\log n)^{-1} ) = o(1).
$$
\item[(c)] If $\zeta = \frac{1}{\gamma-1}$, then for any $k \in \N$ such that $\lfloor \xi \delta \rfloor \le k \le n^{1/\gamma} (\log n)^{-4/\gamma} \ll n^{\zeta}$ random variable $\hat{Y}_k$ satisfies~(\ref{eq:light_nodes_G0}), and for any $k \in \N$ such that $n^{1/\gamma} (\log n)^{-4/\gamma} \le k \le n^{\zeta} (\log n)^{-5/(\gamma-1)} \ll n^{\zeta}$ random variable $Y_k^{\eta}$ satisfies~(\ref{eq:heavy_nodes_G0}). The number of nodes in $G_0$ of degree at least $n^{\zeta} (\log n)^{-5/(\gamma-1)}$ is equal to $\Theta( (\log n)^5 )$.
\item[(d)] The minimum and the maximum degrees of $G_0$ are respectively $\delta_0 \ge \lfloor \xi \delta \rfloor$ and $\Delta_0 \le \lceil \xi n^{\zeta} \rceil \le n^{\zeta}$. 
\item[(e)] The volume of $G_0$ satisfies
$$
\vol_{\, G_0}(V) = \sum_{k = \delta}^{D} k \hat{Y}_k = (1+\bigo( (\log n)^{-1} )) \ \xi d n, \hspace{1cm} \text{ where } d := \sum_{k = \delta}^{D} \ k q_k.
$$
\end{itemize}
\end{lemma}

\begin{proof}
As in the statement of Lemma~\ref{lem:degree_distribution}, let $Y_k$ and $Y_k^{\eta}$ be the counterparts of $\hat{Y}_k$ and $\hat{Y}_k^{\eta}$ but defined for the whole graph $\Ac$ instead of its subgraph, the background graph $G_0$. By Lemma~\ref{lem:degree_distribution}, \wep\ the degree distribution of $\Ac$ is well concentrated around its expectation. 

Recall that a node of degree $i$ in $\Ac$ has degree $\lfloor \xi i \rceil$ in $G_0$, where $\lfloor \xi i \rceil$ is a random variable equal to $\lfloor \xi i \rfloor$ or $\lceil \xi i \rceil$, and the probabilities are tuned such that the expectation is equal to $\xi i$---see definition~(\ref{eq:rounding}). We immediately get part (d) of the lemma (in fact, it holds deterministically, not only \wep). Moreover, it implies that $\E [ \vol_{\, G_0}(V) ] = \xi \, \vol_{\, \Ac}(V)$ and so, by Corollary~\ref{cor:volume_of_A} and Chernoff's bound, we get part (e). More importantly, in order for a node of degree $i$ in $\Ac$ to have a chance to be of degree $k$ in $G_0$, we must have $k-1 < \xi i < k+1$. If $k-1 < \xi i \le k$, then its degree in $G_0$ is $k$ with probability $\xi i - (k-1) = 1 - |\xi i - k|$. On the other hand, if $k < \xi i < k+1$, then the probability is equal to $(k+1)-\xi i = 1 - |\xi i - k|$. Hence, we expect $1 - |\xi i - k|$ fraction of nodes of degree $i$ in $\Ac$ to be of degree $k$ in $G_0$ and so
$$
\E \left[ \hat{Y}_k \right] = \sum_{\substack{i \in \N \\ k-1 < \xi i < k+1 \\ \delta \le i \le \Delta}} \Big( 1-|\xi i - k| \Big) \E \left[ Y_i \right].
$$
Note that for any value of $k$, $\E [ \hat{Y}_k ]$ is the linear combination of $\E \left[ Y_i \right]$ with values of $i \in \N$ in the interval of length $2/\xi$. So there are at least $\lfloor 2/\xi \rfloor \ge 2$ terms in the above sum but no more than $2/\xi + 1$. The coefficient $1-|\xi i - k|$ of at least one of them is bounded away from zero so we get that for the range considered in the statement of the lemma, we have $\E [ \hat{Y}_k ] = \Theta( \E [Y_k] )$. As the result, all expectations tend to infinity fast enough, that is, are of order at least $(\log n)^4$. Hence, Lemma~\ref{lem:degree_distribution} combined with the Chernoff's bound, implies the concentration for $\hat{Y}_k$ and $\hat{Y}_k^{\eta}$. The proof of the lemma is finished.
\end{proof}

\medskip

We will also need the following three results from~\cite{gao2020subgraph}, \cite{joos2018determine}, and~\cite{prokhorenkova2017modularity}. The first two provide sufficient conditions for the graph with a given degree sequence to be connected and, respectively, to have a giant component, that is, a component of linear order. (Both papers provide also necessary conditions but we do not need them so we only concentrate on the sufficient ones.)

For any graphical degree sequence $\textbf{w} := (w_1, w_2, \ldots, w_n)$, let $\Gc(\textbf{w})$ denote a random graph on the set of nodes $[n]$ selected uniformly at random from the family of simple graphs where node $i$ has degree $w_i$ for every $i \in [n]$. Recall that a degree sequence is graphic if there exists at least one simple graph with such degree sequence. Our background graph is a multi-graph ($\Pc(\textbf{w})$ instead of $\Gc(\textbf{w})$) but the result below applies to both families of graphs, and the condition of $\textbf{w}$ being graphical is not needed for $\Pc(\textbf{w})$. (In fact, as discussed in Subsection~\ref{sec:expanders}, in order to prove results that hold \whp\ for $\Gc(\textbf{w})$ one typically proves them for $\Pc(\textbf{w})$ and then transfers them to $\Gc(\textbf{w})$.) 

\medskip

The first lemma will be useful when the background graph $G_0$ has minimum degree at least 3.

\begin{lemma}[\cite{gao2020subgraph}]\label{lem:Jane}
Let $\textbf{w} := (w_1, w_2, \ldots, w_n)$ be any graphical sequence such that $w_1 \ge w_2 \ge \ldots \ge w_n \ge 1$. Let 
$$ \Delta = w_1 \text{ \ \ (maximum degree)}, \hspace{1cm} M=\sum_{i=1}^n w_i \text{ \ \ (volume)}, \hspace{1cm} J = \sum_{i=1}^{\Delta} w_i, $$ 
and for each $k$ let 
$$
n_k = \sum_{i=1}^n \textbf{1}_{w_i = k}  \text{ \ \ (the number of nodes of degree $k$)}.
$$ 
If $n_1 = o(\sqrt{M})$, $n_2 = o(M)$, and $J=o(M)$, then \whp\ $\Gc(\textbf{w})$ is connected. 
\end{lemma}

The next lemma will be used when the background graph has minimum degree at most 2.

\begin{lemma}[\cite{joos2018determine}]\label{lem:Guillem}
Let $\textbf{w} := (w_1, w_2, \ldots, w_n)$ be any graphical sequence such that $w_1 \le w_2 \le \ldots \le w_n $. Let 
\begin{eqnarray*}
J &=& \min \left( \left\{ j : j \in [n] \text{ and } \sum_{i=1}^j w_i (w_i-2) > 0 \right\} \cup \{ n \} \right), \\
R &=& \sum_{i=J}^n d_i, \\
M &=& \sum_{i \in [n] :  d_i \neq 2} d_i.
\end{eqnarray*}
If $R \ge \eps M$ for some $\eps > 0$, then there exists $\alpha=\alpha(\eps)$ such that \whp\ $\Gc(\textbf{w})$ has a component of order at least $\alpha n$. 
\end{lemma}

The final observation that we will need, regardless of the minimum degree of $G_0$, is that the spanning tree of the background graph (and so also a spanning tree of $\Ac$) can be decomposed into small subtrees. This can be relatively easily done by analyzing some greedy algorithm. We direct the reader to~\cite{prokhorenkova2017modularity} for more details. 

\begin{lemma}[\cite{prokhorenkova2017modularity}]\label{lem:Liuda}
Suppose that $G=(V,E)$ is a connected graph on $n$ nodes. Then there exists a partition $\C$ of $V$ such that 
$$
q^*(G) \ge q(\C) \ge \frac {2n}{\vol(V)} - 3 \sqrt{ \frac {\Delta}{\vol(V)} } - \frac {\Delta}{\vol(V)},
$$
where $\Delta$ is the maximum degree of $G$ and $\vol(V)$ is the volume of $G$. In fact, the edge contribution is at least $2n/\vol(V) - 2 \sqrt{\Delta/\vol(V)}$ and the degree tax is at most $\sqrt{\Delta/\vol(V)} + \Delta/\vol(V)$.
\end{lemma}

Now, we are ready to combine all of these observations and prove Theorem~\ref{thm:large_level_of_noise}.

\begin{proof}[Proof of Theorem~\ref{thm:large_level_of_noise}]
Suppose first that $\xi \delta \ge 3$. We will use Lemma~\ref{lem:Jane} to show that the background graph $G_0$ is \whp\ connected. By Lemma~\ref{lem:degree_distribution_of_G0}(d) we know that the minimum degree of $G_0$ is at least $\lfloor \xi \delta \rfloor \ge 3$. Let $\textbf{w} := (w_1, w_2, \ldots, w_n)$ be the degree sequence of $G_0$ such that $w_1 \ge w_2 \ge \ldots \ge w_n \ge 3$. The only property we need to check before we may apply Lemma~\ref{lem:Jane} is to verify that $J=o(M)$. Let $\omega=\omega(n)$ be any function tending to infinity as $n \to \infty$. By Lemma~\ref{lem:degree_distribution_of_G0}(a-c), \wep\ the number of nodes of degree at least $\omega$ is $\sum_{k \ge \omega} \Theta( n k^{-\gamma}) = \Theta( n \omega^{-(\gamma-1)} ) = n^{1-o(1)}$, much more than the maximum degree of $G_0$; trivially, $\Delta \le n^{\zeta}$. Hence, \wep
$$
J = \sum_{i=1}^{\Delta} w_i \le \sum_{i, w_i \ge \omega} w_i = \sum_{k \ge \omega} \Theta( k \cdot n k^{-\gamma}) = \Theta( n \omega^{-(\gamma-2)} ) = o(n) = o(M),
$$
since $M = \vol_{\, G_0}(V) = \Theta(n)$ by Lemma~\ref{lem:degree_distribution_of_G0}(e). Lemma~\ref{lem:Jane} implies that \whp\ $G_0$ is connected and so, trivially, $\Ac$ is connected \whp\ too.

Since $\Delta \le n^{\zeta}$, we get immediately from Lemma~\ref{lem:Liuda} that there exists a partition $\C$ of $V$ such that
$$
q(\C) \ge ( 1 + \bigo( n^{-(1-\zeta)/2}) ) \frac {2n}{\vol_{\, \Ac}(V)}.
$$
The conclusion follows from Corollary~\ref{cor:volume_of_A} which implies that \wep\ $\vol_{\, \Ac}(V) = (1+\bigo( (\log n)^{-1} )) \ d n$. This finishes part~(a) of the theorem.

\medskip

Suppose now that $\xi \delta < 3$. Lemma~\ref{lem:degree_distribution_of_G0}(a-c) tells us that \wep\ there are linearly many nodes of degree $2$ so, unfortunately, Lemma~\ref{lem:Jane} cannot be applied. We will use Lemma~\ref{lem:Guillem} instead. This time we need to label nodes of $G_0$ such that $w_1 \le w_2 \le \ldots \le w_n$. Note that
$$
\sum_{i=1}^j w_i (w_i-2) = - n_1 + \sum_{i=n_1+n_2+1}^j w_i (w_i-2) = - n_1 + \sum_{i=n_1+n_2+1}^j w_i^2 - 2 \sum_{i=n_1+n_2+1}^j w_i,
$$
where $n_1$ and $n_2$ denote the number of nodes of degree one and two, respectively. By Lemma~\ref{lem:degree_distribution_of_G0}(a-c), \wep\ $n_1 = (1+\bigo( (\log n)^{-1} )) \ n u_k$. On the other hand, for any constant $K \ge 3$ we have
\begin{eqnarray*}
\sum_{i \in [n]: 3 \le w_i \le K} w_i^2 &=& \sum_{k = 3}^K \Theta ( k^2 \cdot n k^{-\gamma} ) = \Theta( n K^{3-\gamma} ) \ \ \text{ and} \\
\sum_{i \in [n]: 3 \le w_i \le K} w_i &=& \sum_{k = 3}^K \Theta ( k \cdot n k^{-\gamma} ) = \Theta( n ),
\end{eqnarray*}
where the constants hidden in the $\Theta(\cdot)$ notation depend on parameters of the model but do not depend on $n$ nor $K$. Since the first sum grows with $K$ and the second one does not, $\sum_{i=1}^j w_i (w_i-2)$ becomes positive when the sum is taken over nodes of degree at most $K$ for some sufficiently large constant $K$. (Let us stress it again that $K$ is a universal constant, possibly large, but \emph{not} a function of $n$.) We get that $J \le (1-\lambda) n$ for some $\lambda>0$ and so both $R$ and $M$ are of order $n$. Since $R \ge \eps M$ for some $\eps > 0$, Lemma~\ref{lem:Guillem} can be applied and we get that \whp\ $G_0$ has a component of order at least $\alpha n$ for some constant $\alpha$ which only depends on $\eps$. (One may try to estimate $R$ and $M$ better, which would give some bound for $\eps$ but Lemma~\ref{lem:Guillem} does not provide an explicit function $\alpha=\alpha(\eps)$ anyway so there is no point to do it.) 

As before, it is enough to apply Lemma~\ref{lem:Liuda} but this time to the giant component of $G_0=(V,E)$, graph $G'_0=(V',E')$, instead of $G_0$. By Corollary~\ref{cor:volume_of_A}, \wep\ $\vol_{\, G_0}(V') = \Theta(n) = \Theta( \vol_{\, \Ac}(V) )$ (note that since $G_0$ is connected, $\vol_{\, G_0}(V') \ge |V'| \ge \alpha n$). By Lemma~\ref{lem:Liuda}, there exists a partition $\C'$ of $V'$ such that the corresponding edge contribution for the giant component in $G_0$ is at least $( 1 + \bigo( n^{-(1-\zeta)/2}) ) 2 \alpha n / \vol_{\, G_0}(V')$, that is, the number of edges within one of the parts is at least $( 1 + \bigo( n^{-(1-\zeta)/2}) ) \alpha n$. On the other hand, the degree tax is $\bigo( n^{-(1-\zeta)/2})$.

Let $\C$ be a partition of $V$ obtained by extending the partition $\C'$ of $V' \subseteq V$ by adding a trivial partition of $V \setminus V'$ consisting of parts of size one (singletons). Since partition $\C$ captures exactly the same edges within some of its parts as partition $\C'$, the corresponding edge contribution for $\Ac$ is $( 1 + \bigo( n^{-(1-\zeta)/2}) ) 2 \alpha n / \vol_{\, \Ac}(V)$. On the other hand, since \wep\ $\vol_{\, G_0}(V') = \Theta( \vol_{\, \Ac}(V) )$, the contribution to the degree tax for $\Ac$ from parts in partition $\C'$ is still $\bigo( n^{-(1-\zeta)/2})$. The contribution to the degree tax for $\Ac$ from singletons is clearly 
\begin{eqnarray*}
\sum_{i \in V \setminus V'} \left( \frac {w_i}{\vol_{\, \Ac}(V)} \right)^2 &\le& \frac {\Delta}{\vol_{\, \Ac}(V)^2}  \sum_{i \in V \setminus V'} w_i \le \frac {\Delta}{\vol_{\, \Ac}(V)^2}  \sum_{i \in V} w_i = \frac {\Delta}{\vol_{\, \Ac}(V)} \\
&=& \bigo( n^{-(1-\zeta)}) = \bigo( n^{-(1-\zeta)/2}). 
\end{eqnarray*}
We get that \whp\ 
$$
q(\C) \ge ( 1 + \bigo( n^{-(1-\zeta)/2}) ) \frac {2\alpha n}{\vol(V)},
$$
which finishes the proof of the theorem.
\end{proof}

\subsubsection*{Low Level of Noise}

This time we will investigate graphs with a low level of noise, that is, with $\xi$ close to zero. Let us fix a value of $\delta \in \N$ such that $\delta \ge 100$. For any $a \in \N$ and $b \in \N \setminus \{1,2\}$ such that $ab < \delta$, let
\begin{equation}\label{eq:c_ab}
c(a,b) := \frac {b-2\sqrt{b-1}}{2b} \, \frac {ab}{ab+b-1} - \frac {b-1}{ab+b-1} - 0.011.
\end{equation}
Let
\begin{equation}\label{eq:xi_delta}
\xi_0(\delta) := \max_{a \in \N, b \in \N \setminus \{1,2\}, ab < \delta} \min \left( 1- \frac {ab}{\delta}, \frac {c(a,b)}{4} , \frac {1}{20} \right).
\end{equation}
It is clear that $\xi_0(\delta)$ is a non-decreasing function of $\delta$. Moreover, $\xi_0 (100) \approx 0.0217$ (the maximum is achieved for $a=8$ and $b=12$), and $\xi_0 (\delta) = 1/20$ for $\delta \ge 340$. 

\medskip

Our first result says that \ABCD\ graph $\Ac$ with minimum degree $\delta \ge 100$ and $\xi \in (0,\xi_0(\delta))$ has \whp\ the maximum modularity $q^*(\Ac)$ asymptotically equal to the modularity function on the ground-truth. 

\begin{theorem}\label{thm:modularity_small_xi}
Let $\delta \in \N$ such that $\delta \ge 100$ and $0 < \xi < \xi_0(\delta)$, where $\xi_0(\delta)$ is defined in~(\ref{eq:xi_delta}). Let $\C = \{C_1, C_2, \ldots, C_{\ell}\}$ be the ground-truth partition of the set of nodes of $\Ac$. Then, \whp
$$
q^*(\Ac) \sim q(\C) \sim 1-\xi.
$$
\end{theorem}

The lower bound of $100$ for $\delta$ as well as the constants $\xi_0(\delta)$ are not tuned for the strongest result. Since the proof technique we use will not allow us to close the gap anyway, we aimed for a simple argument that works for large enough $\delta$ and relatively simple constants. Having said that, the above property is not true for $\delta=1$; that is, if $\Ac$ has minimum degree $\delta=1$, then one may find a partition of the nodes of $\Ac$ that yields larger modularity than the one associated with the ground-truth.

\begin{theorem}\label{thm:modularity_delta1}
Fix $\delta = 1$ and let $0 < \xi < 1$. Let $\C = \{C_1, C_2, \ldots, C_{\ell}\}$ be the ground-truth partition of the set of nodes of $\Ac$. Then, \wep
\begin{eqnarray*}
q^*(\Ac) &\ge& (1+\bigo( (\log n)^{-(\gamma-2)} )) \, \left( (1-\xi) + \frac {\xi q_1}{d} \Big( 2 - \frac {q_1}{d} \Big) \right) \\
&>& (1+\bigo( (\log n)^{-(\gamma-2)} )) \, (1-\xi) = q(\C),
\end{eqnarray*}
where $q_k$ is defined in~(\ref{eq:qk_small}) and $d = \sum_{k = \delta}^{D} \ k q_k$.
\end{theorem}

The rest of this subsection is devoted to proving the above theorems. Recall that our ground-truth communities as well the background graph are random multi-graphs $\Pc(\textbf{w})$ on $n'$ nodes with a given degree sequence $\textbf{w} := (w_1, w_2, \ldots, w_{n'})$ generated by the \emph{configuration model}---see Subsection~\ref{sec:creating_graphs} for more details. In order to derive some useful expansion properties of $\Pc(\textbf{w})$, we need to couple $\Pc(\textbf{w})$ with $\mathcal{P}_{n'',b}$, random $b$-regular multi-graph on $n''$ nodes, for some integer $b \ge 3$.
Note that not only $\Pc(\textbf{w})$ but also $\mathcal{P}_{n'',b}$ is generated by the configuration model. In particular, both of them are defined on the set of points that are eventually contracted to form nodes. In what follows, we establish a relationship between points used to generate both of them.

\medskip

\textbf{Coupling Between $\Pc(\textbf{w})$ and $\mathcal{P}_{n'',b}$}: Let us fix $a \in \N$ and $b \in \N \setminus \{1,2\}$ such that $ab \le (1-\xi) \delta$, where $\xi \in (0,1)$ and $\delta \ge 4$ are the parameters of the \ABCD\ model $\Ac$. Note that the minimum degree of each community graph $G_i$ is $\lfloor (1-\xi) \delta \rfloor$ so $ab$ is at most the minimum degree of each $G_i$. 

Consider any random multi-graphs $\Pc(\textbf{w})$ on $n'$ nodes with a given degree sequence $\textbf{w} := (w_1, w_2, \ldots, w_{n'})$ and the minimum degree at most $(1-\xi) \delta$. Let $W = \sum_{i=1}^{n'} w_i$ be the volume of $\Pc(\textbf{w})$ which is equal to the number of points in the corresponding configuration model, as it was defined in Subsection~\ref{sec:creating_graphs}. Recall that the model guarantees that $W$ is even. Our goal is to couple $\Pc(\textbf{w})$ and $\mathcal{P}_{n'',b}$ such that they use, if possible, the same set of points. If $W$ is divisible by $b$, then we have the right number of points to create a random $b$-regular graph on $n'' = W / b$ nodes. Suppose then that $W$ is congruent to $j$ (mod $b$) for some $j \in [b-1]$. If $b-j$ is even, then we add $b-j$ \emph{additional} points; otherwise, the number of additional points is equal to $2b-j$. Note that if $b$ is even, then $j$ has to be even as $W$ is even. So if $b-j$ is odd, then $b$ must be odd and so $2b-j$ is even. Hence, regardless whether $b-j$ is even or not, we add an even number of additional points which will be important for the argument below to hold. This time, a random $b$-regular graph will have a few more points than $\Pc(\textbf{w})$ but, as before, $W$ points will be shared by both models. 

Recall that in the configuration model, points are partitioned into buckets that are eventually contracted to form nodes.
In $\Pc(\textbf{w})$, there are $n'$ buckets, $i$th bucket (corresponding to the $i$th node) consists of $w_i$ points for a total of $W$ points. 
In $\mathcal{P}_{n'',b}$, there are $n''$ buckets (that we will call \emph{auxiliary buckets}), each consisting of $b$ points for a total of $bn'' \ge W$ points (a potential discrepancy is taken care of additional points). 
Before we expose any edges, we need to assign each of the regular $W$ points to buckets in $\Pc(\textbf{w})$ and to auxiliary buckets in $\mathcal{P}_{n'',b}$; additional points will only be assigned to auxiliary buckets in $\mathcal{P}_{n'',b}$. 

First, we arbitrarily place $n'$ buckets associated with $\Pc(\textbf{w})$ on $W$ regular points. Then, for each $i \in [n']$, the $i$th bucket in $\Pc(\textbf{w})$ consisting of $w_i$ points is arbitrarily partitioned into sets of points of size $b$ and, possibly, one set of at most $b-1$ points. By construction, there are at least $a$ sets of points of size $b$ that become buckets in $\mathcal{P}_{n'',b}$ that we will call \emph{internal auxiliary buckets}. The remaining points (including additional points if $W$ is not divisible by $b$) are arbitrarily partitioned into $b$-element \emph{external auxiliary buckets} in $\mathcal{P}_{n'',b}$. Note that this is possible as, by construction, $W$ plus the number of additional points is divisible by $b$. As mentioned earlier, the original buckets will be used to generate $\Pc(\textbf{w})$ and the auxiliary buckets will correspond to random $b$-regular graph generated by $\mathcal{P}_{n'',b}$. 

It is time to start exposing edges, that is, start randomly matching points! One of the advantages of using the pairing model is that the pairs may be chosen sequentially, at each step choosing a point using any rule (possibly randomized) that depends only on the pairs so far chosen and pairing it with a point chosen uniformly at random over the remaining (unchosen) points. We start with exposing additional points (if we have them), one by one. Note that, during this initial phase, it might happen that two additional points are paired together. (This happens with probability $\bigo( 1/n ) = o(1)$ as there are $\bigo(1)$ additional points, so we could announce that the coupling fails and deal with such rare situations in the proof of the theorem differently. However, the coupling can be established even if this rare situation happens.) Suppose that $r$ additional points are matched with not additional ones: for each $i \in [r]$, additional point $p_i$ is matched with point $q_i$. Recall that there are even number of additional points. So, even if an even number of them are matched with other additional points, the number of pairs $p_i q_i$ is even (that is, $r$ is even). Pairs $p_i q_i$ ($i \in [r]$) will be present in $\mathcal{P}_{n'',b}$ but $\Pc(\textbf{w})$ will consist of pairs $q_{2i-1} q_{2i}$ instead ($i \in [r/2]$). Note that the sequence of points $q_i$ ($i \in [r]$) is a sequence of $r$ points selected uniformly at random from the set of non-repeating sequences of points from $\Pc(\textbf{w})$ of that length. Hence, the two corresponding pairings are valid partial random pairings of $\mathcal{P}_{n'',b}$ and, respectively, $\Pc(\textbf{w})$. We continue pairing the remaining points, keeping the obtained pairs in both models. This finishes the coupling. 

\medskip

\textbf{Remark}: Let us summarize the important observations. The coupling between $\Pc(\textbf{w})$ and $\mathcal{P}_{n'',b}$, despite the fact that the number of points associated with these models might be different, has the following properties: a) $\Pc(\textbf{w})$ and $\mathcal{P}_{n'',b}$ are perfect configuration models with their respected fixed degree distributions, and b) almost all pairs of points are coupled; the only discrepancy occurs with pairs associated with nodes in $\mathcal{P}_{n'',b}$ containing additional points and their neighbours ($\bigo(1)$ nodes in $\mathcal{P}_{n'',b}$ and so also $\bigo(1)$ nodes in $\Pc(\textbf{w})$).

\medskip

Now, we are ready to prove the first theorem, Theorem~\ref{thm:modularity_small_xi}.

\begin{proof}[Proof of Theorem~\ref{thm:modularity_small_xi}]
The behaviour of some nodes might be challenging to predict. For example, nodes that belong to a community of constant size may (with positive probability) be partitioned into two parts with no edges between the parts. Such small communities might end up in different parts of the optimal partition yielding the maximum modularity $q^*(\Ac)$ of $\Ac$. Fortunately, \whp\ the number of problematic nodes will be small and so their contribution to the modularity function will be negligible.

We will use the terminology introduced in Lemma~\ref{lem:volume_of_communities}. In particular, we will call community $C_i$ very large if $|C_i| \ge n^{\zeta} (\log n)^4$, and $F$ is the union of communities that are not very large. Each community that is not very large will be called \emph{problematic} and, as a result, all nodes of $F$ will be called \emph{problematic}. By Lemma~\ref{lem:volume_of_communities}, \wep\ $\vol(F) = \bigo( n \, (\log n)^{-2} ) = o(n)$. On the other hand, by Corollary~\ref{cor:volume_of_A}, \wep\ $\vol(V) = (1+\bigo( (\log n)^{-1} )) \ d n$, where $d = \sum_{k = \delta}^{D} \ k q_k$. Hence, \wep\ $\vol(F) = o(\vol(V))$. Since we aim for a statement that hols \whp, we may assume that this property holds.

Let $\delta \ge 100$ be the minimum degree of $\Ac$. Let $a \in \N, b \in \N \setminus \{1,2\}$ be the constants that yield $\xi_0(\delta) > 0$ in~(\ref{eq:xi_delta}), and $c = c(a,b) > 0$ is defined as in~(\ref{eq:c_ab}). 
We couple community graphs $G_i$ of very large communities $C_i$ that are generated according to the model $\Pc(\textbf{w})$ with the model $\mathcal{P}_{n'',b}$. Note that, since $\xi < \xi_0(\delta) \le 1- ab/\delta$, we have $(1-\xi)\delta \ge ab$ and so the coupling may be applied. Let $\lambda_i$ be the largest absolute value of an eigenvalue other than $\lambda = b$ of the adjacency matrix of the random $b$-regular graph $\mathcal{P}_{n'',b}$ associated with $G_i$. (See Section~\ref{sec:expanders} for more on that.) If $\lambda_i > 2 \sqrt{b-1}+10^{-5}$, then we call the associated community $C_i$ and all of its nodes \emph{problematic}. By Lemma~\ref{lem:Fri}, \whp\ each very large community is not problematic so the expected volume of nodes that become problematic is $o(n)$. By the first moment method, we conclude that \whp\ the number of problematic nodes is $o(n)$. Very large communities $G_i$ that are \emph{not} problematic are coupled with $b$-regular graphs that are good expanders ($\lambda_i \le 2 \sqrt{b-1}+10^{-5}$). However, as remarked above, there could be some constant number of nodes in $G_i$ that are incident with edges in $\Pc(\textbf{w})$ that are not present in $\mathcal{P}_{n'',b}$. We will also call these nodes \emph{problematic}. Since the maximum degree is at most $n^{\zeta}$ and each very large community has at least $n^{\zeta} (\log n)^4$ nodes, the volume associated with new problematic nodes is $\bigo(n / (\log n)^4) = o(n)$. 

We call a partition $\mathbf{P}$ of the set of nodes of $\Ac$ \emph{nice} if the following two properties hold: a) all problematic communities (not very large or bad expanders) form a separate part in $\mathbf{P}$, and b) problematic nodes from non-problematic communities belong to a largest part of the partition of its own community that is induced by $\mathbf{P}$ (if there are at least two communities that are largest, then all problematic nodes belong to one of them). Let  $\mathbf{P}^*$ be a partition that yields $q^*(\Ac)$. One can move the problematic nodes around, if needed, to transform $\mathbf{P}^*$ into a nice partition $\mathbf{P}$. Since the volume of the problematic nodes is $o(\vol(V))$, both the edge contribution and the degree tax change by $o(1)$, and so we get that $q(\mathbf{P}) = q(\mathbf{P}^*)+o(1)$. Hence, it is enough to show that the maximum modularity over the family of nice partitions is $1-\xi+o(1)$.

Let $\mathbf{P}$ be a nice partition that yields the largest modularity over this family of partitions. We will show that each non-problematic community is contained in one part of $\mathbf{P}$. For a contradiction, suppose that partition $\mathbf{P}$ partitions a non-problematic community of volume $W$ into $j \ge 2$ parts $U_i$, $i \in [j]$. Let $W_i = \vol(U_i)$. Without loss of generality, we may assume that $W_1 \ge W_2 \ge \ldots \ge W_j$ and that all problematic nodes, if present, belong to $U_1$. Consider any part $U_i$ ($i \ge 2$). Since $U_1$ has a largest volume, $W_i \le W / 2$. Our goal is to estimate $E(U_i, V \setminus U_i)$, the number of edges going from $U_i$ to its complement. To that end, we will use the coupling with $b$-regular graphs. Let $U'_i$ be a subset of nodes of the coupled $b$-regular graph $\mathcal{P}_{n'',b}$ that are associated with internal auxiliary buckets, and let $W'_i$ be its volume. (In order to make it easier for the reader to distinguish sets of nodes in the original graph $\Pc(\textbf{w})$ from those in the coupled graph $\mathcal{P}_{n'',b}$, we will use $U_i$ and, respectively, $U'_i$ as well as $W_i$ and $W'_i$.) Trivially, $W'_i \le W_i$. Moreover, since at least $a$ internal auxiliary buckets of $b$ points (nodes in $\mathcal{P}_{n'',b}$) are contained in each node in $\Pc(\textbf{w})$, 
$$
W'_i \ge \frac {ab}{ab+b-1} W_i.
$$
(The extreme case corresponds to a node in $\Pc(\textbf{w})$ of degree $ab+b-1$ which is associated with $ab$ internal buckets.) Since $\mathcal{P}_{n'',b}$ is a good expander, it follows from~(\ref{eqn:bisection}) that the number of edges from $U'_i$ to its complement satisfies the following inequality:
\begin{eqnarray*}
E(U'_i, V \setminus U'_i) &\ge& (b-\lambda) \frac {|U'_i| |V \setminus U'_i|}{n''} \\
&=& (b-\lambda+o(1)) \frac { (W'_i / b) ((W - W'_i)/b) }{ W/b } \\
&\ge& \frac {b-\lambda+o(1)}{b} \frac {ab}{ab+b-1} \, W_i \, \frac { W - W'_i }{ W } \\
&\ge& \frac {b-\lambda+o(1)}{2b} \frac {ab}{ab+b-1} \, W_i.
\end{eqnarray*}
Note that not all of these edges go from $U_i$ to its complement (in $\Pc(\textbf{w})$) but clearly at most $W_i - W'_i$ of them can stay within $U_i$ as there are $W_i - W'_i$ points in $U_i$ that are assigned to external auxiliary buckets. Since $\lambda \le 2\sqrt{b-1}+10^{-5}$, we get that
\begin{eqnarray*}
E(U_i, V \setminus U_i) &\ge& E(U'_i, V \setminus U'_i) - (W_i - W'_i) \\
&\ge& \frac {b-\lambda+o(1)}{2b} \frac {ab}{ab+b-1} \, W_i - \frac {b-1}{ab+b-1} \, W_i \\
&\ge& \left( \frac {b-2\sqrt{b-1}- 10^{-5} +o(1)}{2b} \frac {ab}{ab+b-1} - \frac {b-1}{ab+b-1} \right) W_i \\
&\ge& \Big( c - 10^{-5} + 0.011 \Big) W_i,
\end{eqnarray*}
where $c = c(a,b)$ is defined in~(\ref{eq:c_ab}).

Now, we will modify the partition $\mathbf{P}$ slightly, keeping the modified partition nice, and show that it improves its modularity which will give us the desired contradiction. We need to independently consider two cases. Suppose first that $U_1$ belongs to a part $P_1$ of $\mathbf{P}$ of volume $\vol(P_1) \le \frac {c-1.1\xi}{2} \, \vol(V)$. We move all nodes in parts $U_i$, $i \ge 2$, to $P_1$. That operation puts together all edges of the non-problematic community we consider within one part but some edges of the background graph might get lost. To estimate the number of background edges that might get removed from some part of partition $\mathbf{P}$, note that the ratio between the background degree and the community degree of a node of degree $w \ge \delta \ge 100$ is at most
$$
\frac { \lceil \xi w \rceil } { \lfloor (1-\xi) w \rfloor } \le \frac { \xi w + 1 } { (1-\xi) w - 1 } \le \frac { \xi \delta + 1 } { (1-\xi) \delta - 1 } \le \frac {100 \xi+1}{94} \le 1.07 \xi + 0.0107,
$$
since $\xi \le 1/20$. Hence, the edge contribution increases by at least 
$$
\Big( c - 10^{-5} + 0.011 \Big) \sum_{i = 2}^{j} \frac {W_i}{\vol(V)} - \Big( 1.07 \xi + 0.0107 \Big) \sum_{i = 2}^{j} \frac {W_i}{\vol(V)} \ge \Big( c - 1.07 \xi \Big) \sum_{i = 2}^{j} \frac {W_i}{\vol(V)}.
$$
On the other hand, the degree tax increases by at most 
\begin{eqnarray*}
\frac { (\vol(P_1) + \sum_{i = 2}^{j} W_i)^2 } {\vol(V)^2} - \frac { \vol(P_1)^2 }{ \vol(V)^2 } &=& \frac {2 \, \vol(P_1) \sum_{i = 2}^{j} W_i } {\vol(V)^2} + \frac { \left( \sum_{i = 2}^{j} W_i \right)^2 } {\vol(V)^2} \\
&=& \frac { (2+o(1)) \, \vol(P_1) \sum_{i = 2}^{j} W_i } {\vol(V)^2} \\
&\le& \Big( c - 1.1 \xi + o(1) \Big) \sum_{i = 2}^{j} \frac {W_i}{\vol(V)}.
\end{eqnarray*}
Since 
$$
\Big( c - 1.07 \xi \Big) - \Big( c-1.1 \xi + o(1) \Big) = 0.03 \xi + o(1) > 0,
$$
the modification of $\mathbf{P}$ increases its modularity and we get a contradiction that $\mathbf{P}$ maximizes the modularity over the family of nice partitions. 

Suppose now that $U_1$ belongs to a part $P_1$ of $\mathbf{P}$ of volume $\vol(P_1) > \frac {c-1.1\xi}{2} \, \vol(V)$. This time we move all nodes in parts $U_i$, $i \ge 1$, and form an independent part. Since we also disconnect $U_1$ from its part ($P_1$), the edge contribution increases by at least 
$$
\Big( c - 1.07 \xi \Big) \sum_{i = 2}^{j} \frac {W_i}{\vol(V)} - \Big( 1.07 \xi + 0.0107 \Big) \frac {W_1}{\vol(V)}.
$$
The second term might potentially dominate the change so the edge contribution might actually decrease. Trivially, it may decrease at most by the absolute value of the second term above. Fortunately, disconnecting $U_1$ from a large part decreases the degree tax substantially. Indeed, the degree tax decreases by at least
\begin{eqnarray*}
\frac { \vol(P_1)^2 }{ \vol(V)^2 }  -  \frac { (\vol(P_1) - W_1)^2 } {\vol(V)^2} - \frac { (\sum_{i = 1}^{j} W_i)^2 } {\vol(V)^2} &=& \frac { (2+o(1)) \, \vol(P_1) W_1 } {\vol(V)^2} \\
&\ge& \Big( c - 1.1 \xi + o(1) \Big) \frac {W_1}{\vol(V)}.
\end{eqnarray*}
Note that $c \ge 4 \xi_0(\delta) \ge 4 \xi_0(100) > 0.08$. Since $\xi < \xi_0(\delta) \le c/4$, we get that
\begin{eqnarray*}
\Big( c - 1.1 \xi + o(1) \Big) - \Big( 1.07 \xi + 0.0107 \Big) &=& c - 2.17 \xi - 0.0107 + o(1)\\
&\ge& 0.4575 c - 0.0107 + o(1) \\
&=& 0.0366 - 0.0107 + o(1) > 0,
\end{eqnarray*}
and so the modification of $\mathbf{P}$ increases its modularity and we get a contradiction too. 

It follows that the nice partition $\mathbf{P}$ that yields the largest modularity over this family of partitions has each non-problematic (very large and good expander) community contained in one part of $\mathbf{P}$. It remains to show that one cannot improve the modularity function by combining some non-problematic communities together. We will use $q_{G_0}(\mathbf{P})$, $e_{G_0}(A_i)$, and $\vol_{G_0}(A_i)$ for counterparts of $q(\mathbf{P})$, $e(A_i)$, and $\vol(A_i)$ that are applied for the background graph $G_0$ instead of the entire graph $\Ac$. We get that
\begin{eqnarray*}
q(\mathbf{P}) &=& \sum_{A_i \in \mathbf{P}} \frac{e(A_i)}{|E|}  - \sum_{A_i \in \mathbf{P}} \left( \frac{\vol(A_i)}{\vol(V)} \right)^2 \\
&=& \Big( 1-\xi+o(1) \Big) + \sum_{A_i \in \mathbf{P}} \frac{e_{G_0}(A_i)}{|E|}  - \sum_{A_i \in \mathbf{P}} \left( \frac{\vol(A_i)}{\vol(V)} \right)^2.
\end{eqnarray*}
By Lemma~\ref{lem:volume_of_communities}, \wep\ the volume of each very large community $C$ satisfies $\vol_{G_0}(C) = \vol(C) - \vol_c(C) \sim \xi \, \vol(C)$. By Lemma~\ref{lem:degree_distribution_of_G0}(e), \wep\ $\vol_{G_0}(V) \sim \xi \, \vol(V)$. We get that \wep\
\begin{eqnarray*}
q(\mathbf{P}) &\le& \Big( 1-\xi+o(1) \Big) + \Big( \xi + o(1) \Big) \sum_{A_i \in \mathbf{P}} \frac{e_{G_0}(A_i)}{|E(G_0)|}  - (1+o(1)) \sum_{A_i \in \mathbf{P}} \left( \frac{\vol_{G_0}(A_i)}{\vol_{G_0}(V)} \right)^2 \\
&\le& \Big( 1-\xi+o(1) \Big) + \Big( 1+ o(1) \Big) \left( \sum_{A_i \in \mathbf{P}} \frac{e_{G_0}(A_i)}{|E(G_0)|}  - \sum_{A_i \in \mathbf{P}} \left( \frac{\vol_{G_0}(A_i)}{\vol_{G_0}(V)} \right)^2 \right) \\
&\le& \Big( 1-\xi+o(1) \Big) + \Big( 1+ o(1) \Big) q_{G_0}(\mathbf{P}).
\end{eqnarray*}

It remains to show that $q_{G_0}(\mathbf{P}) = o(1)$. We may contract each non-problematic community into a single node, since they must belong to one part. Similarly, we contract all problematic communities into a single node. Now, we may couple the entire background graph that is generated as the pairing model $\Pc(\textbf{w})$ with $\mathcal{P}_{n'',d}$ with $a=1$ and $b=d$ for an arbitrarily large integer $d$ (note that after contraction the minimum degree in $\Pc(\textbf{w})$ tends to infinity as $n \to \infty$ so $ab = d$ is certainly less than the minimum degree of $\Pc(\textbf{w})$). By Lemma~\ref{lem:mod_expander} and Lemma~\ref{lem:Fri} we get that \whp\ $q^*(\mathcal{P}_{n'',d}) = \bigo( 1/ \sqrt{d} )$ and so \whp\ $q^*(\Pc(\textbf{w})) = \bigo( 1/ \sqrt{d} )$. Since $d$ can be made arbitrarily large, we conclude that $q_{G_0}(\mathbf{P}) = o(1)$ and the proof of the theorem is finished.
\end{proof}

Now, let us move to the proof of the second theorem. Since the proof is rather straightforward and the reader is already warmed-up, we only sketch it. 

\begin{proof}[Proof of Theorem~\ref{thm:modularity_delta1}]
Suppose that $\delta = 1$. By Lemma~\ref{lem:degree_distribution}, \wep\ there are $(1+\bigo( (\log n)^{-1} )) \ q_1 n$ nodes of degree 1 in $\Ac$, where $q_k$ is defined in~(\ref{eq:qk_small}). It follows easily from Chernoff's bound that \wep\ $(1+\bigo( (\log n)^{-1} )) \ \xi q_1 n$ of them have degree 1 in the background graph (and so degree 0 in their own community graph)---see also Lemma~\ref{lem:degree_distribution_of_G0}. We will call such nodes \emph{lucky}. 

Consider the ground-truth partition $\C = \{C_1, C_2, \ldots, C_{\ell}\}$ of the set of nodes of $\Ac$. It follows from Theorem~\ref{thm:ground-truth} that \wep\ $q(\C) = (1+\bigo( (\log n)^{-(\gamma-2)} )) \, (1-\xi)$. We will modify it to improve slightly the modularity function. All lucky nodes will be moved to the community of their neighbours. Note that two lucky nodes could be neighbours of each other, that is, they may form an isolated edge in $\Ac$. Such lucky nodes will be called \emph{super-lucky} and we arbitrarily assign them to a community of one of them. In fact, edges formed by super-lucky nodes should form independent parts to increase the modularity function but the improvement would be negligible so there is no point to do it. By Lemma~\ref{lem:degree_distribution_of_G0}(e), the volume of the background graph is well concentrated around its mean and so the probability that a lucky node is super-lucky is equal to $(1+\bigo( (\log n)^{-1} )) \ (\xi q_1 n)/(\xi d n)$, where $d = \sum_{k = \delta}^{D} \ k q_k$. Hence, the expected number of super-lucky nodes is equal to 
$$
(1+\bigo( (\log n)^{-1} )) \ \xi q_1 n \, \frac {\xi q_1 n}{\xi d n} = (1+\bigo( (\log n)^{-1} )) \, \frac {\xi q_1^2}{d} \, n,
$$
and so the expected number of the associated isolated edges is half of it. On the other hand the expected number of lucky nodes that are not super-lucky is $(1+\bigo( (\log n)^{-1} )) \, \xi q_1 (1-q_1/d) \, n$. The concentration follows easily from Chernoff's bound. 
Since all edges from the community graphs still remain in some part, \wep\ this modification increases the edge contribution by
\begin{eqnarray*}
(1+\bigo( (\log n)^{-1} )) \, \Big( \frac {\xi q_1^2}{2d} + \xi q_1 \Big( 1 - \frac {q_1}{d} \Big) \Big) \, \frac {n}{|E|} &=& (1+\bigo( (\log n)^{-1} )) \, \Big( \xi q_1 \Big( 1 - \frac {q_1}{2d} \Big) \Big) \, \frac {n}{dn/2} \\
&=& (1+\bigo( (\log n)^{-1} )) \, \frac {\xi q_1}{d} \Big( 2 - \frac {q_1}{d} \Big),
\end{eqnarray*}
as $|E| = \vol(V)/2 = (1+\bigo( (\log n)^{-1} )) \, dn/2$ by Corollary~\ref{cor:volume_of_A}.

Since each part may at most double its volume (deterministically), the degree tax after the modification is of the same order as before the modification, that is, it is $\bigo( (\log n)^{-2} )$. This finishes the proof of the theorem. 
\end{proof}

\subsubsection*{Simulation Corner}

Theorem~\ref{thm:large_level_of_noise} shows that the maximum modularity is larger than the modularity of the ground-truth partition, provided that $\Ac$ is very noisy. On the other hand, Theorem~\ref{thm:modularity_small_xi} shows the opposite: the maximum modularity is asymptotic to the modularity of the ground-truth partition, provided $\Ac$ has low level of noise. For some technical reason, we assumed in that theorem that $\delta$, the minimum degree of $\Ac$, is large enough but the same property should hold for much smaller values of $\delta$. 

In order to investigate this, for each value of $\xi=(0.1) i$, $i \in [9]$, we independently generated 30 graphs on $n=1{,}000$ nodes and $n=1{,}000{,}000$ nodes and with the same parameters as in the previous experiment: $\gamma = 2.5$, $\delta = 5$, $\zeta = 1/2 < 2/3 = 1/(\gamma-1)$ (that is, $D=\sqrt{n}$), $\beta = 1.5$, $s = 50$, and $\tau = 3/4$ (that is, $S=n^{3/4}$). In order to approximate the maximum modularity, we used the ensemble clustering algorithm for unweighted graphs (ECG) which is based on the Louvain algorithm~\cite{blondel2008fast} and the concept of consensus clustering~\cite{poulin2018ensemble}, and is shown to have good stability. 

The experiments coincide with theoretical predictions, despite the fact that $\delta=5$ in the experimental graphs is much smaller than the lower bound of 100 assumed in the corresponding theorem. For small values of $\xi$, the modularity found by ECG is very close to $1-\xi$, an asymptotic prediction for the modularity of the ground-truth. On the other hand, if $\xi$ is large, then the partition found by the algorithm is of better quality than the ground-truth partition---see Figure~\ref{fig:q_ECG}.

In Table~\ref{tab:comparison}, we show comparison of the modularity of the ground-truth partition and a partition found using ECG. Additionally, we provide the AMI and ARI measures of similarities between the two partitions. If $\xi$ is close to $0$, then ground-truth and ECG partitions are similar. Based on our theoretical results, we also expect that they are close to the maximum graph modularity, $q^*$. The reason is that there is a low level of noise coming from the background graph and so ECG has no problem finding a good partition that is close to the ground-truth one. On the other hand, if $\xi$ is close to $1$, then the two partitions are dissimilar. Also, note that then the modularity of the ECG partition is much higher than the one of the ground-truth, as expected based on the theoretical results. For intermediate values of $\xi$ we observe two effects. The first observation is that AMI and ARI of the ECG partition and ground-truth partition drops sharply when $\xi$ becomes greater than $0.5$. The reason for this is that for such values of $\xi$ the volume of the background graph, which is independent from the community graphs, becomes dominant and ECG tries to recover its structure. The second observation, which is a consequence of the first one, is that for intermediate values of $\xi$ the modularity of the ground-truth is noticeably larger than for the ECG partition (for example, for $\xi=0.6$ with $n=10^3$ and for $\xi=0.7$ with $n=10^6$). The reason is that ECG already starts to get a lot of signal from the background graph while it still would be more efficient to stick to the partition closer to the ground-truth. However, because of the large level of noise already present in the graph coming from the background graph, the ECG algorithm is unable to recover it.

\begin{figure}[ht]
     \centering
     \includegraphics[width=0.48\textwidth]{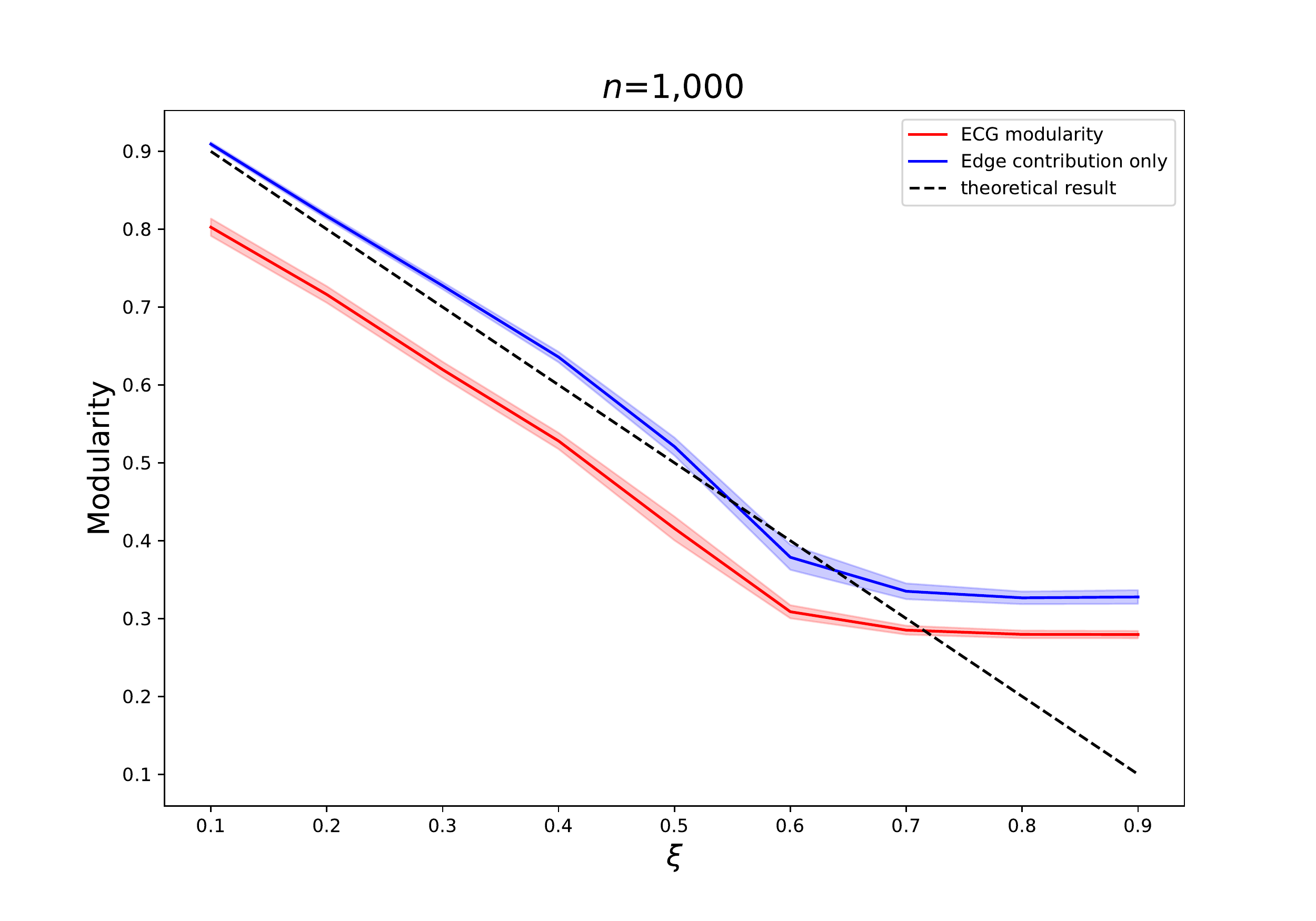}
     \hspace{.1cm}
     \includegraphics[width=0.48\textwidth]{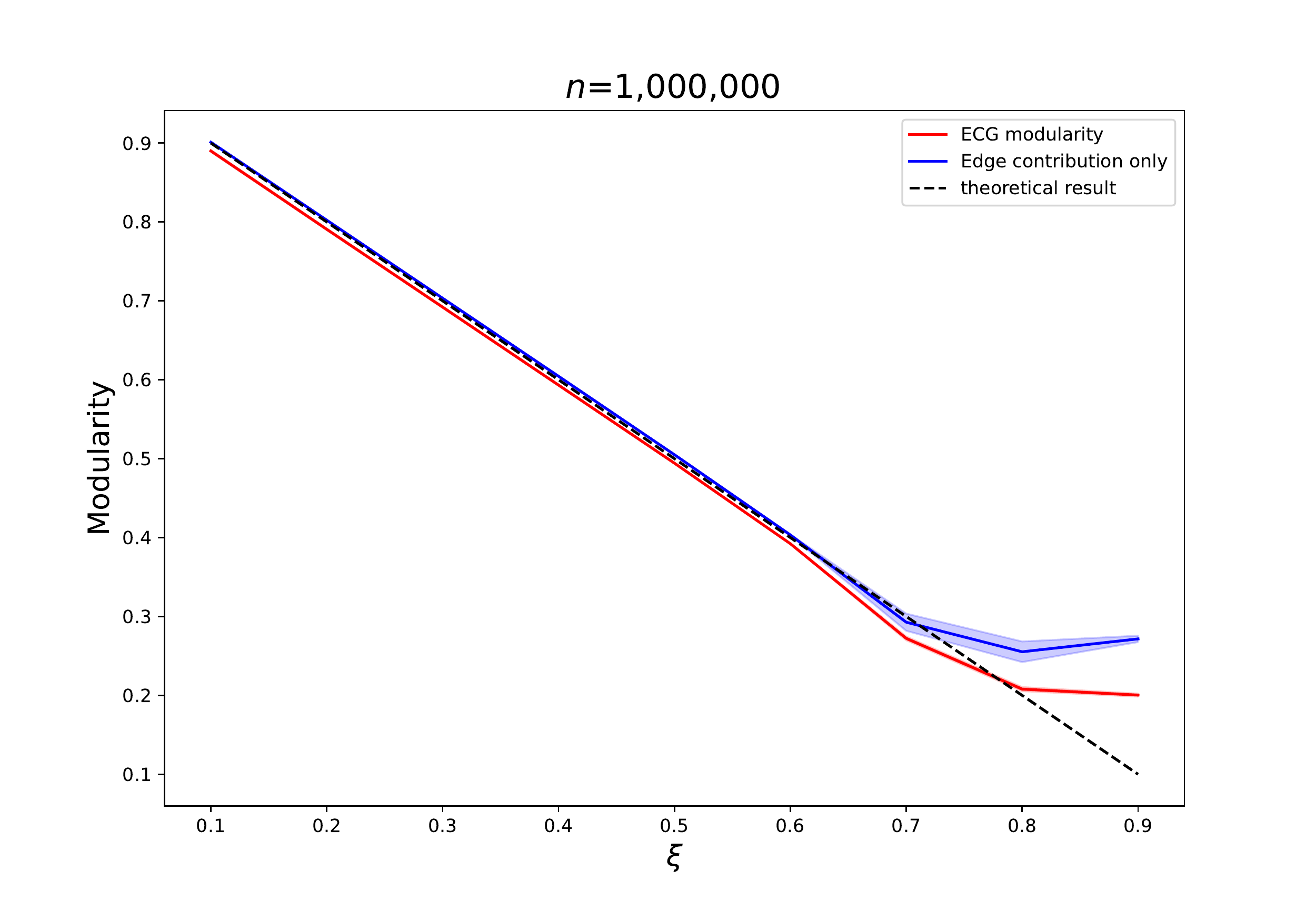} 
     \caption{The modularity $q(\C)$ obtained with ECG (red) and the corresponding edge contribution (blue) for 30 independently generated graphs; shaded areas represent the standard deviation. The dashed line at $1-\xi$ corresponds to a perfect prediction for the ground-truth. Parameters used: $\gamma = 2.5$, $\delta = 5$, $\zeta = 1/2$, $\beta = 1.5$, $s = 50$, and $\tau = 3/4$. Two different graph sizes are investigated.}
\label{fig:q_ECG}
\end{figure}

\begin{table}[]
    \centering
{    \small
    \begin{tabular}{c|c|c|c|c|c}
        $\xi$ & size ($n$) &  $q$(ground truth) & $q$(ECG) & AMI & ARI \\
        \hline \hline
         0.1 & $10^3$ & 0.802594 & 0.802594	& 1.000000 & 1.0000000 \\ 
             & $10^6$ & 0.890027 &  0.890038 &  0.999367 &  0.999968 \\ 
        \hline
         0.2 & $10^3$ & 0.716289 & 0.716190	& 0.999721 & 0.9996570 \\ 
             & $10^6$ & 0.790705 &  0.790722 &  0.998952 &  0.999943 \\ 
        \hline
         0.3 & $10^3$ & 0.620174 & 0.619678	& 0.997716 & 0.997692 \\
                      & $10^6$ & 0.692048 &  0.692070 &  0.998584 &  0.999903 \\ 
        \hline
         0.4 & $10^3$ & 0.531436 & 0.528090 & 0.974488 & 0.973746 \\
                      & $10^6$ & 0.593213 &  0.593239 &  0.997954 &  0.999763 \\ 
        \hline
        0.5 & $10^3$ & 0.443102 & 0.415710 & 0.688916 & 0.632881 \\
                      & $10^6$ & 0.494549 &  0.494295 &  0.986875 &  0.987448 \\ 
        \hline
        0.6 & $10^3$ & 0.353922 & 0.308721 & 0.199603 & 0.109946 \\
                      & $10^6$ & 0.395613 &  0.392237 &  0.889011 &  0.856030 \\ 
        \hline
        0.7 & $10^3$ & 0.263305 & 0.285068 & 0.071929 & 0.022793 \\     
                      & $10^6$ & 0.296665 &  0.272223 &  0.366945 &  0.164060 \\ 
        \hline
        0.8 & $10^3$ & 0.176988 & 0.279705 & 0.032120 & 0.008722 \\
                      & $10^6$ & 0.197707 &  0.207766 &  0.027720 &  0.001900 \\ 
        \hline
        0.9 & $10^3$ & 0.087200 & 0.279385 & 0.013933 & 0.003501 \\
                      & $10^6$ & 0.098857 &  0.200353 &  0.001541 &  0.000025 \\ 
    \end{tabular}
}    
    \caption{Comparison of the modularity of the ground-truth communities and the modularity obtained using ECG along with AMI and ARI measures between the two partitions. Data is presented for graphs with $10^3$ and $10^6$ nodes. The results are averages over 30 randomly sampled graphs. Parameters used: $\gamma = 2.5$, $\delta = 5$, $\zeta = 1/2$, $\beta = 1.5$, $s = 50$, and $\tau = 3/4$.}
    \label{tab:comparison}
\end{table}

\bibliography{ref}

\end{document}